\documentclass[11pt]{article}
\usepackage[utf8]{inputenc}
\usepackage[letterpaper,margin=1in]{geometry}
\usepackage{amsmath,amsthm,amssymb}
\usepackage{graphicx}
\usepackage[linesnumbered,ruled]{algorithm2e}
\usepackage{booktabs}
\usepackage{array}
\usepackage{microtype}
\usepackage{needspace}
\usepackage[font=small,labelfont=bf]{caption}
\usepackage{tikz}
\usetikzlibrary{arrows.meta,calc,decorations.pathreplacing,positioning,fit,backgrounds}
\definecolor{figblue}{RGB}{33,102,172}
\definecolor{figred}{RGB}{178,24,43}
\definecolor{figgray}{RGB}{115,115,115}
\tikzset{
  fig/.style={x=1cm,y=1cm,>={Stealth[length=2.1mm,width=1.6mm]},line cap=round,line join=round,font=\small},
  figpt/.style={circle,fill=black,draw=none,inner sep=0pt,minimum size=4.4pt},
  figbigpt/.style={circle,fill=black,draw=none,inner sep=0pt,minimum size=6.8pt},
  figbluept/.style={circle,fill=figblue,draw=none,inner sep=0pt,minimum size=4.4pt},
  figctr/.style={circle,draw=figblue,line width=0.9pt,fill=none,inner sep=0pt,minimum size=10pt},
  figopt/.style={circle,draw=figred,line width=0.9pt,fill=none,inner sep=0pt,minimum size=10pt},
  figaxis/.style={draw=figgray,line width=0.5pt},
  figlab/.style={font=\small,inner sep=1.2pt},
  figlabs/.style={font=\footnotesize,inner sep=1.2pt},
  figband/.style={rounded corners=2.5pt,fill=figblue!9,draw=figblue!35,line width=0.4pt},
  fighull/.style={rounded corners=5pt,fill=figblue!10,draw=figblue!45,line width=0.4pt},
}
\usepackage[hidelinks]{hyperref}

\hypersetup{pdftitle={A (3+epsilon)-Approximation for Fair k-Means and a Matching Lower Bound},pdfauthor={Kangke Cheng, Guanlin Mo, Shihong Song, Hu Ding},pdfsubject={Approximation algorithms and hardness of approximation for fair k-means},pdfkeywords={fair clustering, k-means, approximation algorithms, hardness of approximation, Frank-Wolfe method, Wasserstein barycenter}}
\newtheorem{theorem}{Theorem}
\newtheorem{lemma}[theorem]{Lemma}
\newtheorem{proposition}[theorem]{Proposition}
\newtheorem{corollary}[theorem]{Corollary}
\newtheorem{definition}[theorem]{Definition}
\newtheorem{remark}[theorem]{Remark}

\newcommand{\BlackBox}{\rule{1.5ex}{1.5ex}}
\renewenvironment{proof}[1][Proof]{\par\noindent\textbf{#1}\ }{\unskip\nobreak\hfil\penalty50\hskip2em\mbox{}\nobreak\hfill\BlackBox\par\medskip}

\newcommand{\R}{\mathbb{R}}
\newcommand{\dist}{\mathtt{dist}}
\newcommand{\Cost}{\mathtt{Cost}}
\newcommand{\Cen}{\mathtt{Cen}}
\newcommand{\CS}{\mathtt{CS}}
\newcommand{\WCS}{\mathtt{WCS}}
\newcommand{\supp}{\mathtt{supp}}
\newcommand{\diam}{\operatorname{\mathtt{diam}}}
\newcommand{\Out}{\operatorname{\mathtt{out}}}
\newcommand{\OPT}{\mathrm{OPT}}
\newcommand{\OPTfrac}{\OPT_{\mathrm{frac}}}
\newcommand{\Sopt}{S_{\mathrm{opt}}}
\newcommand{\cN}{\mathcal{N}}
\newcommand{\cT}{\mathcal{T}}
\newcommand{\cL}{\mathcal{L}}
\newcommand{\cU}{\mathcal{U}}
\newcommand{\cV}{\mathcal{V}}
\newcommand{\ind}[1]{\mathbf 1[#1]}
\newcommand{\keyterm}[1]{\textbf{#1}}

\title{A $(3+\epsilon)$-Approximation for Fair $k$-Means\\and a Matching Lower Bound}
\author{%
\begin{tabular}{c@{\hspace{4em}}c}
Kangke Cheng\textsuperscript{1,\dag} & Guanlin Mo\textsuperscript{1,\dag}\\
\normalsize\texttt{ke314159@mail.ustc.edu.cn} & \normalsize\texttt{moguanlin@mail.ustc.edu.cn}\\[2.5ex]
Shihong Song\textsuperscript{2} & Hu Ding\textsuperscript{1,*}\\
\normalsize\texttt{S.Song-29@sms.ed.ac.uk} & \normalsize\texttt{huding@ustc.edu.cn}
\end{tabular}\\[3ex]
\normalsize\textsuperscript{1}University of Science and Technology of China, Hefei, China\\
\normalsize\textsuperscript{2}School of Informatics, University of Edinburgh, Edinburgh, UK\\[1.5ex]
\small\textsuperscript{\dag}Kangke Cheng and Guanlin Mo contributed equally.\qquad
\textsuperscript{*}Corresponding author.}
\date{}

\begin{document}

\pagenumbering{roman}
\maketitle
\thispagestyle{empty}

\begin{abstract}
Fair clustering aims to prevent algorithmic bias against protected groups in high-stakes applications such as resource allocation, electoral redistricting, and credit scoring. We study Euclidean fair $k$-means, in which the proportion of each protected group in every cluster lies between prescribed lower and upper bounds. Because these constraints violate the nearest-center assignment property, center selection is coupled with point assignment. We introduce the drift-regularized relaxation—a new convex relaxation solved by a modified Frank--Wolfe method with linear-programming iterations—whose optimality condition controls the cost of converting its solution into $k$ fair clusters. Combined with any $\rho$-approximation algorithm for weighted $k$-means, the relaxation yields an exactly fair fractional solution of cost at most $2\rho+1+\epsilon$ times the fractional optimum, improving the previous ratio $1+4\rho+\epsilon$. In fixed dimension, where a $(1+\epsilon)$-approximation for weighted $k$-means exists, this improves the best known ratio from $5+O(\epsilon)$ to $3+O(\epsilon)$. If an integral assignment is required, the fractional solution can be rounded without changing the centers or increasing the cost, at a bounded additive violation of the fairness bounds. We also show that a $(3-\delta)$-approximation is NP-hard for every fixed $\delta>0$, even on the line with fractional assignments, establishing $3$ as the exact approximation threshold for fractional fair $k$-means in fixed dimension. The same algorithm also solves $k$-sparse Wasserstein barycenters, which summarize multiple probability distributions by a distribution supported on at most $k$ points. Since each input distribution forms a group that contributes equal mass to every cluster, this is a special case of fractional fair $k$-means, and the best known ratio in fixed dimension again improves from $5+O(\epsilon)$ to $3+O(\epsilon)$.
More generally, the approximation ratio $2\rho+1+\epsilon$ holds whenever every cluster must satisfy given homogeneous linear inequalities in the fractions of the points assigned to it, and $3$ remains the approximation threshold in fixed dimension. Besides fair $k$-means, such clustering problems include, for example, clustering with a bound on the ratio between two groups or on the average of a numerical attribute in every cluster.
\end{abstract}


\clearpage
\tableofcontents
\clearpage
\pagenumbering{arabic}

\section{Introduction}\label{sec-intro}

The $k$-means problem is among the most widely used models of clustering~\cite{DBLP:journals/prl/Jain10,ikotun2023kmeans}: given a set $P$ of $n$ points in $\R^d$, the goal is to choose at most $k$ centers that minimize the sum of the squared distances from the points to their nearest centers.
It has broad applications in fields ranging from feature engineering~\cite{glassman2014feature,alelyani2018feature} and image processing~\cite{coleman1979image,chang2017deep} to bioinformatics~\cite{ronan2016avoiding,yuan2023index,Zhang2023,nugent2010overview}.

When clustering is used to make decisions about people, however, the composition of the clusters matters as much as their cost.
A loan approval system that clusters applicants may systematically place members of a minority group into a high-risk cluster, even when the overall cost is low; hiring pipelines, school assignments, and resource allocation face the same risk.
A growing body of work therefore studies clusterings that do not disadvantage protected groups~\cite{chierichetti2017fair,bera2019fair,huang2019coresets,chen2019proportionally,ghadiri2021socially}; see~\cite{dickerson2025sok} for a recent critical review.
In the group-balance model of~\cite{chierichetti2017fair}, generalized by~\cite{bera2019fair} to overlapping groups with explicit lower and upper bounds, every cluster must contain each protected group in a prescribed proportion.
The input specifies $m$ protected groups $P^{(1)},\dots,P^{(m)}\subseteq P$, which may overlap since a person can belong to several groups, together with lower and upper bounds $\beta_i$ and $\alpha_i$ on the fraction of each group in every cluster; Section~\ref{sec-problems} defines these constraints formally.

These constraints break the \keyterm{locality property} on which most $k$-means algorithms rely~\cite{ding2020unified,bhattacharya2018faster}.
In standard $k$-means, the centers determine the clusters: every point goes to its nearest center, and the clusters are the Voronoi cells of the centers.
This no longer holds under fairness constraints, since Voronoi cells need not be balanced and some points end up assigned to a center other than their nearest one.
As a result, it is unclear where an optimal solution places its centers.
One way to address this is to use an \keyterm{approximate centroid set}: a finite set of points that contains, for every subset of the input, a point close to the centroid of this subset.
Such a set contains a good candidate center for every cluster of every fair solution, including clusters that are not Voronoi cells.

Fair $k$-means is closely related to sparse Wasserstein barycenters.
A \keyterm{Wasserstein barycenter} of several probability distributions on $\R^d$ is a distribution that minimizes the average squared Wasserstein distance to them~\cite{agueh2011barycenters}.
Since it moves mass along the geometry of the space rather than superimposing densities, it is a natural tool for combining information from several sources, with applications in statistics, machine learning, image processing and fair prediction~\cite{srivastava2018scalable,dognin2019wasserstein,montesuma2021wasserstein,rabin2012wasserstein,chzhen2020fair}; see~\cite{peyre2019computational} for a survey.
Without a bound on its support, a barycenter can be as large as the input distributions themselves~\cite{anderes2016discrete}.
A barycenter supported on at most $k$ points is called a \keyterm{$k$-sparse barycenter}; it is a compact summary described by only $k$ points and their masses, and in some applications $k$ is dictated by the model, for example when collections of distributions are clustered~\cite{ye2017fast,ho2017multilevel} or neural networks are fused into one network of the same width~\cite{singh2020model}.
With the restriction to $k$ support points, the barycenter problem becomes a special case of fractional fair $k$-means.

\subsection{Problem Definitions}\label{sec-problems}

An instance of fair $k$-means consists of a set $P$ of $n$ points in $\R^d$, groups $P^{(1)},\dots,P^{(m)}\subseteq P$, which may overlap, an integer $k\ge1$, and vectors $\beta,\alpha\in[0,1]^m$ of lower and upper bounds.
A solution consists of a set $S\subset\R^d$ of at most $k$ centers and an \keyterm{assignment matrix} $\phi_S$ with one row for every point of $P$ and one column for every center of $S$, so $\phi_S$ is an $n\times|S|$ matrix.
Its entry $\phi_S(p,s)\ge0$ is the fraction of the point $p\in P$ assigned to the center $s\in S$, and every row sums to one, since every point is assigned in full.
The assignment matrix is \keyterm{integral} if all its entries are $0$ or $1$, so that every point belongs to exactly one cluster; a general assignment matrix, which may or may not be integral, is called \keyterm{fractional}.
The \keyterm{column} of a center $s$, that is, the vector $(\phi_S(p,s))_{p\in P}$, describes the cluster of $s$: it records which fraction of each point belongs to $s$, and for an integral matrix it is the indicator vector of the points assigned to $s$.
The cluster of a center $s$ has total weight $w(s)=\sum_{p\in P}\phi_S(p,s)$, of which the group $P^{(i)}$ contributes $w^{(i)}(s)=\sum_{p\in P^{(i)}}\phi_S(p,s)$.
The assignment is \keyterm{$(\alpha,\beta)$-fair} if
\[
\beta_i\,w(s)\le w^{(i)}(s)\le\alpha_i\,w(s)\qquad\text{for all }s\in S\text{ and }i\in[m],
\]
where $[r]=\{1,\dots,r\}$ for a positive integer $r$, and the cost of the solution is
\[
\Cost(P,S,\phi_S)=\sum_{p\in P}\sum_{s\in S}\phi_S(p,s)\lVert p-s\rVert^2 .
\]

\begin{definition}[$(\alpha,\beta)$-fair $k$-means~\cite{bera2019fair}]\label{prob-fair}
Given an instance $(P,P^{(1)},\dots,P^{(m)},k,\alpha,\beta)$, find at most $k$ centers $S\subset\R^d$ and an $(\alpha,\beta)$-fair assignment matrix $\phi_S$ of $P$ to $S$ that minimize $\Cost(P,S,\phi_S)$.
In the \keyterm{integral version}, $\phi_S$ is required to be integral; in the \keyterm{fractional version}, it may be fractional.
\end{definition}

We denote by $\OPT$ the optimal cost of the integral version and by $\OPTfrac$ the infimum of the costs in the fractional version.
In both versions, the centers may lie anywhere in $\R^d$.
Every integral assignment matrix is also fractional, so $\OPTfrac\le\OPT$.
The integral version is the usual formulation of fair $k$-means.

The barycenter problem is stated for distributions rather than points, but, as the definition of the Wasserstein distance shows, its cost has the same form as the cost of a fractional assignment.
To move a discrete probability distribution $P'$ to a distribution $B$ on $\R^d$, the mass of every atom $q$ of $P'$ is divided among the atoms $s$ of $B$.
This is a \keyterm{transport plan}, a nonnegative function $F$ on $\supp(P')\times\supp(B)$ whose values $F(q,s)$ sum to the mass of $q$ over $s\in\supp(B)$ and to the mass of $s$ over $q\in\supp(P')$.
Moving the mass $F(q,s)$ costs $F(q,s)\lVert q-s\rVert^2$, and the squared \keyterm{Wasserstein distance} between $P'$ and $B$ is the minimum total cost of a transport plan,
\[
\mathcal W^2(P',B)=\min_F\sum_{q\in\supp(P')}\sum_{s\in\supp(B)}F(q,s)\lVert q-s\rVert^2 .
\]

\begin{definition}[$k$-sparse Wasserstein barycenter]\label{prob-wb}
Given discrete probability distributions $P^{(1)},\allowbreak\dots,\allowbreak P^{(m)}$ on $\R^d$ with $n$ atoms in total and an integer $k\ge1$, find a probability distribution $B$ supported on at most $k$ points that minimizes
\[
\Cost_{\mathrm{WB}}(B)=\frac1m\sum_{i=1}^m\mathcal W^2\bigl(P^{(i)},B\bigr).
\]
\end{definition}

We denote by $\OPT_{\mathrm{WB}}$ the infimum of $\Cost_{\mathrm{WB}}(B)$ over these distributions.

The $k$-sparse barycenter problem is a special case of fractional fair $k$-means.
The points of the corresponding instance are the atoms of all input distributions, each weighted by its mass; the atoms of $P^{(i)}$ form the $i$-th group, and $\alpha_i=\beta_i=1/m$, so that every distribution contributes the same mass to every center.
The definitions above extend to weighted points by multiplying every term of a sum over points by the weight of the point.
For a center set $S$, multiplying the entries of a fair assignment matrix in the rows of $P^{(i)}$ by the masses of the atoms gives a transport plan from $P^{(i)}$ to the distribution $B$ with mass $w(s)/m$ at each $s\in S$, and every family of transport plans from $P^{(1)},\dots,P^{(m)}$ to a common distribution supported on $S$ arises in this way.
The cost of the assignment is the total cost of these transport plans, which is at least $m\Cost_{\mathrm{WB}}(B)$, with equality if the plans are optimal.
Hence the optimal cost of this instance of fractional fair $k$-means is $m\,\OPT_{\mathrm{WB}}$, and an approximation algorithm for fractional fair $k$-means yields one for $k$-sparse barycenters with the same ratio; Lemma~\ref{lem-wb-correspondence} in Section~\ref{sec-wb} proves this correspondence.

\subsection{Our Results}\label{sec-results}

We give an algorithm for fair $k$-means with approximation ratio $2\rho+1+\epsilon$, where $\rho$ is the ratio of an algorithm for weighted $k$-means that it calls once and $\epsilon>0$ is an accuracy parameter.
The same algorithm achieves the same ratio for $k$-sparse Wasserstein barycenters.
In fixed dimension, this ratio is $3+O(\epsilon)$, and we show that the constant $3$ cannot be improved: for every fixed $\delta>0$, unless $\mathrm P=\mathrm{NP}$, no polynomial-time algorithm can return a fair fractional assignment of cost less than $(3-\delta)\,\OPT$, even for points on a line.

\paragraph{Fractional fair $k$-means.}
\begin{theorem}[Fractional fair $k$-means]\label{thm-main}
Let $\rho\ge1$, and assume access to a deterministic polynomial-time $\rho$-approximation algorithm for weighted Euclidean $k$-means.
Then, for every $\epsilon\in(0,1]$, a deterministic polynomial-time algorithm computes at most $k$ centers $S$ together with a fractional assignment $\phi_S$ that satisfies all $(\alpha,\beta)$-fairness constraints exactly and has cost
\begin{equation}\label{eq-main-guarantee}
\Cost(P,S,\phi_S)\le(2\rho+1+\epsilon)\,\OPTfrac\le(2\rho+1+\epsilon)\,\OPT .
\end{equation}
Besides a single call to the $k$-means algorithm, it solves $O((\rho/\epsilon)^2)$ linear programs, each of size polynomial in the input.
\end{theorem}

If a randomized $k$-means algorithm with ratio $\rho$ is used instead, the same guarantee holds with the corresponding success probability.

This guarantee improves the ratio $1+4\rho$ of Song, Mo, and Ding~\cite{song2025relax}: the coefficient of $\rho$ decreases from $4$ to $2$.
In fixed dimension, local search~\cite{friggstad2019local,cohen2019local} computes, for every fixed $\epsilon'>0$, a $(1+\epsilon')$-approximation for $k$-means in polynomial time, and the version of Friggstad, Rezapour, and Salavatipour~\cite{friggstad2019local} also applies to weighted points.
With $\epsilon'=\epsilon$, that is, $\rho=1+\epsilon$, Theorem~\ref{thm-main} gives the ratio $2(1+\epsilon)+1+\epsilon=3+3\epsilon=3+O(\epsilon)$, whereas the ratio $1+4\rho$ of~\cite{song2025relax} gives $5+5\epsilon=5+O(\epsilon)$.
Theorem~\ref{thm-main} holds in every dimension; only the size of its linear programs depends on whether the dimension is fixed.
For fixed $d$, $\rho$ and $\epsilon$, each linear program has $O(n^2)$ variables.
If $d$ is part of the input, the candidate sets of the algorithm have $n^{O(\rho/\epsilon)}$ points, so the running time is polynomial for fixed $\rho$ and $\epsilon$; Proposition~\ref{prop-running-time} states both bounds precisely.
Table~\ref{tab-comparison} compares these ratios with the previous ones.

\paragraph{$k$-sparse Wasserstein barycenters.}
For $k$-sparse barycenters, a fractional assignment is exactly what is needed: a transport plan may split the mass of an atom among several support points, so the fractional assignment computed by the algorithm of Theorem~\ref{thm-main} directly yields a barycenter together with transport plans to it.

\begin{theorem}[$k$-sparse Wasserstein barycenters]\label{thm-wb}
Let $\rho\ge1$, and assume access to a deterministic polynomial-time $\rho$-approximation algorithm for weighted Euclidean $k$-means.
Then, for every $\epsilon\in(0,1]$, a deterministic polynomial-time algorithm computes a probability distribution $B$ supported on at most $k$ points with
\[
\Cost_{\mathrm{WB}}(B)\le(2\rho+1+\epsilon)\,\OPT_{\mathrm{WB}} .
\]
In particular, in every fixed dimension, using the PTAS for weighted $k$-means with $\rho=1+\epsilon$ gives a deterministic polynomial-time algorithm with approximation ratio $3+O(\epsilon)$.
\end{theorem}

Without the bound $k$ on the support size, the problem has been studied extensively.
Altschuler and Boix-Adser\`a~\cite{altschuler2021wasserstein} compute an exact barycenter in every fixed dimension~$d$ in polynomial time whose degree grows linearly in $d$; this dependence is unavoidable unless $\mathrm P=\mathrm{NP}$~\cite{altschuler2022wasserstein}.
Borgwardt~\cite{borgwardt2022lp} gives a strongly polynomial $2$-approximation by restricting the support to the union of the input supports and computing an optimal vertex of the resulting linear program.
Wang and Xie~\cite{wang2025closing} use a larger candidate set and a linear program to give a $(1+\epsilon)$-approximation in time polynomial in $n^{1/\epsilon}$ and $d$, a PTAS in arbitrary dimension.
These algorithms return up to $n-m+1$ or $n^{O(1/\epsilon)}$ support points, so they do not solve the $k$-sparse problem, which is harder: for $m=1$, the unrestricted barycenter is the distribution itself at cost~$0$, whereas a $k$-sparse barycenter is exactly weighted $k$-means, NP-hard in the plane~\cite{DBLP:journals/tcs/MahajanNV12}.
For the $k$-sparse problem, the best previous ratios of polynomial-time algorithms were $(2+\sqrt\rho)^2$~\cite{yang2024approximate} and $1+4\rho$~\cite{song2025relax}; the ratio $2\rho+1+\epsilon$ improves on both for every $\rho\le19$.
Table~\ref{tab-comparison} lists these results next to those for fair $k$-means.

\begin{table}[t]
\centering
\small
\begin{tabular}{llll}
\toprule
Algorithm & Ratio & Fixed dimension & Setting \\
\midrule
\multicolumn{4}{@{}l}{\emph{Fair $k$-means}} \\
\cite{bera2019fair} & $(2+\sqrt\rho)^2$ & $9+O(\epsilon)$ & general case \\
\cite{schmidt2020fair} & $5.5\rho+1$ & $6.5+O(\epsilon)$ & two groups \\
\cite{bohm2021algorithms} & $(2+\sqrt\rho)^2$ & $9+O(\epsilon)$ & exact fairness \\
\cite{song2025relax} & $1+4\rho$ & $5+O(\epsilon)$ & general case \\
Theorem~\ref{thm-main} & $2\rho+1$ & $3+O(\epsilon)$ & general case \\
Theorem~\ref{thm-hardness} (lower bound) & no $3-\delta$ & no $3-\delta$ & general case \\
\midrule
\multicolumn{4}{@{}l}{\emph{Wasserstein barycenters with unrestricted support}} \\
\cite{altschuler2021wasserstein} & $1$ & $1$ (exact) & support $n-m+1$ \\
\cite{borgwardt2022lp} & $2$ & $2$ & support $n-m+1$ \\
\cite{wang2025closing} & $1+\epsilon$ & $1+\epsilon$ & support $n^{O(1/\epsilon)}$ \\
\midrule
\multicolumn{4}{@{}l}{\emph{$k$-sparse Wasserstein barycenters}} \\
\cite{wang2026closingfull} & $1+\epsilon$ & $1+\epsilon$ & mixed-integer program \\
\cite{yang2024approximate} & $(2+\sqrt\rho)^2$ & $9+O(\epsilon)$ & \\
\cite{song2025relax} & $1+4\rho$ & $5+O(\epsilon)$ & \\
Theorem~\ref{thm-wb} & $2\rho+1$ & $3+O(\epsilon)$ & \\
\bottomrule
\end{tabular}
\caption{Approximation ratios for fair $k$-means and Wasserstein barycenters, up to an additive $\epsilon>0$.
The column ``Fixed dimension'' substitutes $\rho=1+\epsilon$~\cite{friggstad2019local}.
The algorithms of the middle group are compared with the best barycenter of any support size, and their column ``Setting'' gives an upper bound on the support size of the output, where $n$ is the total number of atoms of the $m$ input distributions; all other algorithms output at most $k$ points.
The algorithm of~\cite{altschuler2021wasserstein} is polynomial only for fixed~$d$.
The method of~\cite{wang2026closingfull} solves a mixed-integer program and is not guaranteed to run in polynomial time.
The lower bound holds for every fixed $\delta>0$, unless $\mathrm P=\mathrm{NP}$.}
\label{tab-comparison}
\end{table}

\paragraph{Integral fair $k$-means.}
The fractional assignment of Theorem~\ref{thm-main} can be rounded to an integral one without changing the centers or increasing the cost, at a bounded additive violation of the fairness bounds.
Recall that $P$ is the set of input points and that $P^{(1)},\dots,P^{(m)}\subseteq P$ are the groups.
We distinguish two cases: the groups may overlap, that is, a point may belong to several groups, or they may not overlap, that is, every point belongs to at most one group.
For overlapping groups, the additive violation depends on $\Delta$, the maximum number of groups that contain a common point.

\begin{corollary}[Overlapping groups]\label{cor-overlapping}
Using the rounding procedure of~\cite{bera2019fair}, one can replace, in time polynomial in $n$, $k$ and $m$, the assignment $\phi_S$ of Theorem~\ref{thm-main} by an integral assignment $\phi'_S$ to the same centers whose additive violation is at most $4\Delta+3$ and which satisfies
\begin{equation}\label{eq-rounding-cost}
\Cost(P,S,\phi'_S)\le\Cost(P,S,\phi_S)\le(2\rho+1+\epsilon)\,\OPTfrac .
\end{equation}
\end{corollary}

\begin{corollary}[Non-overlapping groups]\label{cor-partition}
If the groups do not overlap, then the rounding procedure of~\cite{song2025relax} computes, in $O(n^3k^2)$ time, an integral assignment $\phi'_S$ to the center set computed in Theorem~\ref{thm-main} that satisfies~\eqref{eq-rounding-cost} and has additive violation at most $2$.
\end{corollary}

\paragraph{Lower bound.}
No inapproximability result specific to fair $k$-means was previously known.
The following lower bound, which we prove in Section~\ref{sec-hardness}, shows that the ratio $3+O(\epsilon)$ that Theorem~\ref{thm-main} achieves in fixed dimension almost matches the best ratio achievable in polynomial time.

\begin{theorem}[Hardness of approximation]\label{thm-hardness}
Let $\delta\in(0,1]$ be fixed.
If $\mathrm P\ne\mathrm{NP}$, then no polynomial-time algorithm\footnote{We count arithmetic operations in the real RAM model and the time of the linear program solver separately.
Theorem~\ref{thm-hardness}, like every NP-hardness result, concerns polynomial time on a Turing machine; the numbers computed by our algorithm have bit length polynomial in the input length, and the ellipsoid method~\cite{khachiyan1980polynomial} solves its linear programs exactly in polynomial time, so our algorithm also runs in polynomial time on a Turing machine whenever its $k$-means subroutine does.} returns, for every instance of fair $k$-means, at most $k$ centers $S$ and an $(\alpha,\beta)$-fair fractional assignment $\phi_S$ satisfying
\[
\Cost(P,S,\phi_S)\le(3-\delta)\,\OPT .
\]
This remains true even for instances with $d=1$ and groups with $\beta_i=\frac12$ and $\alpha_i\in\{\frac12,1\}$.
\end{theorem}

Both bounds compare the cost of a fair fractional assignment with $\OPT$, so in fixed dimension they differ only by terms that vanish with $\epsilon$ and $\delta$.
Since every integral assignment is also fractional, the lower bound applies to algorithms that return integral solutions as well.

\paragraph{Beyond group fairness.}
Our algorithm uses the fairness constraints only through three properties: each constraint involves the column of a single cluster, and it is linear and homogeneous in this column, so it is preserved when columns are scaled or added.
Theorem~\ref{thm-main} therefore holds, more generally, for \keyterm{$k$-means with homogeneous cluster constraints}, in which every cluster has to satisfy $r$ given linear inequalities in place of the fairness constraints.
The $j$-th inequality assigns a coefficient $a_j(p)$ to every point $p$ and requires that the fractions of the points in the cluster of $s$, weighted by these coefficients, have a nonnegative sum:
\[
\sum_{p\in P}a_j(p)\,\phi_S(p,s)\ge0\qquad\text{for all }s\in S\text{ and }j\in[r].
\]
For example, if every point is red or blue, the requirement that every cluster contain at least as many red points as blue points is the single inequality ($r=1$) with $a_1(p)=1$ for red points and $a_1(p)=-1$ for blue points.

Fair $k$-means, which contains the $k$-sparse barycenter problem, is the special case with $r=2m$, in which every inequality is a lower or an upper bound on the share of a single group in the cluster (Remark~\ref{rem-fair-special}).
The class also contains constraints of other forms, such as a bound on the ratio between two groups in every cluster, which generalizes the balance constraints of~\cite{chierichetti2017fair} and can require, for example, at least one mentor for every three trainees in every project team, and bounds on the average of a numerical attribute, such as income, in every cluster.
Section~\ref{sec-beyond} gives the formal definition, states the generalization as Theorem~\ref{thm-homogeneous}, and proves it.

Since the class contains fair $k$-means, the lower bound of Theorem~\ref{thm-hardness} applies to it as well, and Corollary~\ref{cor-homogeneous-threshold} concludes that $3$ is the approximation threshold of the whole class in fixed dimension.
This complements the constrained $k$-means frameworks of Ding and Xu~\cite{ding2020unified} and of Bhattacharya, Jaiswal, and Kumar~\cite{bhattacharya2018faster}, which achieve ratio $1+\epsilon$ for general constraints but in time exponential in~$k$.

\paragraph{Organization.}
Section~\ref{sec-overview} gives an overview of our techniques, and Section~\ref{sec-related} discusses related work.
After the preliminaries in Section~\ref{sec-prelim}, Section~\ref{sec-algorithm} presents the algorithm and Section~\ref{sec-analysis} proves Theorem~\ref{thm-main} and extends it to homogeneous cluster constraints.
Section~\ref{sec-hardness} proves Theorem~\ref{thm-hardness}, Section~\ref{sec-wb} proves Theorem~\ref{thm-wb}, and Section~\ref{sec-conclusion} concludes.

\section{Technical Overview}\label{sec-overview}

This section explains the main ideas behind Theorems~\ref{thm-main}, \ref{thm-wb} and~\ref{thm-hardness}.

The starting observation is that choosing at most $k$ centers that admit a fair assignment is hard, but \emph{choosing many} such centers is easy: solving a linear program on a large candidate set produces a fair fractional clustering whose many small clusters, which we call \keyterm{fragment clusters}, together cost close to the optimum.
The challenge is to merge these fragment clusters into at most $k$ final clusters without paying too much.
A natural approach, and the one taken in~\cite{song2025relax}, is to represent each fragment cluster by its centroid and run a $k$-means algorithm on these centroids. The fragment clusters assigned to the same final center are then combined, and the combined cluster is again fair, because the fairness constraints are homogeneous linear inequalities in the column of a cluster and therefore hold for every sum of columns that satisfy them.

The cost of the final solution splits into two parts: the cost $C_T$ of the fragment clusters themselves and the cost of rerouting them to the $k$ chosen centers.
The rerouting cost depends on a cross term, the \keyterm{alignment} $C_e$, which measures whether the displacements introduced by the fractional assignment point toward or away from the optimal centers.
Without any control over the alignment, the Cauchy--Schwarz inequality gives a final ratio of $1{+}4\rho$~\cite{song2025relax}, where $\rho$ is the approximation ratio of the $k$-means subroutine.

Our key idea is to add a regularizer, the \keyterm{drift}, that penalizes how far the fractional assignment moves each point from its original position.
Mixing a small fraction of the optimal solution into the current fractional assignment decreases the drift at a rate that grows with the alignment, so at a near-optimal point of the regularized objective, where no such mixture improves the objective much, the alignment cannot be large.
As a result, $C_T$ drops out of the bound entirely, and the final ratio improves to $2\rho{+}1$.

We develop this in four steps.
Section~\ref{sec-overview-merge} analyzes the merge step and isolates the alignment as the quantity to control.
Section~\ref{sec-overview-relaxation} introduces the drift-regularized relaxation and shows that its near-optimality bounds the alignment; this is where the improved ratio comes from.
Section~\ref{sec-overview-fw} shows how to solve the relaxation approximately by a sequence of linear programs, with an iterative scheme modeled on the Frank--Wolfe method~\cite{frank1956algorithm,jaggi2013revisiting}.
Section~\ref{sec-overview-wb} explains why the algorithm also solves $k$-sparse Wasserstein barycenters, and Section~\ref{sec-overview-hardness} describes the reduction for the lower bound.

\subsection{From Fragment Clusters to \texorpdfstring{$k$}{k} Clusters}\label{sec-overview-merge}

For a fixed set of centers, a minimum-cost fair fractional assignment can be computed by a linear program, which we call the \keyterm{fair assignment LP}; the difficulty lies in choosing at most $k$ centers.
Without the bound $k$, the problem becomes easy.
The fair assignment LP of a large candidate set has optimal value at most $(1+\epsilon)\OPT$ if the set contains a good center for every cluster of the optimal solution; an approximate centroid set of $P$, defined in Definition~\ref{def-acs} of Section~\ref{sec-prelim}, has this property.
We call a finite set $T$ of \keyterm{candidates} together with a fair fractional assignment $\phi_T$ of $P$ to $T$ a \keyterm{relaxed solution}.
The column of a candidate $t$ in $\phi_T$ is its \keyterm{fragment cluster}: it specifies the fraction of each point that belongs to $t$. Its weight and weighted centroid are
\[
w(t)=\sum_{p\in P}\phi_T(p,t),
\qquad
\pi(t)=\frac1{w(t)}\sum_{p\in P}\phi_T(p,t)\,p ,
\]
and if every candidate is moved to its centroid, the relaxed solution costs
\begin{equation}\label{eq-overview-ct}
C_T=\sum_{p\in P}\sum_{t\in T}\phi_T(p,t)\lVert p-\pi(t)\rVert^2 .
\end{equation}

To convert a relaxed solution into one with at most $k$ centers, we cluster the centroids of the fragment clusters.
We run a $\rho$-approximation algorithm for weighted $k$-means on the weighted centroids $\{(\pi(t),w(t))\}$ and obtain a set $S$ of at most $k$ centers. Each candidate $t$ is then reassigned, column by column, to a center in $S$ nearest to $\pi(t)$.
Each center of $S$ then receives a sum of fair columns, which is again fair, so the reassigned solution is fair.
Serving a weighted set by a point $s$ instead of its centroid increases the cost by the total weight times the squared distance between $s$ and the centroid.
This fact about squared Euclidean distances, the \keyterm{centroid identity}, is stated as~\eqref{eq-centroid-identity} in Section~\ref{sec-prelim}; it follows by expanding each squared distance around the centroid, because the cross terms sum to zero.
By the centroid identity, applied to the fragment cluster of $t$, whose total weight is $w(t)$ and whose centroid is $\pi(t)$, moving the column of $t$ from $\pi(t)$ to a center $s$ adds exactly $w(t)\lVert\pi(t)-s\rVert^2$ to its cost, and the resulting solution costs at most
\begin{equation}\label{eq-overview-routing}
C_T+\rho\,\OPT_\pi,
\qquad\text{where}\quad
\OPT_\pi=\min_{|S'|\le k}\ \sum_{t\in T}w(t)\min_{s\in S'}\lVert\pi(t)-s\rVert^2
\end{equation}
is the optimal $k$-means cost of the weighted centroids.
The algorithm then replaces this assignment by an optimal solution of the fair assignment LP of $S$, which costs no more.

To bound $\OPT_\pi$, we compare the centroids with an optimal solution.
We denote its clusters by $\widetilde C_1,\dots,\widetilde C_k$ and its centers by $\widetilde s_1,\dots,\widetilde s_k$, where each center is the centroid of its cluster.
A relaxed solution splits each point $p$ among several fragment clusters.
When these clusters are later merged, $p$ is no longer at its original position but is effectively represented by the weighted mean of the centroids it belongs to,
\[
\nu_p=\sum_{t\in T}\phi_T(p,t)\,\pi(t) ;
\]
see Figure~\ref{fig-drift}.
The vector $\nu_p-p$ is the \keyterm{displacement} of $p$: it records how far $p$ moves when the fractional assignment replaces it by cluster centroids.
The total squared displacement,
\begin{equation}\label{eq-overview-drift}
C_\nu=\sum_{p\in P}\lVert\nu_p-p\rVert^2 ,
\end{equation}
is the \keyterm{drift} of the relaxed solution.
A small drift means that the centroid representation is faithful to the original positions; a large drift means the merge step inherits positions that may be far from the points they represent, which inflates the alignment and the final cost.
Since the fractions $\phi_T(p,t)$ of $p$ are nonnegative and sum to one, $\nu_p$ is a convex combination of the centroids $\pi(t)$, and Jensen's inequality for the convex function $x\mapsto\lVert p-x\rVert^2$ gives $\lVert p-\nu_p\rVert^2\le\sum_t\phi_T(p,t)\lVert p-\pi(t)\rVert^2$.
Summing over $p$ and comparing with~\eqref{eq-overview-ct}, we obtain $C_\nu\le C_T$.

Figure~\ref{fig-drift} shows where the drift comes from.
The centroid of a fragment cluster is pulled toward the other points of that cluster.
If the fragment clusters that share $p$ extend to different sides of $p$, these pulls cancel and $\nu_p=p$; if they extend to the same side, the pulls add up and $\nu_p$ lies away from $p$.
In both cases $p$ has the same squared distances to its centroids, so it contributes equally to $C_T$, but it contributes to the drift only in the second case.

\begin{figure}[!ht]
\centering
\begin{tikzpicture}[fig,scale=0.95,x=1.35cm]
\begin{scope}
  \begin{scope}[transparency group,opacity=0.6]
    \filldraw[figblue!25,line width=14pt,line join=round] (0,0)--(-1.4,1.4)--(-1.7,2.1)--(-0.9,2.0)--cycle;
  \end{scope}
  \begin{scope}[transparency group,opacity=0.6]
    \filldraw[figgray!30,line width=14pt,line join=round] (0,0)--(1.4,-1.4)--(1.7,-2.1)--(0.9,-2.0)--cycle;
  \end{scope}
  \foreach \q in {(-1.4,1.4),(-1.7,2.1),(-0.9,2.0),(1.4,-1.4),(1.7,-2.1),(0.9,-2.0)}
    \node[figpt,fill=figgray] at \q {};
  \draw[->,figgray,line width=0.6pt,shorten >=3pt] (0,0) -- (-1.1429,1.5714);
  \draw[->,figgray,line width=0.6pt,shorten >=3pt] (0,0) -- (1.1429,-1.5714);
  \node[figbluept] at (-1.1429,1.5714) {};
  \node[figbluept] at (1.1429,-1.5714) {};
  \node[figctr] at (0,0) {};
  \node[figbigpt] at (0,0) {};
  \node[figlab,left=6pt] at (-1.75,1.65) {$\pi(t_1)$};
  \node[figlab,right=6pt] at (1.75,-1.65) {$\pi(t_2)$};
  \draw[figgray!80,line width=0.4pt] (-1.1429,1.5714) -- (-1.95,1.65);
  \draw[figgray!80,line width=0.4pt] (1.1429,-1.5714) -- (1.95,-1.65);
  \node[figlab,below left=3pt and 2pt] at (0,0) {$p=\nu_p$};
  \node[figlabs,align=center] at (0,-3.05) {pulls cancel:\\ drift of $p$ is $0$};
\end{scope}
\begin{scope}[xshift=7.6cm]
  \begin{scope}[transparency group,opacity=0.6]
    \filldraw[figblue!25,line width=14pt,line join=round] (0,0)--(1.4,1.4)--(1.7,2.1)--(0.9,2.0)--cycle;
  \end{scope}
  \begin{scope}[transparency group,opacity=0.6]
    \filldraw[figgray!30,line width=14pt,line join=round] (0,0)--(1.4,-1.4)--(1.7,-2.1)--(0.9,-2.0)--cycle;
  \end{scope}
  \foreach \q in {(1.4,1.4),(1.7,2.1),(0.9,2.0),(1.4,-1.4),(1.7,-2.1),(0.9,-2.0)}
    \node[figpt,fill=figgray] at \q {};
  \draw[figblue!55,line width=0.5pt,densely dashed] (1.1429,1.5714) -- (1.1429,-1.5714);
  \draw[->,figgray,line width=0.6pt,shorten >=3pt] (0,0) -- (1.1429,1.5714);
  \draw[->,figgray,line width=0.6pt,shorten >=3pt] (0,0) -- (1.1429,-1.5714);
  \draw[->,figblue,line width=1.1pt,shorten >=5.5pt] (0,0) -- (1.1429,0);
  \node[figbluept] at (1.1429,1.5714) {};
  \node[figbluept] at (1.1429,-1.5714) {};
  \node[figctr] at (1.1429,0) {};
  \node[figbigpt] at (0,0) {};
  \node[figlab,right=6pt] at (2.05,1.65) {$\pi(t_1)$};
  \node[figlab,right=6pt] at (2.05,-1.65) {$\pi(t_2)$};
  \draw[figgray!80,line width=0.4pt] (1.1429,1.5714) -- (2.25,1.65);
  \draw[figgray!80,line width=0.4pt] (1.1429,-1.5714) -- (2.25,-1.65);
  \node[figlab,left=4pt] at (0,0) {$p$};
  \node[figlab,right=5pt] at (1.1429,0) {$\nu_p$};
  \node[text=figblue,figlabs,below=2pt] at (0.6,0) {$\nu_p-p$};
  \node[figlabs,align=center] at (0.9,-3.05) {pulls add up:\\ drift of $p$ is $\lVert\nu_p-p\rVert^2>0$};
\end{scope}
\end{tikzpicture}
\caption{Where the drift comes from. The point $p$ assigns half of its mass to each of two fragment clusters (shaded); the other points of these clusters are shown in gray. Each centroid $\pi(t_i)$ is pulled toward the other points of its cluster, and $\nu_p$ is the mean of $\pi(t_1)$ and $\pi(t_2)$. Left: the two clusters extend to opposite sides of $p$, the displacements toward the centroids cancel, and $\nu_p=p$. Right: both clusters extend to the same side, and $\nu_p$ moves away from $p$. In both configurations $p$ has the same distances to its centroids and hence the same contribution to $C_T$; only its contribution to the drift $C_\nu$ differs.}
\label{fig-drift}
\end{figure}

Serving each centroid $\pi(t)$ by the optimal centers of the points in its cluster, in proportion to their fractions, and expanding
\[
\lVert\pi(t)-\widetilde s_j\rVert^2=\lVert\pi(t)-p\rVert^2+\lVert p-\widetilde s_j\rVert^2+2\,(\pi(t)-p)^\top(p-\widetilde s_j)
\qquad\text{for }p\in\widetilde C_j
\]
gives the bound
\begin{equation}\label{eq-overview-comparison}
\OPT_\pi\le C_T+\OPT+2C_e,
\qquad
C_e=\sum_{j=1}^k\sum_{p\in\widetilde C_j}(\nu_p-p)^\top(p-\widetilde s_j).
\end{equation}
The cross term $C_e$ is the \keyterm{alignment} of the relaxed solution with the optimal solution.
The term of $p$ is negative if the displacement of $p$ points toward its optimal center $\widetilde s_j$, and positive if it points away from it (Figure~\ref{fig-idea}).
The difficulty is that the fair assignment LP controls only $C_T$, not $C_e$. Its optimality guarantees $C_T\le(1+\epsilon)\OPT$ but says nothing about the directions of the displacements.

\begin{figure}[!htb]
\centering
\begin{minipage}[c]{0.40\textwidth}
\centering
\begin{tikzpicture}[fig]
\begin{scope}
  \coordinate (p)   at (0,1.0);
  \coordinate (pi1) at (1.9,2.95);
  \coordinate (pi2) at (2.35,-0.35);
  \coordinate (nu)  at ($(pi1)!0.55!(pi2)$);
  \coordinate (st)  at (4.7,3.0);
  \coordinate (s)   at (4.55,4.05);
  \draw[figgray!70,line width=0.5pt] (p) -- (pi1);
  \draw[figgray!70,line width=0.5pt] (p) -- (pi2);
  \draw[figblue!55,line width=0.5pt,densely dashed] (pi1) -- (pi2);
  \draw[->,figgray,line width=0.6pt,shorten >=5.5pt] (pi1) -- node[figlabs,above,sloped,text=figgray] {merge} (s);
  \draw[->,figred,line width=0.8pt,densely dashed,shorten >=5.5pt] (p) -- (st);
  \draw[->,figblue,line width=1.1pt,shorten >=5.5pt] (p) -- (nu);
  \node[text=figblue,below=1pt] at ($(p)!0.62!(nu)$) {$\nu_p-p$};
  \node[figpt] at (p) {};
  \node[figbluept] at (pi1) {};
  \node[figbluept] at (pi2) {};
  \node[figctr] at (nu) {};
  \node[figopt] at (st) {};
  \node[circle,draw=figgray,line width=0.9pt,inner sep=0pt,minimum size=10pt] at (s) {};
  \node[left=3pt] at (p) {$p$};
  \node[above left=0pt and 1pt] at (pi1) {$\pi(t_1)$};
  \node[below=3pt] at (pi2) {$\pi(t_2)$};
  \node[right=5pt] at (nu) {$\nu_p$};
  \node[right=5pt] at (st) {$\widetilde s$};
  \node[right=5pt] at (s) {$s$};
\end{scope}
\end{tikzpicture}
\end{minipage}\hfill
\begin{minipage}[c]{0.56\textwidth}
\caption{A point $p$ split between two fragment clusters is moved by the relaxation to the weighted mean $\nu_p$ of their centroids $\pi(t_1)$ and $\pi(t_2)$.
The alignment compares the displacement $\nu_p-p$ with the direction to the optimal center $\widetilde s$ (red), and clustering the centroids moves each centroid to a center $s$ chosen by the $k$-means algorithm.}
\label{fig-idea}
\end{minipage}
\end{figure}

\subsection{The Drift-Regularized Relaxation}\label{sec-overview-relaxation}

By~\eqref{eq-overview-routing} and~\eqref{eq-overview-comparison}, the cost of the final solution is at most $C_T+\rho\,(C_T+\OPT+2C_e)$.
To achieve the ratio $2\rho{+}1$, we need to control both $C_T$ and $C_e$.

The alignment depends on the unknown optimal centers, so we cannot compute or minimize it directly.
We can, however, detect a large alignment indirectly through the drift.
Consider what happens when the current relaxed solution is mixed with a small fraction $\theta\in(0,1)$ of the optimal solution: every point $p\in\widetilde C_j$ keeps $1-\theta$ times each of its current fractions and assigns the remaining fraction $\theta$ to the optimal center $\widetilde s_j$.
Scaling a column does not move its centroid, and $\widetilde s_j$ is the centroid of $\widetilde C_j$, so every $\nu_p$ moves to $(1-\theta)\nu_p+\theta\,\widetilde s_j$, that is, toward $\widetilde s_j$.
Differentiating the drift~\eqref{eq-overview-drift} at $\theta=0$ shows that it changes at the rate $2\sum_{j=1}^k\sum_{p\in\widetilde C_j}(\nu_p-p)^\top(\widetilde s_j-\nu_p)$, which equals $-2(C_\nu+C_e)$ by the definition~\eqref{eq-overview-comparison} of the alignment.
A large positive $C_e$ therefore makes the drift drop steeply.

This suggests adding a multiple of the drift to the objective.
If the combined objective---cost plus drift---is near its minimum, then no mixture, including the one with the optimal solution, can decrease it much further.
Since the drift changes at the rate $-2(C_\nu+C_e)$ along the mixture with the optimal solution, the near-optimality forces $C_e$ to be small.
Thus the algorithm never needs to compute the alignment, which depends on the unknown optimal centers: it controls the alignment by bringing the regularized objective close to its minimum.

The weight $\lambda$ of the regularizer should be chosen to make the final bound as small as possible.
In the bound~\eqref{eq-overview-merge} on the final cost, derived later in this subsection, the cost $C_T$ of the relaxed solution appears with a coefficient that depends on $\lambda$; the choice $\lambda=\frac\rho{1+\rho}$ makes this coefficient vanish, so that $C_T$ drops out of the guarantee entirely.
The last paragraph of this subsection, \emph{Why the weight $\rho/(1+\rho)$}, carries out this calculation for an arbitrary weight $\lambda$.
For the chosen weight, $\lambda\in[\frac12,1)$ because $\rho\ge1$.

\begin{definition}[Drift-regularized cost]\label{def-drift}
Let $\rho\ge1$ be the ratio of the $k$-means algorithm, and let $\lambda=\frac\rho{1+\rho}$.
The \keyterm{drift-regularized cost} of a relaxed solution $(T,\phi_T)$ is
\begin{equation}\label{eq-overview-psi}
\Psi(\phi_T)=C_T+\lambda\,C_\nu ,
\end{equation}
where $C_T$ and $C_\nu$ are the cost~\eqref{eq-overview-ct} and the drift~\eqref{eq-overview-drift} of $(T,\phi_T)$.
The \keyterm{drift-regularized relaxation} asks for the infimum of $\Psi(\phi_T)$ over all relaxed solutions $(T,\phi_T)$, with any finite number of candidates anywhere in $\R^d$.
\end{definition}

The quantities $C_T$ and $C_\nu$ depend only on the fractions $\phi_T(p,t)$, because the centroids $\pi(t)$ are determined by these fractions alone.
The drift-regularized cost is therefore a function of the assignment matrix alone, and we write $\Psi(\phi_T)$ for its value at $\phi_T$.

The regularizer favors relaxed solutions with smaller point displacements.
A point that lies in a single fragment cluster is moved to the centroid of this cluster, whereas a point split between several fragment clusters is moved to the weighted mean of their centroids, which can be much closer to the point than the centroids themselves if they lie on different sides of it.

Since $C_\nu$ is quadratic in the centroids, the relaxation is not a linear program and cannot be solved directly.
To optimize it, we use an iterative scheme modeled on the Frank--Wolfe method, described in Section~\ref{sec-overview-fw}; like this method, it requires the objective to be convex along certain combinations of feasible points.
The natural combination of two relaxed solutions is to take a convex mixture of their assignment matrices: keep both candidate sets, and scale the fractions of the first solution by $1-\theta$ and those of the second by $\theta$, for some $\theta\in[0,1]$.

Formally, let $(T_1,\phi_{T_1})$ and $(T_2,\phi_{T_2})$ be relaxed solutions, let $\theta\in[0,1]$, and let $T_\theta=T_1\uplus T_2$ be the disjoint union of their candidate sets.
The \keyterm{$\theta$-mixture} of the two relaxed solutions is the pair $(T_\theta,\phi_{T_\theta})$ whose assignment matrix $\phi_{T_\theta}=(1-\theta)\phi_{T_1}\oplus\theta\phi_{T_2}$ assigns the fraction $(1-\theta)\phi_{T_1}(p,t)$ of $p$ to each $t\in T_1$ and the fraction $\theta\phi_{T_2}(p,t)$ to each $t\in T_2$.
It is indeed a relaxed solution, because scaling a column changes neither its fairness nor its centroid; for the same reason, its cost $C_{T_\theta}$ and its weighted means $\nu_p(\phi_{T_\theta})$ are affine functions of $\theta$.
Since the drift is a sum of squared distances between the points and their weighted means, $\Psi(\phi_{T_\theta})$ is a convex quadratic function of $\theta$.
This is the convexity that our scheme needs; Lemma~\ref{lem-mixing} in Section~\ref{sec-algorithm} gives the exact formula.

If $(T_\theta,\phi_{T_\theta})$ is the $\theta$-mixture of a relaxed solution $(T,\phi_T)$ with another relaxed solution $(T',\phi_{T'})$, then $\Psi(\phi_{T_\theta})$ changes at $\theta=0$ at the rate $\cL(\phi_{T'})-\cL(\phi_T)$, where
\[
\cL(\phi_{T'})=C_{T'}+2\lambda\sum_{p\in P}(\nu_p-p)^\top\bigl(\nu_p(\phi_{T'})-p\bigr)
\]
is the \keyterm{linearized cost} at $(T,\phi_T)$; here $C_{T'}$ and $\nu_p(\phi_{T'})$ are the cost and the weighted means of $(T',\phi_{T'})$, the points $\nu_p$ are those of $(T,\phi_T)$, and $\cL(\phi_T)=C_T+2\lambda C_\nu$.

\begin{definition}[Stationarity]\label{def-stationary}
Let $\tau\ge0$, and let $\cL$ be the linearized cost at a relaxed solution $(T,\phi_T)$.
We say that $(T,\phi_T)$ is \keyterm{$\tau$-stationary} if $\cL(\phi_T)\le\cL(\phi_{T'})+\tau$ holds for all relaxed solutions $(T',\phi_{T'})$.
\end{definition}

A relaxed solution $(T,\phi_T)$ that exactly minimizes $\Psi(\phi_T)$ is $0$-stationary, but an iterative method can only approach the minimum: after finitely many iterations, it produces a solution that is $\tau$-stationary for a small $\tau>0$.
A $\tau$-stationary solution has drift-regularized cost within $\tau$ of the infimum of the relaxation: for every relaxed solution $(T',\phi_{T'})$, the function $\theta\mapsto\Psi(\phi_{T_\theta})$ along the mixture with $(T',\phi_{T'})$ is convex, so it lies above its tangent at $\theta=0$, and at $\theta=1$ this gives $\Psi(\phi_{T'})\ge\Psi(\phi_T)-\bigl(\cL(\phi_T)-\cL(\phi_{T'})\bigr)\ge\Psi(\phi_T)-\tau$.
This near-optimality is all we need, because $\tau$ enters the bound on the final cost, derived in the next paragraph, only as an additive error.

\paragraph{Stationarity bounds the alignment.}
The optimal solution is itself a relaxed solution, whose candidates are its centers.
Since its centers are the centroids of its clusters, it has cost $\OPT$, and its weighted means are its centers: the weighted mean of $p\in\widetilde C_j$ is $\widetilde s_j$.
Therefore its linearized cost is
\[
\OPT+2\lambda\sum_{j=1}^k\sum_{p\in\widetilde C_j}(\nu_p-p)^\top(\widetilde s_j-p)=\OPT-2\lambda C_e ,
\]
and $\tau$-stationarity, applied to this single relaxed solution, gives
\begin{equation}\label{eq-overview-stationarity}
C_T+2\lambda\,(C_\nu+C_e)\le\OPT+\tau .
\end{equation}
Equivalently, mixing a small fraction of the optimal solution into $(T,\phi_T)$ moves every $\nu_p$ toward $\widetilde s_j$ and changes the cost at the rate $\OPT-C_T$ and the term $\lambda C_\nu$ at the rate $-2\lambda(C_\nu+C_e)$, and~\eqref{eq-overview-stationarity} states that the total rate is at least $-\tau$.
We multiply~\eqref{eq-overview-comparison} by $\rho$ and use~\eqref{eq-overview-stationarity} together with $\rho/\lambda=1+\rho$; by~\eqref{eq-overview-routing}, the final solution costs at most
\begin{equation}\label{eq-overview-merge}
\begin{split}
C_T+\rho\,\OPT_\pi&\le C_T+\rho\,(C_T+\OPT)+(1+\rho)\bigl(\OPT-C_T-2\lambda C_\nu+\tau\bigr)\\
&=(2\rho+1)\OPT-2\rho\,C_\nu+(1+\rho)\tau .
\end{split}
\end{equation}
The cost $C_T$ of the relaxed solution cancels: in the first line, the first two terms contribute $(1+\rho)C_T$ and the last term contributes $-(1+\rho)C_T$.
The remaining term $-2\rho\,C_\nu$ is nonpositive.
Thus a $\tau$-stationary relaxed solution yields a fair solution with at most $k$ centers and cost at most $(2\rho+1)\OPT+(1+\rho)\tau$.
This explains the factor $3$ of the upper bound: for an exact minimizer of the relaxation ($\tau=0$) and an exact $k$-means algorithm ($\rho=1$), the bound is $3\,\OPT$.
Since an exact minimizer may not be computable in polynomial time, the algorithm instead uses the iterative scheme of Section~\ref{sec-overview-fw} to find a $\tau$-stationary solution for which the additive error $(1+\rho)\tau$ in the bound above is at most $\epsilon\,\OPT$.
In fixed dimension, combining this with the $(1+\epsilon)$-approximation for weighted $k$-means gives the ratio $3+O(\epsilon)$.
The comparison uses only that the optimal solution is fair, has at most $k$ centers, and has each center at the weighted centroid of its points; every fair fractional solution can be brought into this form without increasing its cost, so~\eqref{eq-overview-merge} also holds with $\OPTfrac$ in place of $\OPT$.
The algorithm is stated in Section~\ref{sec-algorithm}; Section~\ref{sec-analysis} proves~\eqref{eq-overview-merge}, with $\OPTfrac$ in place of $\OPT$, for every $\tau$-stationary relaxed solution and shows that the algorithm finds such a solution with $\tau$ small enough.

\paragraph{Why the weight $\rho/(1+\rho)$.}
For an arbitrary weight $\lambda>0$, the same computation bounds the cost of the final solution by
\[
\Bigl(1+\rho-\frac\rho\lambda\Bigr)C_T+\rho\Bigl(1+\frac1\lambda\Bigr)\OPT-2\rho\,C_\nu+\frac{\rho\tau}\lambda ,
\]
and $\lambda=\frac\rho{1+\rho}$ is the weight for which the coefficient of $C_T$ vanishes.
For a smaller weight, this coefficient is negative, and the bound is $\rho(1+\frac1\lambda)\OPT$, up to the term in $\tau$.
For a larger weight, the term $C_T$ has to be bounded; since the optimal solution, viewed as a relaxed solution, has drift equal to its cost $\OPT$, near-optimality gives $C_T\le C_T+\lambda C_\nu\le(1+\lambda)\OPT+\tau$, and the bound becomes $(1+\rho)(1+\lambda)\OPT+(1+\rho)\tau$.
This argument therefore gives the ratio $\max\{\rho(1+\frac1\lambda),(1+\rho)(1+\lambda)\}$: the first bound is decreasing in $\lambda$ and applies below $\frac\rho{1+\rho}$, the second is increasing and applies above it.
The two bounds are equal at $\lambda=\frac\rho{1+\rho}$, where both equal $2\rho+1$, so the maximum is smallest at this weight, and its value there is $2\rho+1$.

\subsection{Solving the Relaxation by Linear Programs}\label{sec-overview-fw}

The relaxation is not finite-dimensional, since a relaxed solution may use any number of candidates anywhere in $\R^d$.
We solve it approximately by an iterative scheme modeled on the Frank--Wolfe method~\cite{frank1956algorithm,jaggi2013revisiting}, which minimizes a convex function over a compact convex set: every step minimizes the linearization of the function at the current point and moves a small step toward this minimizer.
The relaxation has the corresponding structure: the feasible set is closed under mixtures, the objective $\Psi(\phi_T)$ is convex along them, and its directional derivative is given by the linearized cost.
Our scheme nevertheless differs from the method in three essential respects.
First, minimizing the linearized cost over all relaxed solutions is not a linear program; each iteration solves this problem approximately by a linear program over finite candidate sets that are constructed anew for the current solution.
Second, the step is a $\theta$-mixture, which adds the candidates of the new solution to the current ones, so the number of candidates grows with every update.
Third, the convergence guarantees of the method do not apply in this setting, and we bound the number of iterations directly.

Our algorithm starts with an optimal solution of the fair assignment LP of a finite candidate set, whose drift-regularized cost is at most $1+\lambda$ times its cost because $C_\nu\le C_T$.
In each iteration, it computes, by a single linear program, a relaxed solution whose linearized cost is close to the minimum.
If this solution improves the linearized cost by more than a threshold, the algorithm mixes a small fraction of it into the current solution, decreasing $\Psi(\phi_T)$ by at least $\Omega((\epsilon/\rho)^2)$ times the initial cost. Otherwise, the current solution is nearly stationary, and the algorithm stops and clusters the centroids.
Since $\Psi(\phi_T)\ge0$ for every relaxed solution, there are $O((\rho/\epsilon)^2)$ iterations, and~\eqref{eq-overview-merge} with a small $\tau$ gives the ratio $2\rho+1+\epsilon$.
The number of iterations depends on $\rho$ because the stationarity error $\tau$ enters~\eqref{eq-overview-merge} with the factor $1+\rho$: the ratio $2\rho+1+\epsilon$ requires $\tau$ of order $(\epsilon/\rho)\OPT$, so the threshold above must be of this order, and the guaranteed decrease per update is then only of order $(\epsilon/\rho)^2$ times the initial cost.
A weaker $k$-means algorithm therefore requires a more accurate solution of the relaxation.

\paragraph{The linear program of an iteration.}
The linearized cost is not a linear function of $\phi_{T'}$, since it depends on the centroids of the new columns.
Shifting the points brings it into a form that a linear program can handle.
Let
\[
u_p=\lambda\,(\nu_p-p),
\qquad
p_\lambda=p-u_p ,
\]
so that $p_\lambda$ is obtained by moving $p$ backwards by the fraction $\lambda$ of its displacement.
Expanding the squares and applying the centroid identity to each column, as in Lemma~\ref{lem-column-identity} of Section~\ref{sec-algorithm}, shows that each column contributes to $\cL(\phi_{T'})$ the difference of two optimal single-center costs:
\[
\cL(\phi_{T'})=\sum_{t\in T'}\Bigl(\min_{c\in\R^d}\sum_{p\in P}\phi_{T'}(p,t)\lVert p_\lambda-c\rVert^2-\min_{v\in\R^d}\sum_{p\in P}\phi_{T'}(p,t)\lVert u_p-v\rVert^2\Bigr).
\]
The first term is the cost of the column on the shifted points $p_\lambda$ with its best center; once the column has a fixed center $c$, this cost is linear in the column, so, as in the fair assignment LP, the linear program has one column for each point $c$ of a finite candidate set.
The second term has a negative sign, so it is handled differently: if $v$ is restricted to a finite candidate set, then minus the minimum is the maximum of finitely many linear functions of the column, which a linear program expresses by one auxiliary variable per column that is at least each of these functions.
Since the columns are fractional, the candidate sets must contain a near-optimal center for every nonnegative weighting of the points; sets of polynomial size with this property exist for every fixed accuracy.
Each iteration therefore solves a single linear program that contains the fairness constraints of every column, so all iterates are fair.
Algorithm~\ref{alg-main} in Section~\ref{sec-algorithm} states the method, and Section~\ref{sec-computing} proves that it stops after $O((\rho/\epsilon)^2)$ updates at a nearly stationary relaxed solution.

\subsection{\texorpdfstring{$k$}{k}-Sparse Wasserstein Barycenters}\label{sec-overview-wb}

Having described the algorithm for fair $k$-means, we explain why it also solves the $k$-sparse Wasserstein barycenter problem.
As described in Section~\ref{sec-problems}, this problem is a special case of fractional fair $k$-means in which the atoms, weighted by their masses, are the points and every input distribution is a group.
The analysis of Theorem~\ref{thm-main} uses only the centroid identity and inequalities that hold for weighted sums, so it applies to these weighted points without change.
In the resulting algorithm, the centroids of the fragment clusters are clustered into at most $k$ support points, and the fair assignment LP of these points computes their masses and the transport plans; no rounding is needed, and no constraint is violated.

\subsection{The Lower Bound}\label{sec-overview-hardness}

We prove Theorem~\ref{thm-hardness} by a reduction from label cover, a classical NP-hard constraint satisfaction problem that is the usual starting point for tight inapproximability results~\cite{arora1998proof,raz1998parallel,feige1998threshold}.

\begin{definition}[Label cover]\label{def-label-cover}
An instance of \keyterm{label cover} consists of a bipartite multigraph $(U,V,E)$, a label set $\mathcal A$ for the vertices in $U$, a label set $\mathcal B$ for the vertices in $V$, and a map $f_e:\mathcal A\to\mathcal B$ for every edge $e=(u,v)$, where $u\in U$ and $v\in V$.
A \keyterm{labeling} assigns a label $a_u\in\mathcal A$ to every $u\in U$ and a label $b_v\in\mathcal B$ to every $v\in V$, and it \keyterm{satisfies} an edge $e=(u,v)$ if $f_e(a_u)=b_v$.
The \keyterm{value} of the instance is the maximum fraction of edges that a labeling satisfies, where parallel edges are counted separately.
The graph is \keyterm{biregular} if all vertices of $U$ have the same degree and all vertices of $V$ have the same degree.
\end{definition}

By the PCP theorem~\cite{arora1998proof} and the parallel repetition theorem~\cite{raz1998parallel}, for every constant $\sigma>0$ it is NP-hard to distinguish instances of value~$1$ from instances of value at most~$\sigma$, even for biregular graphs and label sets whose sizes depend only on~$\sigma$.

The proof builds a fair $k$-means instance on the line from a label cover instance.
On the line, standard $k$-means is solvable exactly by dynamic programming~\cite{wang2011ckmeans}, so the whole gap has to be created by the fairness constraints.
They are powerful enough: with suitable groups, they can impose any homogeneous linear inequality with integer coefficients on the fractions of the points in a cluster, as the paragraph \emph{Encoding the constraints as fairness groups} explains.
The construction has two parts.

The first part turns a fractional labeling into a cost.
Each label of a vertex $u\in U$ becomes a pair of heavy points at distance~$1$.
We choose their common weight $M$ to be a sufficiently large integer, so that the light points introduced in the second part cost only a small fraction of the total; \eqref{eq-heavy-weight} in Section~\ref{sec-satisfiable} gives its value.
The costs that matter are therefore proportional to $M$.
Vertex $u$ gets one center fewer than its pairs need, so one pair is served by a single center at its midpoint, at cost $\frac M2$; this is the label that $u$ selects.
Vertex $u$ may also split its selection over several labels.
A concentrated choice---one label gets all the weight---costs $\frac M2$; a spread choice costs about $\frac{3M}2$, because each pair that carries part of the choice but keeps both centers pays about twice as much per unit as the selected pair.
The ratio $\frac{3M/2}{M/2}=3$ is the factor of Theorem~\ref{thm-hardness}.

The second part makes spreading necessary when the label cover has small value: light label points tie the choices at the two endpoints of each edge $e$ to the map~$f_e$.
When the label cover is satisfiable, every vertex concentrates consistently, and the total cost is $\frac M2|U|$; when the value is small, most vertices of $U$ must spread, and the cost rises to nearly $\frac{3M}2|U|$.
Figure~\ref{fig-lower} illustrates the construction for a single edge.

\begin{figure}[!htb]
\centering
\begin{tikzpicture}[fig]
\begin{scope}
  \def\xl{2.0}\def\xr{3.4}\def\xm{2.7}
  \node[figlabs,anchor=west,text=figgray] at (4.0,3.3) {cost};
  \begin{scope}[yshift=2.6cm]
    \draw[figaxis] (1.55,0) -- (3.85,0);
    \node[anchor=east] at (1.4,0) {selected};
    \node[figpt] at (\xl,0) {}; \node[figpt] at (\xr,0) {};
    \node[figctr] at (\xm,0) {};
    \node[anchor=west] at (4.0,0) {$\frac12$};
  \end{scope}
  \begin{scope}[yshift=1.35cm]
    \draw[figaxis] (1.55,0) -- (3.85,0);
    \node[anchor=east] at (1.4,0) {not selected};
    \node[figpt] at (\xl,0) {}; \node[figpt] at (\xr,0) {};
    \node[figctr] at (\xl,0) {}; \node[figctr] at (\xr,0) {};
    \node[anchor=west] at (4.0,0) {$0$};
  \end{scope}
  \begin{scope}[yshift=0.1cm]
    \draw[figaxis] (1.55,0) -- (3.85,0);
    \node[anchor=east] at (1.4,0) {fraction $t$};
    \node[figpt] at (\xl,0) {}; \node[figpt] at (\xr,0) {};
    \node[figbigpt] at (\xm,0) {};
    \node[figctr] at (\xl,0) {};
    \node[figctr] at ({\xm+0.7*(1-0.3)/(1+0.3)},0) {};
    \node[figlab,above=4pt] at (\xl,0) {$1{-}t$};
    \node[figlab,above=4pt] at (\xm,0) {$2t$};
    \node[figlab,above=4pt] at (\xr,0) {$1{-}t$};
    \node[anchor=west] at (4.0,0) {$\frac t{1+t}$};
  \end{scope}
  \node at (2.0,-1.0) {(a) label selection};
\end{scope}
\begin{scope}[xshift=5.75cm]
  \foreach \y/\w/\a in {2.75/1.0/1, 1.75/0.4/2, 0.75/0.6/3} {
    \fill[figblue!30] (1.0-\w,\y-0.15) rectangle (1.0,\y+0.15);
    \node[figlab] (a\a) at (1.25,\y) {$a_\a$};
  }
  \foreach \y/\w/\b in {2.25/1.4/1, 0.75/0.6/2} {
    \fill[figblue!30] (2.9,\y-0.15) rectangle (2.9+\w,\y+0.15);
    \node[figlab] (b\b) at (2.65,\y) {$b_\b$};
  }
  \draw[->,figgray,line width=0.6pt] (a1.east) -- (b1.west);
  \draw[->,figgray,line width=0.6pt] (a2.east) -- (b1.west);
  \draw[->,figgray,line width=0.6pt] (a3.east) -- node[figlabs,below,text=figgray] {$f_e$} (b2.west);
  \node[figlabs] at (0.5,3.3) {$\mu^e_u$};
  \node[figlabs] at (3.6,3.3) {$\mu^e_v$};
  \node[figlabs] at (1.95,0.0) {$\mu^e_v=(f_e)_*\mu^e_u$};
  \node at (1.95,-1.0) {(b) edge constraint};
\end{scope}
\begin{scope}[xshift=11.35cm]
  \def\sq{0.27}\def\gap{0.05}
  \foreach \y/\c/\lab/\cl in {2.6/5/x_1/{+2}, 1.95/2/x_2/{-1}, 1.3/2/x_3/{-1}} {
    \node[figlab,anchor=east] at (-0.05,\y) {$\lab$};
    \foreach \i in {1,...,6} {
      \pgfmathsetmacro{\xx}{(\i-1)*(\sq+\gap)}
      \ifnum\i>\c
        \draw[figgray,line width=0.4pt,fill=white] (\xx,\y-\sq/2) rectangle ++(\sq,\sq);
      \else
        \draw[figblue!70,line width=0.4pt,fill=figblue!30] (\xx,\y-\sq/2) rectangle ++(\sq,\sq);
      \fi
    }
    \node[figlabs,anchor=west] at (6*\sq+5*\gap+0.12,\y) {$c_{\lab}=\cl$};
  }
  \draw[figred,line width=0.6pt,densely dashed] ({3*\sq+2.5*\gap},0.95) -- ({3*\sq+2.5*\gap},2.95);
  \node[figlabs,text=figred] at ({3*\sq+2.5*\gap},3.2) {half};
  \node[figlabs,align=center] at (1.25,0.3) {group $\ge\frac12$ of the cluster of $s$\\[1pt]$\iff\ \sum_xc_x\,\phi(x,s)\ge0$};
  \node at (1.25,-1.0) {(c) fairness group};
\end{scope}
\draw[->,figgray,line width=0.7pt] (4.95,1.6) -- (5.5,1.6);
\draw[->,figgray,line width=0.7pt] (10.25,1.6) -- (10.8,1.6);
\end{tikzpicture}
\caption{The reduction for one edge $e=(u,v)$.
(a) The pair of a label of $u$, two heavy points at unit distance (black), and its centers (blue circles), drawn with $M=1$ for readability; the construction uses a large $M$, and all costs are proportional to $M$.
Vertex $u$ selects a label by serving its pair with one center, at cost $\frac12$, and a pair that keeps two centers but carries a fraction $t$ of the choice of $u$ costs at least $\frac t{1+t}\approx t$; in the last row, this fraction of both points is drawn at the midpoint, where it has weight $2t$.
(b) Summed over all clusters, the witnesses form distributions $\mu^e_u$ on the labels of $u$ and $\mu^e_v$ on the labels of $v$, and the label points of $e$ force $\mu^e_v$ to be the image of $\mu^e_u$ under $f_e$.
(c) Every inequality $\sum_xc_x\phi(x,s)\ge0$ that remains after the witnesses are eliminated is the lower bound of one group: each point $x$ is replaced by $2Kw(x)$ copies (squares, drawn for $K=3$ and $w(x)=1$), and the group (shaded) contains $Kw(x)+c_x$ copies of each $x$.}
\label{fig-lower}
\end{figure}

\paragraph{Representing labels by heavy points.}
The fair $k$-means instance lives on a line, with each label occupying its own interval far from all others.
A label $(u,a)$ of a vertex $u\in U$ becomes a \keyterm{pair} of heavy points---two points of weight $M$ at distance~$1$---while a label $(v,b)$ of $v\in V$ is a single heavy point.
Each vertex of $U$ has $|\mathcal A|$ labels and hence $|\mathcal A|$ pairs; serving a pair at zero cost takes two centers, one at each endpoint.
The center budget, however, gives each vertex only $2|\mathcal A|-1$ centers instead of the $2|\mathcal A|$ that the pairs need: the total is $k=(2|\mathcal A|-1)|U|+|\mathcal B||V|$, which also includes one center for each heavy point of $V$.
This deficit of one center per vertex forces at least one pair to be served by a single center.
To \keyterm{select} a label, vertex~$u$ places this single center at the midpoint of the corresponding pair; since each of the two points lies at distance $\frac12$ from the midpoint, the cost is $2M\bigl(\frac12\bigr)^2=\frac M2$, and the other pairs of $u$ receive two centers each and cost nothing; Figure~\ref{fig-lower}(a) shows this configuration, drawn with $M=1$ for readability.
Thus selecting a label costs $\frac M2$ per vertex of~$U$, and this placement is optimal: with only $2|\mathcal A|-1$ centers, some pair of $u$ has at most one center, and such a pair costs at least $\frac M2$.
This low cost, however, requires that the whole deficit of one center fall on a single pair.

\paragraph{The cost of spreading.}
When the deficit is spread over several pairs, the cost rises sharply.
The consistency constraints of the next paragraph assign a fraction of the selection to each label of~$u$; these fractions form a probability distribution on the labels, the \keyterm{choice} of~$u$.
If the pair of a label carries a fraction $t$ of the choice, then a fraction $t$ of both its points is committed to the selection and served by centers shared with other pairs, while the remaining fraction $1-t$ of each point stays to be served by the pair's own two centers.
Since $u$ has one fewer center than its pairs need, at least one pair has a single center, and that pair costs at least $\frac M2$ as before.
For a concentrated choice---one label gets nearly all the weight---this single pair accounts for nearly all of the cost.
For a spread choice, however, the pairs that keep two centers carry fractions summing to almost~$1$, and each of them incurs an additional cost.
Consider such a pair carrying fraction $t$.
A fraction $t$ of both its endpoints---altogether weight $2Mt$---sits effectively at the midpoint between the two points, while weight $M(1-t)$ remains at each endpoint.
The pair's two centers must serve these three loads: $M(1-t)$ at each end and $2Mt$ in the middle.
Placing them optimally costs at least $\frac{Mt}{1+t}\approx Mt$ for small~$t$ (Figure~\ref{fig-lower}(a)); Lemma~\ref{lem-pair-cost} in Section~\ref{sec-small-value} proves this bound.
The two halves of this cost have simple sources.
The fraction $t$ of both points pays $\frac{Mt}2$ for its spread around the midpoint, exactly as the selected pair does.
In addition, the pair still needs its two centers near its endpoints for the remaining weight $M(1-t)$ at each end, so the weight $2Mt$ at the midpoint is served from distance about $\frac12$, which costs another $2Mt\cdot\frac14=\frac{Mt}2$.
The selected pair avoids this second half because its single center sits at the midpoint.
Since the fractions of all pairs with two centers sum to almost~$1$, their costs add up to about $M$ times this sum, that is, to about $M$, and together with $\frac M2$ from the single-center pair, vertex~$u$ pays roughly $\frac M2+M=\frac{3M}2$, three times the concentrated cost.
So far, however, nothing prevents every vertex from concentrating on a single label and paying only $\frac M2$.

\paragraph{Forcing consistency along edges.}
The edge maps $f_e$ provide the missing pressure.
We add light points together with constraints that record, cluster by cluster, how each vertex distributes its selection among its labels, and we use the maps $f_e$ to force these records to agree across every edge.
At the origin we place a \keyterm{selector} $x_u$ or $x_v$ for each vertex and a \keyterm{label point} $x_{e,b}$ for each edge $e$ and each $b\in\mathcal B$, all of weight~$1$, which is negligible next to~$M$.
As before, $\phi(x,s)$ denotes the fraction of a point $x$ assigned to a center $s$.

Fix an edge $e=(u,v)$ and the cluster of a center $s$.
The constraint for $e$ asks that the fraction $\phi(x_u,s)$ of the selector of $u$ in this cluster can be split among the labels of $u$, in such a way that the share of each label $a$ is at most the fraction of each point of its pair in the cluster.
To write this condition as linear constraints, we introduce auxiliary variables $\chi_{u,a}\ge0$ for the shares, with $\sum_a\chi_{u,a}=\phi(x_u,s)$, and call them the \keyterm{witnesses} of $u$ in the cluster: $\chi_{u,a}$ is the part of the selector's presence in the cluster that is explained by the pair of $(u,a)$.
We do not include the witnesses in a solution, which consists only of the centers and the fractions $\phi(x,s)$; the constraint asks only that suitable witnesses exist, and they may differ between clusters, and also between the edges at $u$, since each edge has its own constraint.
The witnesses $\chi_{v,b}\ge0$ of $v$ split $\phi(x_v,s)$ in the same way, each bounded by the fraction of the heavy point of $(v,b)$ in the cluster.
The label points tie the two sides together: for every $b\in\mathcal B$, the constraint also requires
\[
\phi(x_{e,b},\,s)\;=\;\sum_{a:\,f_e(a)=b}\chi_{u,a}\;+\;\sum_{b'\ne b}\chi_{v,b'}.
\]

To see what this achieves, sum the witnesses of $e$ over all clusters.
The total share of a label $a$ of $u$ is a number $\mu^e_u(a)$, and these numbers form a probability distribution $\mu^e_u$ on the labels of $u$, because the selector $x_u$ is fully assigned; in the same way, the witnesses of $v$ give a distribution $\mu^e_v$ on the labels of $v$.
The superscript $e$ indicates that both distributions come from the witnesses of the edge $e$.
Summing the last equation over all clusters, we obtain $1$ on the left, since each label point is fully assigned, and $\sum_{a:\,f_e(a)=b}\mu^e_u(a)+1-\mu^e_v(b)$ on the right.
Hence $\mu^e_v(b)=\sum_{a:\,f_e(a)=b}\mu^e_u(a)$ for every~$b$, that is, $\mu^e_v$ is the image of $\mu^e_u$ under the map $f_e$: the distribution of labels on the $V$~side is completely determined by the distribution on the $U$~side and the map $f_e$ (Figure~\ref{fig-lower}(b)).
For an integral labeling, where the selector of $u$ is recorded only at a label $a$ and that of $v$ only at a label $b'$, the right-hand side summed over all clusters is $\ind{f_e(a)=b}+\ind{b'\ne b}$, which equals $1$ for every $b$ exactly when $f_e(a)=b'$; so the label points of $e$ check that $e$ is satisfied.
The vertices of $V$ pay nothing; they matter because $x_v$ is a single point, so in a solution of low cost the distributions $\mu^e_v$ for the different edges $e$ at $v$ are all close to a single distribution, determined by the clusters that contain $x_v$, just as $v$ has a single label in a labeling; Lemma~\ref{lem-canonical} in Section~\ref{sec-small-value} makes this precise.

\paragraph{Encoding the constraints as fairness groups.}
The constraints above are not yet fairness constraints: they ask that witnesses exist, whereas a fairness constraint is a linear inequality in the fractions $\phi(x,s)$ alone.
We therefore remove the witnesses by Fourier--Motzkin elimination~\cite[Section~12.2]{schrijver1986theory}.
This method eliminates the variables of a system of linear inequalities one at a time: it keeps the inequalities that do not contain the variable, and it combines every inequality that bounds the variable from above with every inequality that bounds it from below, adding positive multiples of the two so that the variable cancels.
The resulting system holds for given values of the remaining variables if and only if the original system has a solution in the eliminated variable.
Since the label sets have constant size, eliminating all witnesses of an edge produces a constant number of homogeneous linear inequalities $\sum_x c_x\phi(x,s)\ge0$ with integer coefficients, which must hold in every cluster.
By Lemma~\ref{lem-consistency-linear} in Section~\ref{sec-intermediate}, there is a positive integer $K$, depending only on $|\mathcal A|$ and $|\mathcal B|$, that bounds both the number of these inequalities for each edge and the absolute values $|c_x|$ of their coefficients.
Each such inequality can be realized as the lower bound of a fairness group, as follows.
Replace every point $x$ of weight $w(x)$ by $2Kw(x)$ unit-weight copies, and form a group containing $Kw(x)+c_x$ copies of each $x$, with lower bound $\beta_i=\frac12$ and upper bound $\alpha_i=1$; this number of copies lies between $0$ and $2Kw(x)$ because $|c_x|\le K$ and $w(x)\ge1$.
When all copies of a point share the same cluster fractions, this group has $\sum_x(Kw(x)+c_x)\phi(x,s)$ members in a cluster of total weight $\sum_x 2Kw(x)\phi(x,s)$; the ratio is at least $\beta_i=\frac12$ exactly when $\sum_x c_x\phi(x,s)\ge0$ (Figure~\ref{fig-lower}(c)).
Additional groups with $\beta_i=\alpha_i=\frac12$ enforce that all copies of a point do share the same fractions: each such group requires every cluster to contain the same proportion of every copy.
The upshot is that every fair solution---fractional or integral---produces, for every edge $e=(u,v)$, distributions $\mu^e_u$ and $\mu^e_v$ such that $\mu^e_v$ is the image of $\mu^e_u$ under $f_e$.

\paragraph{Completeness and soundness.}
Whether these distributions can be concentrated determines the cost.
If the label cover has value~$1$, a labeling satisfies all edges; letting every vertex of $U$ select the label assigned by this labeling gives a concentrated choice that satisfies all consistency constraints, and each vertex pays $\frac M2$.
If the label cover has value at most~$\sigma$, no consistent concentrated selection can exist.
To see why, suppose most vertices have concentrated choices: vertex $u$ puts most of the weight of $\mu^e_u$ on a single label $a_u$, independently of the edge $e$.
Since $\mu^e_v$ is the image of $\mu^e_u$ under $f_e$, it must then concentrate on $f_e(a_u)$, so vertex $v$ must put most of its weight on $f_e(a_u)$ for every edge $e$ incident to~$v$.
A random labeling that assigns each vertex its most likely label would then satisfy most edges, contradicting the assumption that the value is at most~$\sigma$.
Hence many vertices of $U$ are forced to spread their choices, and by the paragraph \emph{The cost of spreading}, each of them pays close to $\frac{3M}2$.
Proposition~\ref{prop-small-value} in Section~\ref{sec-small-value} carries out the accounting: for $\sigma=\delta^2/128$, the total cost is at least $\bigl(\frac32-\frac\delta8\bigr)M|U|$.
Because the argument uses only the fractions of the assignment, the bound holds even for fractional solutions, and the cost is compared with the integral optimum $\OPT$: no algorithm can return a fractional solution of cost less than $(3-\delta)\,\OPT$.
Section~\ref{sec-hardness} gives the full proof.
\section{Related Work}\label{sec-related}

\paragraph{Euclidean $k$-means.}
Standard $k$-means is already NP-hard in the plane if $k$ is part of the input~\cite{DBLP:journals/tcs/MahajanNV12}, and it is APX-hard in high dimension~\cite{awasthi2015hardness,lee2017improved}; stronger lower bounds hold under the Johnson coverage hypothesis~\cite{cohenaddad2022johnson}.
On the positive side, local search gives a PTAS in every fixed dimension~\cite{cohen2019local,friggstad2019local}; we use the version of~\cite{friggstad2019local}, which also applies to weighted points, as our subroutine with $\rho=1+\epsilon$. Recent work gives almost matching upper and lower bounds on the running time of such schemes~\cite{cohenaddad2026almost}.
In high dimension, the approximation ratio has been improved in a series of works~\cite{ahmadian2020better,grandoni2022refined,cohenaddad2022improved,charikar2025improved}.
Recently, \cite{charikar2026nonmonotone} obtained a $(4+\epsilon)$-approximation by non-monotone dual fitting, and \cite{byrka2026kclustering} matched this ratio by iterative randomized rounding.
\cite{anand2026spectral} further improved the ratio to $3+\ln2+\epsilon$ via a spectral analysis.
Since our algorithm calls the $k$-means subroutine only once, every improvement of the ratio $\rho$ for weighted $k$-means directly improves the bound of Theorem~\ref{thm-main}.

\paragraph{Fair clustering.}
The group-balance model was introduced in~\cite{chierichetti2017fair} for two groups, along with the fairlet decomposition. Bera, Chakrabarty, Flores, and Negahbani~\cite{bera2019fair} extended it to overlapping groups with lower and upper bounds, which is the model we study; their rounding procedure gives Corollary~\ref{cor-overlapping}.
For the continuous Euclidean setting, Song, Mo, and Ding~\cite{song2025relax} obtained the ratio $1+4\rho+\epsilon$ for fair $k$-means and $k$-sparse Wasserstein barycenters by solving the fair assignment LP of an approximate centroid set and clustering the centroids of the resulting fractional clusters; their rounding procedure for non-overlapping groups gives Corollary~\ref{cor-partition}.
Subsequent work on this model includes algorithms for multiple protected attributes~\cite{bohm2021algorithms} and algorithms for fair $k$-median that satisfy balance constraints without any violation~\cite{dai2022fair,wu2024approximation,bandyapadhyay2026polynomial}.
Coresets for fair clustering are studied in~\cite{huang2019coresets,schmidt2020fair,braverman2022power,DBLP:journals/jcss/BandyapadhyayFS24}. A different approach by~\cite{kim2025fair} aligns two groups via a transport plan and clusters the aligned data.
Other fairness notions include proportional fairness~\cite{chen2019proportionally,micha2020proportionally,kellerhals2024proportional}, socially fair clustering~\cite{ghadiri2021socially,chlamtac2022approximating} and individual fairness~\cite{vakilian2022improved}.
Group balance has also been studied in streaming and uncertain settings. In the sliding-window model, \cite{cohenaddad2025sliding} show that exact fairness with any bounded approximation factor requires linear space, though small coresets exist if the bounds may be violated multiplicatively. Under uncertain group memberships, \cite{duppala2025robust} give robust guarantees for fair $k$-center.
The balance requirement has been combined with quotas on the chosen centers~\cite{dickerson2023doubly}; for disjoint groups in the discrete setting, \cite{funk2026constant} obtain constant-factor approximations for $k$-median and $k$-means that violate the balance bounds by at most $2$.
More generally, clustering under constraints that break the locality property has been studied by Ding and Xu~\cite{ding2020unified} and by Bhattacharya, Jaiswal, and Kumar~\cite{bhattacharya2018faster}.

\paragraph{Hardness of clustering.}
For standard $k$-means, hardness of approximation is known only in high dimension~\cite{awasthi2015hardness,lee2017improved,cohenaddad2022johnson}, whereas fixed dimension admits a PTAS.
For fair $k$-center under a different fairness notion---where a prescribed number of the $k$ centers comes from each group, rather than requiring balance within each cluster---a threshold of $3$ is known: the $3$-approximation of~\cite{thejaswi2026hardness} is optimal unless $\mathrm P=\mathrm{NP}$, even for two groups.
Although this threshold coincides numerically with ours (Theorem~\ref{thm-hardness}), the two results address different objectives ($k$-center vs.\ $k$-means), different fairness models (center selection vs.\ cluster balance), and different proof techniques (the reduction of~\cite{thejaswi2026hardness} exploits metric structure, whereas ours encodes label cover via assignment constraints).

\paragraph{Wasserstein barycenters.}
Agueh and Carlier~\cite{agueh2011barycenters} introduced barycenters in the Wasserstein space; see~\cite{peyre2019computational} for their computation and applications.
For discrete distributions, a barycenter can be computed in polynomial time in every fixed dimension~\cite{altschuler2021wasserstein}, but the problem is NP-hard when the dimension is part of the input~\cite{altschuler2022wasserstein}.
Every discrete instance has a barycenter whose support size is at most the total support size of the input distributions~\cite{anderes2016discrete}; Borgwardt~\cite{borgwardt2022lp} computes a $2$-approximate barycenter with a support of this size by linear programming, and Wang and Xie~\cite{wang2025closing} give a polynomial-time approximation scheme when the support size is not restricted.
Deciding whether a barycenter with at most a given number of support points exists is NP-hard, even for three distributions in the plane~\cite{borgwardt2021computational}.
When the support size is limited to a prescribed number $k$, as in the applications of Section~\ref{sec-intro}, the problem becomes a clustering problem.
Free-support methods~\cite{cuturi2014fast,claici2018stochastic,ye2017fast} alternate between moving the support points and updating the weights and transport plans.
Yang and Ding~\cite{yang2024approximate} obtained the ratio $(2+\sqrt\rho)^2$, and Song, Mo, and Ding~\cite{song2025relax} improved it to $1+4\rho$.
In the full version of their paper, Wang and Xie~\cite{wang2026closingfull} show that their candidate set also suits this problem: restricting the support to it increases the optimal $k$-sparse cost by a factor of at most $1+1/t$ for every integer $t\ge1$.
With the bound $k$, however, the restricted problem is no longer a linear program but a mixed-integer program that has to be solved to optimality, so this result does not give a polynomial-time algorithm for the $k$-sparse problem.
For a fixed support, the weights and transport plans are obtained from a linear program or by related methods~\cite{cuturi2014fast,cuturi2016smoothed,lin2020fixed}.

\section{Preliminaries}\label{sec-prelim}

The algorithm works with weighted point sets: it runs the weighted $k$-means algorithm on the weighted centroids of the fragment clusters, and Section~\ref{sec-wb} applies it to distributions with arbitrary masses.
This section therefore restates for weighted point sets the definitions of Section~\ref{sec-problems} that involve the weights, and collects the tools that we use.
The input points form an indexed set $P$, so that points at the same location are distinguished.
In a \keyterm{weighted} point set, every point $q$ has a weight $w(q)\ge0$; all input points have weight $1$, except in Section~\ref{sec-wb}.
For $z,z'\in\R^d$, $\lVert z-z'\rVert$ is the Euclidean distance and $z^\top z'$ the standard inner product.
For a finite set $Q\subset\R^d$ and a point $z\in\R^d$, we write $\dist(z,Q)=\min_{q\in Q}\lVert z-q\rVert$ for the distance from $z$ to $Q$, and $\cN(z,Q)$ for a point of $Q$ nearest to $z$; if there are several such points, we choose one arbitrarily.
Finally, $\Cen(Q)$ is the centroid of $Q$, weighted by $w$ if $Q$ is weighted.

\paragraph{Assignments, costs and fairness.}
Assignment matrices of a weighted point set $Q$ to a set $S$ of centers, and integral and fractional assignment matrices, are defined as in Section~\ref{sec-problems}, with $Q$ in place of $P$.
The weights enter only the cost and the weights of the clusters.
The cost of $\phi_S$ is
\begin{equation}\label{eq-assignment-cost}
\Cost(Q,S,\phi_S)=\sum_{q\in Q}\sum_{s\in S}w(q)\,\phi_S(q,s)\,\lVert q-s\rVert^2 .
\end{equation}
Without further constraints, it is optimal to assign every point to a nearest center, and we write
\[
\Cost(Q,S)=\min_{\phi_S}\Cost(Q,S,\phi_S)=\sum_{q\in Q}w(q)\dist^2(q,S)
\]
for the resulting cost; we use this two-argument notation for standard $k$-means.
For $Q=P$, the total weight assigned to a center $s\in S$ and the weight of the group $P^{(i)}$ assigned to $s$ are
\begin{equation}\label{eq-cluster-weight}
w(s)=\sum_{p\in P}w(p)\,\phi_S(p,s),
\qquad
w^{(i)}(s)=\sum_{p\in P^{(i)}}w(p)\,\phi_S(p,s).
\end{equation}
When every weight is $1$, \eqref{eq-assignment-cost} and~\eqref{eq-cluster-weight} are the cost and the cluster weights of Section~\ref{sec-problems}.

\begin{definition}[Fairness and additive violation]\label{def-fair}
An assignment matrix $\phi_S$ of $P$ to $S$ is $(\alpha,\beta)$-fair if
\begin{equation}\label{eq-fairness}
\beta_i\,w(s)\le w^{(i)}(s)\le\alpha_i\,w(s)\qquad\text{for all }s\in S\text{ and }i\in[m],
\end{equation}
and it has \keyterm{additive violation} at most $\xi\ge0$ if every bound holds up to an additive error of $\xi$:
\begin{equation}\label{eq-violation}
\beta_i\,w(s)-\xi\le w^{(i)}(s)\le\alpha_i\,w(s)+\xi\qquad\text{for all }s\in S\text{ and }i\in[m].
\end{equation}
\end{definition}

For unit weights, \eqref{eq-fairness} is the fairness condition of Section~\ref{sec-problems}.
The parameter $\Delta$ of Corollary~\ref{cor-overlapping}, the maximum number of groups that contain a common point, is
\[
\Delta=\max_{p\in P}\bigl|\{i\in[m]:p\in P^{(i)}\}\bigr| .
\]

When the center set $S$ is fixed, both the cost~\eqref{eq-assignment-cost} and the constraints~\eqref{eq-fairness} are linear in the entries $\phi_S(p,s)$, so a minimum-cost fair fractional assignment to $S$ reduces to a linear program~\cite{song2025relax}.

\begin{definition}[Fair assignment LP]\label{def-fair-lp}
For a finite set $S\subset\R^d$, the \keyterm{fair assignment LP} of $S$ is the linear program that minimizes $\Cost(P,S,\phi_S)$ over all $(\alpha,\beta)$-fair fractional assignment matrices $\phi_S$ of $P$ to $S$.
We denote by $\phi^*_S$ an optimal solution of it.
\end{definition}

\paragraph{Feasibility and the two optima.}
As explained in Section~\ref{sec-results}, the constraints~\eqref{eq-fairness} of each center are homogeneous linear inequalities in its column, so a sum of fair columns is fair.
Consequently, an instance has a fair solution if and only if assigning all of $P$ to a single center is fair, and in this case it also has an integral fair solution.
We assume throughout that the instance has a fair solution.
As in Section~\ref{sec-overview}, we fix an optimal solution of the integral version of fair $k$-means (Definition~\ref{prob-fair}), with center set $\Sopt=(\widetilde s_1,\dots,\widetilde s_k)$ and integral assignment $\phi_{\Sopt}$, so that
\begin{equation}\label{eq-opt}
\OPT=\Cost(P,\Sopt,\phi_{\Sopt}).
\end{equation}
Its clusters $\widetilde C_j=\{p\in P:\phi_{\Sopt}(p,\widetilde s_j)=1\}$ partition $P$.
Moving a center to the centroid of its cluster does not change the clusters (which are defined by their indices) or their fairness, and the cost does not increase. We may therefore assume that every center $\widetilde s_j$ is the centroid of its cluster $\widetilde C_j$, so that $\OPT=\sum_j\sum_{p\in\widetilde C_j}\lVert p-\widetilde s_j\rVert^2$.
Recall that $\OPTfrac\le\OPT$ is the infimum of $\Cost(P,S,\phi_S)$ over all center sets with $|S|\le k$ and all $(\alpha,\beta)$-fair fractional assignments $\phi_S$.

\paragraph{Centroid identity.}
Consider a finite weighted indexed set $Q$ with positive total weight $w(Q)=\sum_{q\in Q}w(q)$, and let $\Cen(Q)=\frac1{w(Q)}\sum_{q\in Q}w(q)\,q$ be its weighted centroid.
Then every $z\in\R^d$ satisfies
\begin{equation}\label{eq-centroid-identity}
\sum_{q\in Q}w(q)\lVert q-z\rVert^2=\sum_{q\in Q}w(q)\lVert q-\Cen(Q)\rVert^2+w(Q)\lVert\Cen(Q)-z\rVert^2 .
\end{equation}
We refer to this well-known fact as the \keyterm{centroid identity}: serving $Q$ by a center at $z$ instead of its centroid increases the cost by exactly $w(Q)\lVert\Cen(Q)-z\rVert^2$.

\paragraph{Approximate centroid sets.}
Under fairness constraints, the clusters of an optimal solution need not be Voronoi cells, so a center can be the centroid of an arbitrary subset of $P$; our algorithm therefore takes its candidate centers from a finite set that approximates the centroid of every subset.

\begin{definition}[Approximate centroid set~\cite{matouvsek2000approximate}]\label{def-acs}
Let $\delta>0$.
A set $\CS_\delta(P)\subset\R^d$ is a \keyterm{$\delta$-approximate centroid set} of $P$ if for every subset $Q\subseteq P$ there is a point $v\in\CS_\delta(P)$ such that
\begin{equation}\label{eq-approximate-centroid-set}
\lVert v-\Cen(Q)\rVert\le\frac\delta3\sqrt{\frac1{|Q|}\sum_{q\in Q}\lVert q-\Cen(Q)\rVert^2} .
\end{equation}
\end{definition}

The bound in~\eqref{eq-approximate-centroid-set} is $\delta/3$ times the root mean square distance between the points of $Q$ and their centroid, so a concentrated subset needs a candidate very close to its centroid (Figure~\ref{fig-centroid-set}).
For fixed $d$ and $\delta\in(0,1]$, the algorithm of Matou\v{s}ek~\cite{matouvsek2000approximate} computes a $\delta$-approximate centroid set with $O(|P|\delta^{-d}\log(2/\delta))$ points in time $O(|P|\log|P|+|P|\delta^{-d}\log(2/\delta))$.

\begin{figure}[!htb]
\centering
\begin{minipage}[c]{0.34\textwidth}
\centering
\begin{tikzpicture}[fig,x=1.25cm,y=1.25cm]
\tikzset{figcand/.style={circle,fill=figgray!75,draw=none,inner sep=0pt,minimum size=2.2pt}}
\foreach \i in {0,...,5} \foreach \j in {0,...,4} \node[figcand] at ({-0.1+0.6*\i},{-0.25+0.55*\j}) {};
\foreach \q in {(0.2,0.9),(1.4,1.8),(2.5,1.1),(1.9,-0.1),(0.6,-0.3)} \node[figpt] at \q {};
\draw[figred,line width=0.6pt,densely dashed] (1.32,0.68) circle (0.382);
\draw[figred,line width=0.8pt] (1.32-0.06,0.68-0.06) -- (1.32+0.06,0.68+0.06) (1.32-0.06,0.68+0.06) -- (1.32+0.06,0.68-0.06);
\node[circle,draw=figblue,line width=0.8pt,inner sep=0pt,minimum size=6pt] at (1.1,0.85) {};
\node[figlab,text=figblue] at (0.86,1.1) {$v$};
\node[figlab,text=figred,anchor=west] at (1.75,0.62) {$\Cen(Q)$};
\end{tikzpicture}
\end{minipage}\hfill
\begin{minipage}[c]{0.62\textwidth}
\caption{The points of a set $Q$ (black), their centroid $\Cen(Q)$ (red cross), and an approximate centroid set (gray).
The point $v$ of this set lies in the ball around $\Cen(Q)$ whose radius is $\frac\delta3$ times the root mean square distance of the points of $Q$ to their centroid (dashed), so it satisfies~\eqref{eq-approximate-centroid-set}.}
\label{fig-centroid-set}
\end{minipage}
\end{figure}

\paragraph{Candidate sets for weighted centroids.}
Our algorithm produces fractional clusters, so their centroids are weighted means with arbitrary nonnegative weights, which are not covered by~\eqref{eq-approximate-centroid-set}.

\begin{definition}[Candidate set for weighted centroids]\label{def-wcs}
Let $Q\subset\R^d$ be a finite indexed set and $\delta\in(0,1]$.
A finite set $W\subset\R^d$ is called a \keyterm{candidate set for weighted centroids} of $Q$, and $\delta$ is called its \keyterm{accuracy}, if every choice of nonnegative weights $\mu(q)$, $q\in Q$, admits a point $t\in W$ with
\begin{equation}\label{eq-weighted-candidates}
\sum_{q\in Q}\mu(q)\lVert q-t\rVert^2\le(1+\delta)\min_{c\in\R^d}\sum_{q\in Q}\mu(q)\lVert q-c\rVert^2 .
\end{equation}
We write $\WCS_\delta(Q)$ for a candidate set for weighted centroids of $Q$ with accuracy $\delta$.
\end{definition}

The right-hand side of~\eqref{eq-weighted-candidates} is minimized by the $\mu$-weighted centroid of $Q$, and by the centroid identity, a point $t$ satisfies~\eqref{eq-weighted-candidates} if and only if its squared distance to this centroid is at most $\delta$ times the $\mu$-weighted variance of $Q$.
An approximate centroid set as in Definition~\ref{def-acs} only has to approximate centroids of subsets, that is, weighted centroids with weights $0$ and $1$, and it bounds the distance instead of the squared distance; for such weights, a $\delta$-approximate centroid set satisfies~\eqref{eq-weighted-candidates} with accuracy $\delta^2/9$.
We use Definition~\ref{def-acs} for every finite indexed set in place of $P$.

Our constructions use an averaging argument of Inaba, Katoh, and Imai~\cite{inaba1994applications}: if $r$ points of $Q$ are drawn independently with probabilities proportional to $\mu$, the expected squared distance of their mean to the $\mu$-weighted centroid is $1/r$ times the $\mu$-weighted variance of $Q$, so the mean of some multiset of $r$ points of $Q$ is at least this close.
The candidate sets of Lemma~\ref{lem-candidates} contain, for every multiset of $r$ points of $Q$, its mean or a point close to it, so they contain a good point for every choice of $\mu$ although they are constructed deterministically.
Since a multiset may contain a point several times, we also use the indexed set $Q^{[r]}$ that consists of $r$ copies of every point of $Q$; the mean of every multiset of $r$ points of $Q$ is the centroid of a subset of $Q^{[r]}$.

\begin{lemma}[Candidate sets for weighted centroids]\label{lem-candidates}
Let $\delta\in(0,1]$, and let $Q\subset\R^d$ be a finite indexed set with $n$ points.
\begin{enumerate}
\item \emph{(Fixed dimension.)} Let $d$ be fixed and $r=\lceil4/\delta\rceil$. Applying the algorithm of Matou\v{s}ek~\cite{matouvsek2000approximate} with accuracy $\sqrt\delta$ to $Q^{[r]}$ gives, in time $O\bigl(n\delta^{-1}\log(n/\delta)+n\delta^{-1-d/2}\log(2/\delta)\bigr)$, a set with $O(n\delta^{-1-d/2}\log(2/\delta))$ points that satisfies~\eqref{eq-weighted-candidates}.
\item \emph{(Arbitrary dimension.)} Let $r=\lceil1/\delta\rceil$. The set that contains the mean of every multiset of $r$ points of $Q$ satisfies~\eqref{eq-weighted-candidates}; its size is at most $n^r$, and it can be computed in $O(n^r rd)$ time.
\end{enumerate}
\end{lemma}

Appendix~\ref{app-candidate} contains the proof.
From now on, we fix the candidate set $\WCS_\delta(Q)$ of Definition~\ref{def-wcs}: it is the set of part~1 if the dimension $d$ is a constant and the set of part~2 if $d$ is part of the input.
The size of the set of part~1 is linear in $n$ but exponential in $d$, so its size is not polynomial when $d$ is part of the input; the set of part~2 has at most $n^{\lceil1/\delta\rceil}$ points, which is polynomial for fixed $\delta$ in every dimension.
Part~2 alone would therefore suffice for a polynomial running time; we use part~1 in fixed dimension because its linear size gives the running time of Proposition~\ref{prop-running-time}.
In both cases, $\WCS_\delta(Q)$ satisfies~\eqref{eq-weighted-candidates} for all weights simultaneously, so the set is fixed before any weights are chosen.

\paragraph{Weighted $k$-means.}
Given a weighted point set $Q$, a \keyterm{$\rho$-approximate weighted $k$-means algorithm} ($\rho\ge1$) outputs at most $k$ centers $S$ whose cost $\Cost(Q,S)$ is at most $\rho$ times the optimum $\min_{|S'|\le k}\Cost(Q,S')$.
In fixed dimension, $\rho$ can be made arbitrarily close to $1$.

\begin{lemma}[PTAS for weighted $k$-means, \cite{friggstad2019local}]\label{lem-ptas}
For every fixed $d$ and every $\epsilon'\in(0,1]$, there is a deterministic polynomial-time $(1+\epsilon')$-approximation algorithm for weighted $k$-means in $\R^d$.
\end{lemma}

The algorithm of Friggstad, Rezapour, and Salavatipour~\cite{friggstad2019local} is a local search that repeatedly swaps a bounded number of centers, where the bound depends only on $\epsilon'$ and $d$; they prove that it is a PTAS in every fixed dimension and show that their analysis extends to weighted points.
Cohen-Addad, Klein, and Mathieu~\cite{cohen2019local} independently obtained a local search PTAS for unweighted points.
Both algorithms choose the centers from a finite candidate set, for which one can take $\WCS_{\epsilon'/4}(Q)$.

\paragraph{Roadmap of the proofs.}
With the definitions and tools in place, we outline how the remaining sections build the proofs.
\begin{itemize}
\item Section~\ref{sec-algorithm} recalls the relaxed solutions, the drift-regularized cost $\Psi(\phi_T)$ and the linearized cost $\cL(\phi_{T'})$ of Section~\ref{sec-overview}, states the algorithm, and proves two identities: for the $\theta$-mixture $(T_\theta,\phi_{T_\theta})$ of two relaxed solutions, $\Psi(\phi_{T_\theta})$ is a convex quadratic function of $\theta$ whose slope is given by the linearized cost (Lemma~\ref{lem-mixing}), and $\cL(\phi_{T'})$ splits into one term per column of $\phi_{T'}$ (Lemma~\ref{lem-column-identity}).
The first identity justifies the Frank--Wolfe steps of Algorithm~\ref{alg-main}, and the second turns each step into a linear program.
\item Section~\ref{sec-analysis} proves Theorem~\ref{thm-main} in two steps.
The first step shows that every $\tau$-stationary relaxed solution leads to a fair solution of cost at most $(2\rho+1)\OPTfrac+(1+\rho)\tau$: rerouting the columns to the computed centers costs $C_T$ plus the $k$-means cost of the centroids (Lemma~\ref{lem-fair-routing}), comparison with a fair solution bounds this $k$-means cost up to the alignment $C_e$ (Lemma~\ref{lem-comparison}), stationarity bounds $C_e$ (Lemma~\ref{lem-stationarity}), and Lemma~\ref{lem-merge-bound} combines the three bounds.
The second step shows, using the identities of Section~\ref{sec-algorithm}, that Algorithm~\ref{alg-main} finds such a solution with $\tau$ small enough: the linear program of an iteration minimizes $\cL(\phi_{T'})$ up to a small error (Lemma~\ref{lem-step-lp}), every update decreases $\Psi(\phi_T)$ by a fixed amount (Lemma~\ref{lem-descent}), and hence the algorithm stops after $O((\rho/\epsilon)^2)$ updates at a nearly stationary solution (Lemma~\ref{lem-termination}).
Section~\ref{sec-proof-main} combines the two steps, and Proposition~\ref{prop-running-time} bounds the running time.
\item Section~\ref{sec-hardness} proves Theorem~\ref{thm-hardness}.
It first builds an instance on the line with constraints on the clusters that are linear inequalities in the columns (Lemma~\ref{lem-consistency-linear}), then shows that this instance has cost about $\frac M2$ per vertex of $U$ if the label cover instance is satisfiable (Lemma~\ref{lem-satisfiable}) and almost three times as much if its value is small (Proposition~\ref{prop-small-value}), and finally replaces the constraints by groups without changing the costs (Lemma~\ref{lem-encoding}, proved in Appendix~\ref{app-groups}).
\item Section~\ref{sec-wb} shows that transport plans to a distribution with at most $k$ support points correspond to fair fractional assignments (Lemma~\ref{lem-wb-correspondence}), so that Theorem~\ref{thm-main} for weighted points yields Theorem~\ref{thm-wb}.
\end{itemize}

\section{The Algorithm}\label{sec-algorithm}

Algorithm~\ref{alg-main} has three stages.
The first stage solves the drift-regularized relaxation approximately.
It starts with an optimal solution of the fair assignment LP of a candidate set, and then runs the iterative scheme of Section~\ref{sec-overview-fw}, which is modeled on the Frank--Wolfe method and solves one linear program per iteration, until the current fair fractional clustering is nearly stationary.
In the second stage, each candidate is moved to the centroid of its fractional cluster.
The third stage clusters these weighted centroids using the given $k$-means algorithm and returns a minimum-cost fair assignment to the resulting centers.
This section proves the two identities on which the first stage is based and states the algorithm; Section~\ref{sec-analysis} proves that Algorithm~\ref{alg-main} has the guarantees of Theorem~\ref{thm-main}.

\paragraph{Notation recalled from Section~\ref{sec-overview-relaxation}.}
To make this section self-contained, we list the objects that it uses; their motivation is explained in Section~\ref{sec-overview}.
A \keyterm{relaxed solution} $(T,\phi_T)$ consists of a finite indexed set $T\subset\R^d$ of \keyterm{candidates} and an $(\alpha,\beta)$-fair fractional assignment $\phi_T$ of $P$ to $T$.
The column of a candidate $t$ is its \keyterm{fragment cluster}, with weight $w(t)$ and centroid $\pi(t)$,
\begin{equation}\label{eq-candidate-centroid}
w(t)=\sum_{p\in P}\phi_T(p,t),
\qquad
\pi(t)=\frac1{w(t)}\sum_{p\in P}\phi_T(p,t)\,p ,
\end{equation}
and the \keyterm{weighted centroid set} $\pi(T)$ contains the point $\pi(t)$ with weight $w(t)$ for every $t\in T$.
The solution implicitly moves every point $p$ to $\nu_p$, so the displacement of $p$ is $\nu_p-p$:
\begin{equation}\label{eq-weighted-mean}
\nu_p=\sum_{t\in T}\phi_T(p,t)\,\pi(t).
\end{equation}
Its cost after every candidate is moved to its centroid, and its drift, are
\begin{align}
C_T&=\sum_{p\in P}\sum_{t\in T}\phi_T(p,t)\lVert p-\pi(t)\rVert^2 ,\label{eq-CT}\\
C_\nu&=\sum_{p\in P}\lVert p-\nu_p\rVert^2 .\label{eq-Cnu}
\end{align}
With the ratio $\rho\ge1$ of the weighted $k$-means algorithm of the third stage, the weight of the drift and the \keyterm{drift-regularized cost} are
\begin{equation}\label{eq-lambda}
\lambda=\frac\rho{1+\rho}\in\Bigl[\frac12,1\Bigr),
\qquad
\Psi(\phi_T)=C_T+\lambda\,C_\nu .
\end{equation}
For another relaxed solution $(T',\phi_{T'})$ we write $C_{T'}$, $\nu_p(\phi_{T'})$ and $C_\nu(\phi_{T'})$ for the same quantities, and $w(t)$, $\pi(t)$ for $t\in T'$ refer to $\phi_{T'}$; $\nu_p$ and $C_\nu$ without an argument refer to $(T,\phi_T)$.
The \keyterm{$\theta$-mixture} of two relaxed solutions $(T_1,\phi_{T_1})$ and $(T_2,\phi_{T_2})$, for $\theta\in[0,1]$, is the pair $(T_\theta,\phi_{T_\theta})$ with the candidate set $T_\theta=T_1\uplus T_2$ and the assignment matrix $\phi_{T_\theta}=(1-\theta)\phi_{T_1}\oplus\theta\phi_{T_2}$: each $p$ gets fraction $(1-\theta)\phi_{T_1}(p,t)$ at $t\in T_1$ and $\theta\phi_{T_2}(p,t)$ at $t\in T_2$.
It is a relaxed solution because the fairness constraints~\eqref{eq-fairness} are invariant under column scaling.
The \keyterm{linearized cost} at $(T,\phi_T)$ is
\begin{equation}\label{eq-linearized}
\cL(\phi_{T'})=C_{T'}+2\lambda\sum_{p\in P}(\nu_p-p)^\top\bigl(\nu_p(\phi_{T'})-p\bigr),
\end{equation}
where the points $\nu_p$ are those of $(T,\phi_T)$, and $\cL(\phi_T)=C_T+2\lambda C_\nu$.
Finally, as in Definition~\ref{def-stationary}, the solution $(T,\phi_T)$ is $\tau$-stationary if $\cL(\phi_T)\le\cL(\phi_{T'})+\tau$ for every relaxed solution $(T',\phi_{T'})$.

\paragraph{Mixing identity.}
The linearized cost gives the rate at which $\Psi(\phi_T)$ changes when a small fraction of another relaxed solution is mixed into $(T,\phi_T)$.
The next lemma expresses the drift-regularized cost of a mixture exactly.

\begin{lemma}[Mixing identity]\label{lem-mixing}
Fix $\theta\in[0,1]$ and two relaxed solutions $(T_1,\phi_{T_1})$ and $(T_2,\phi_{T_2})$, and let $(T_\theta,\phi_{T_\theta})$ be their $\theta$-mixture, so that $T_\theta=T_1\uplus T_2$ and $\phi_{T_\theta}=(1-\theta)\phi_{T_1}\oplus\theta\phi_{T_2}$.
Then
\begin{equation}\label{eq-mixing-linear}
C_{T_\theta}=(1-\theta)C_{T_1}+\theta C_{T_2},
\qquad
\nu_p(\phi_{T_\theta})=(1-\theta)\nu_p(\phi_{T_1})+\theta\,\nu_p(\phi_{T_2})
\quad\text{for every }p\in P.
\end{equation}
If, moreover, $\cL$ is the linearized cost at $(T_1,\phi_{T_1})$, then
\begin{equation}\label{eq-mixing-identity}
\Psi(\phi_{T_\theta})
=\Psi(\phi_{T_1})-\theta\bigl(\cL(\phi_{T_1})-\cL(\phi_{T_2})\bigr)
+\lambda\theta^2\sum_{p\in P}\bigl\lVert\nu_p(\phi_{T_1})-\nu_p(\phi_{T_2})\bigr\rVert^2 .
\end{equation}
\end{lemma}

\begin{proof}
Let $t\in T_1$.
In $\phi_{T_\theta}$, the column of $t$ is its column in $\phi_{T_1}$ scaled by $1-\theta$; hence the weight of $t$ is scaled by $1-\theta$, and its centroid $\pi(t)$ does not change, since a weighted centroid is invariant under scaling all weights by the same factor.
So $t$ contributes $(1-\theta)$ times its share of $C_{T_1}$ to $C_{T_\theta}$ and $(1-\theta)$ times its share of $\nu_p(\phi_{T_1})$ to $\nu_p(\phi_{T_\theta})$.
The candidates in $T_2$ are handled in the same way, with $\theta$ instead of $1-\theta$, and summing over $T_1\uplus T_2$ proves~\eqref{eq-mixing-linear}.

For~\eqref{eq-mixing-identity}, the second identity in~\eqref{eq-mixing-linear} gives
\[
p-\nu_p(\phi_{T_\theta})=\bigl(p-\nu_p(\phi_{T_1})\bigr)+\theta\bigl(\nu_p(\phi_{T_1})-\nu_p(\phi_{T_2})\bigr)
\qquad\text{for every }p\in P .
\]
We square both sides, sum over $p$ and multiply by $\lambda$.
By the definition~\eqref{eq-linearized} of the linearized cost at $(T_1,\phi_{T_1})$, the cross terms contribute
\[
2\theta\lambda\sum_{p\in P}\bigl(p-\nu_p(\phi_{T_1})\bigr)^\top\bigl(\nu_p(\phi_{T_1})-\nu_p(\phi_{T_2})\bigr)
=\theta\Bigl(\bigl(\cL(\phi_{T_2})-C_{T_2}\bigr)-\bigl(\cL(\phi_{T_1})-C_{T_1}\bigr)\Bigr),
\]
and therefore
\[
\lambda\,C_\nu(\phi_{T_\theta})=\lambda\,C_\nu(\phi_{T_1})+\theta\bigl(\cL(\phi_{T_2})-\cL(\phi_{T_1})-C_{T_2}+C_{T_1}\bigr)+\lambda\theta^2\sum_{p\in P}\bigl\lVert\nu_p(\phi_{T_1})-\nu_p(\phi_{T_2})\bigr\rVert^2 .
\]
By the first identity in~\eqref{eq-mixing-linear}, $C_{T_\theta}=C_{T_1}+\theta(C_{T_2}-C_{T_1})$, and adding the last display yields~\eqref{eq-mixing-identity}.
\end{proof}

By~\eqref{eq-mixing-identity}, if $\cL(\phi_{T'})<\cL(\phi_T)$, then the $\theta$-mixture $(T_\theta,\phi_{T_\theta})$ of $(T,\phi_T)$ and $(T',\phi_{T'})$ satisfies $\Psi(\phi_{T_\theta})<\Psi(\phi_T)$ for all sufficiently small $\theta>0$.
The first stage therefore repeats the following \keyterm{iteration}.
It computes a relaxed solution $(T',\phi_{T'})$ whose linearized cost is close to the minimum over all relaxed solutions.
If $\cL(\phi_{T'})$ is sufficiently smaller than $\cL(\phi_T)$, the iteration replaces $(T,\phi_T)$ by its $\theta$-mixture with $(T',\phi_{T'})$ for a small fixed $\theta$; otherwise, no relaxed solution has a much smaller linearized cost, so $(T,\phi_T)$ is $\tau$-stationary for a small $\tau$ in the sense of Definition~\ref{def-stationary}, and the first stage stops.

We note three consequences of the mixing identity, which explain our terminology.
First, along every mixture the drift-regularized cost is a convex quadratic function of $\theta$; the first stage relies on this convexity of the relaxation.
Second, a relaxed solution $(T,\phi_T)$ that minimizes $\Psi(\phi_T)$ is $0$-stationary: if some relaxed solution had a smaller linearized cost, then the $\theta$-mixture $(T_\theta,\phi_{T_\theta})$ with this solution would satisfy $\Psi(\phi_{T_\theta})<\Psi(\phi_T)$ for small $\theta>0$, because for small $\theta>0$ the linear term in~\eqref{eq-mixing-identity} dominates the quadratic term.
Third, the drift-regularized cost of a $\tau$-stationary relaxed solution $(T,\phi_T)$ exceeds the optimal value of the relaxation by at most $\tau$.
To see this, let $(T',\phi_{T'})$ be any relaxed solution.
By~\eqref{eq-mixing-linear}, the $1$-mixture of $(T,\phi_T)$ and $(T',\phi_{T'})$ has drift-regularized cost $\Psi(\phi_{T'})$, and~\eqref{eq-mixing-identity} with $\theta=1$, whose last term is nonnegative, gives
\[
\Psi(\phi_{T'})\ge\Psi(\phi_T)-\bigl(\cL(\phi_T)-\cL(\phi_{T'})\bigr)\ge\Psi(\phi_T)-\tau .
\]

\paragraph{The linear program of an iteration.}
It remains to explain how each iteration finds a relaxed solution with small linearized cost.
After a transformation of the input points, the linearized cost can be approximately minimized by a linear program.
For $p\in P$, let
\begin{equation}\label{eq-shift}
u_p=\lambda\,(\nu_p-p),
\qquad
p_\lambda=p-u_p,
\end{equation}
so $p_\lambda$ is obtained by moving $p$ backwards by the fraction $\lambda$ of its displacement, and let $P_\lambda=\{p_\lambda:p\in P\}$ and $\cU=\{u_p:p\in P\}$ be the indexed sets of the shifted points and of the vectors $u_p$.
The next identity shows that $\cL(\phi_{T'})$ decomposes into one term per column, and that the term of each column is its optimal single-center cost on $P_\lambda$ minus its optimal single-center cost on $\cU$.

\begin{lemma}[Column identity]\label{lem-column-identity}
Define $u_p$ and $p_\lambda$ by~\eqref{eq-shift} for $(T,\phi_T)$, and let $(T',\phi_{T'})$ be a relaxed solution.
Every candidate $t\in T'$ satisfies
\begin{equation}\label{eq-column-identity}
\begin{split}
&\sum_{p\in P}\phi_{T'}(p,t)\Bigl(\lVert p-\pi(t)\rVert^2+2\lambda(\nu_p-p)^\top\bigl(\pi(t)-p\bigr)\Bigr)\\
&\qquad=\min_{c\in\R^d}\sum_{p\in P}\phi_{T'}(p,t)\lVert p_\lambda-c\rVert^2-\min_{v\in\R^d}\sum_{p\in P}\phi_{T'}(p,t)\lVert u_p-v\rVert^2 ,
\end{split}
\end{equation}
where $\pi(t)$ is the centroid of $t$ in $\phi_{T'}$.
Consequently,
\begin{equation}\label{eq-linearized-columns}
\cL(\phi_{T'})=\sum_{t\in T'}\Bigl(\min_{c\in\R^d}\sum_{p\in P}\phi_{T'}(p,t)\lVert p_\lambda-c\rVert^2-\min_{v\in\R^d}\sum_{p\in P}\phi_{T'}(p,t)\lVert u_p-v\rVert^2\Bigr).
\end{equation}
\end{lemma}

\begin{proof}
Since $2\lambda(\nu_p-p)=2u_p$ and $p_\lambda=p-u_p$, each $p\in P$ satisfies
\begin{align*}
\lVert p_\lambda-\pi(t)\rVert^2&=\lVert(p-\pi(t))-u_p\rVert^2=\lVert p-\pi(t)\rVert^2+2u_p^\top(\pi(t)-p)+\lVert u_p\rVert^2\\
&=\lVert p-\pi(t)\rVert^2+2\lambda(\nu_p-p)^\top(\pi(t)-p)+\lVert u_p\rVert^2 .
\end{align*}
Multiplying by $\phi_{T'}(p,t)$ and summing over $p$, we see that the left-hand side of~\eqref{eq-column-identity} equals
\[
\sum_p\phi_{T'}(p,t)\lVert p_\lambda-\pi(t)\rVert^2-\sum_p\phi_{T'}(p,t)\lVert u_p\rVert^2 .
\]
Let $\bar u=\frac1{w(t)}\sum_p\phi_{T'}(p,t)\,u_p$ be the mean of the points $u_p$ weighted by the column of $t$; with the same weights, the centroid of the points $p_\lambda$ is $\pi(t)-\bar u$, so the centroid identity gives
\begin{align*}
\sum_{p\in P}\phi_{T'}(p,t)\lVert p_\lambda-\pi(t)\rVert^2&=\min_{c\in\R^d}\sum_{p\in P}\phi_{T'}(p,t)\lVert p_\lambda-c\rVert^2+w(t)\lVert\bar u\rVert^2,\\
\sum_{p\in P}\phi_{T'}(p,t)\lVert u_p\rVert^2&=\min_{v\in\R^d}\sum_{p\in P}\phi_{T'}(p,t)\lVert u_p-v\rVert^2+w(t)\lVert\bar u\rVert^2;
\end{align*}
subtracting the second line from the first yields~\eqref{eq-column-identity}.

For the second statement, recall that $\nu_p(\phi_{T'})=\sum_{t\in T'}\phi_{T'}(p,t)\,\pi(t)$ and that each row of $\phi_{T'}$ sums to one; therefore the cross term of~\eqref{eq-linearized} decomposes over the candidates,
\[
2\lambda\sum_{p\in P}(\nu_p-p)^\top\bigl(\nu_p(\phi_{T'})-p\bigr)=\sum_{t\in T'}\sum_{p\in P}\phi_{T'}(p,t)\,2\lambda(\nu_p-p)^\top\bigl(\pi(t)-p\bigr),
\]
and so does $C_{T'}=\sum_{t\in T'}\sum_p\phi_{T'}(p,t)\lVert p-\pi(t)\rVert^2$.
Hence $\cL(\phi_{T'})$ is the sum of the left-hand sides of~\eqref{eq-column-identity} over all $t\in T'$, and~\eqref{eq-linearized-columns} follows.
\end{proof}

Given finite sets $T'$ and $\cV$ of candidate centers for the two terms in~\eqref{eq-linearized-columns}, each iteration solves the linear program
\begin{equation}\label{eq-step-lp}
\begin{aligned}
\min_{\phi_{T'},z}\quad
&\sum_{c\in T'}z_c\\
\text{s.t.}\quad
&\sum_{c\in T'}\phi_{T'}(p,c)=1,
&&\forall p\in P,\\
&\phi_{T'}(p,c)\ge0,
&&\forall p\in P,\ c\in T',\\
&z_c\ge\sum_{p\in P}\phi_{T'}(p,c)\Bigl(\lVert p_\lambda-c\rVert^2-\lVert u_p-v\rVert^2\Bigr),
&&\forall c\in T',\ v\in\cV,\\
&\beta_i\sum_{p\in P}\phi_{T'}(p,c)\le\sum_{p\in P^{(i)}}\phi_{T'}(p,c)\le\alpha_i\sum_{p\in P}\phi_{T'}(p,c),
&&\forall c\in T',\ i\in[m].
\end{aligned}
\end{equation}
The LP has a variable $\phi_{T'}(p,c)$ for each $p\in P$ and $c\in T'$ and an auxiliary variable $z_c$ for each $c\in T'$; for every feasible $\phi_{T'}$, $(T',\phi_{T'})$ is a relaxed solution.
When $\phi_{T'}$ is fixed, the minimum feasible value of $z_c$ equals the maximum right-hand side of the constraints in the third group, so the column of $c$ contributes its cost on $P_\lambda$ with center $c$ minus its minimum cost on $\cU$ with a center in $\cV$.
The algorithm uses the candidate sets $T'=\WCS_{\eta/88}(P_\lambda)$ and $\cV=\WCS_{\eta/4}(\cU)$ of Section~\ref{sec-prelim}, where $\eta$ is the accuracy parameter defined in the first line of Algorithm~\ref{alg-main}.

\paragraph{Clustering the centroids.}
When the first stage ends, the $\rho$-approximate weighted $k$-means algorithm is run on $\pi(T)$ and returns a set $S_\pi$ of at most $k$ centers.
The algorithm outputs $S_\pi$ together with an optimal solution of the fair assignment LP of $S_\pi$.
Algorithm~\ref{alg-main} summarizes the steps.

\begin{algorithm}[!ht]
\DontPrintSemicolon
\caption{Fair $k$-means by the drift-regularized relaxation}\label{alg-main}
\KwIn{$P$, $P^{(1)},\dots,P^{(m)}$, $k$, $\alpha$, $\beta$, $\epsilon\in(0,1]$, and a $\rho$-approximate weighted $k$-means algorithm}
\KwOut{a set $S_\pi$ of at most $k$ centers and an $(\alpha,\beta)$-fair assignment $\phi^*_{S_\pi}$ of $P$ to $S_\pi$}
Set $\eta=\min\{1/10,\ \epsilon/(3\rho)\}$ and $\lambda=\rho/(1+\rho)$, and construct $T_0=\WCS_\eta(P)$\;
Solve the fair assignment LP of $T_0$; let $\phi^*_{T_0}$ be an optimal solution and $L_0$ its value; set $(T,\phi_T)\leftarrow(T_0,\phi^*_{T_0})$\;
\Repeat{$g\le\eta L_0/2$}{
  Compute $w(t)$, $\pi(t)$, $\nu_p$, $u_p$ and $P_\lambda$ by \eqref{eq-candidate-centroid}, \eqref{eq-weighted-mean} and \eqref{eq-shift}\;
  Construct $T'=\WCS_{\eta/88}(P_\lambda)$ and $\cV=\WCS_{\eta/4}(\cU)$, and solve LP~\eqref{eq-step-lp} to obtain $\phi_{T'}$\;
  Compute the gap $g=\cL(\phi_T)-\cL(\phi_{T'})$ by \eqref{eq-linearized}\;
  \If{$g>\eta L_0/2$}{
    $(T,\phi_T)\leftarrow\bigl(T\uplus T',\ (1-\frac{\eta}{72})\phi_T\oplus\frac{\eta}{72}\phi_{T'}\bigr)$\;
  }
}
Compute $\pi(T)$ and run the weighted $k$-means algorithm on it to obtain $S_\pi$\;
Solve the fair assignment LP of $S_\pi$ to obtain $\phi^*_{S_\pi}$\;
\Return{$(S_\pi,\phi^*_{S_\pi})$}\;
\end{algorithm}

An update replaces $(T,\phi_T)$ by its $\frac\eta{72}$-mixture with $(T',\phi_{T'})$.
The parameter $\eta$ controls the accuracy of the first stage, and the constants $88$, $4$ and $72$ are chosen for the estimates in Section~\ref{sec-computing} and have not been optimized: each linear program minimizes the linearized cost up to an additive error of $\frac\eta2L_0$, each update decreases $\Psi(\phi_T)$ by at least $\eta^2L_0/288$, and therefore at most $\lceil576/\eta^2\rceil$ updates occur.
Hence the algorithm calls the $k$-means algorithm once and solves $O((\rho/\epsilon)^2)$ linear programs, each of size polynomial in the input for fixed $\epsilon$ and $\rho$.

\section{Approximation Analysis}\label{sec-analysis}

We now prove Theorem~\ref{thm-main}; Figure~\ref{fig-analysis-map} shows how the lemmas of this section fit together.
Throughout this section, $(T,\phi_T)$ denotes a relaxed solution, and we use the notation listed at the beginning of Section~\ref{sec-algorithm}: the weights $w(t)$ and centroids $\pi(t)$ of the fragment clusters, the points $\nu_p$, the costs $C_T$ and $C_\nu$, the drift-regularized cost $\Psi(\phi_T)$ with $\lambda=\rho/(1+\rho)$, and the linearized cost $\cL$ at $(T,\phi_T)$.
In addition, $w^{(i)}(t)=\sum_{p\in P^{(i)}}\phi_T(p,t)$ is the weight of the group $P^{(i)}$ in the fragment cluster of a candidate $t$; for a candidate of another relaxed solution, all these quantities refer to that solution.

\begin{figure}[!htb]
\centering
\begin{tikzpicture}[fig,
  box/.style={draw=figblue!60,fill=figblue!6,rounded corners=2pt,align=center,font=\footnotesize,text width=2.55cm,inner sep=2.5pt,minimum height=1.35cm},
  arr/.style={->,figgray,line width=0.7pt}]
\node[box] (rt) at (0,0) {Lemma~\ref{lem-fair-routing}: routing\\$\Cost(P,S,\phi^*_S)$\\$\le C_T+\Cost(\pi(T),S)$};
\node[box] (cp) at (3.1,0) {Lemma~\ref{lem-comparison}: comparison\\$\OPT_\pi\le C_T+\Cost^\circ$\\$+\,2C_e$};
\node[box] (st) at (6.2,0) {Lemma~\ref{lem-stationarity}: stationarity\\$C_T+2\lambda C_\nu\le\Cost^\circ$\\$-\,2\lambda C_e+\tau$};
\node[box] (mb) at (9.6,0) {Lemma~\ref{lem-merge-bound}: final cost\\output $\le(2\rho+1)\Cost^\circ$\\$+\,(1+\rho)\tau$};
\node[box,draw=figred!70,fill=figred!5] (th) at (12.7,0) {Theorem~\ref{thm-main}\\ratio $2\rho+1+\epsilon$};
\node[box] (in) at (0,-2.45) {Lemma~\ref{lem-initial}: initial\\$\Psi(\phi^*_{T_0})$\\$\le(1+\lambda)L_0$,\\$L_0\le(1+\eta)\Cost^\circ$};
\node[box] (bi) at (3.1,-2.45) {Lemma~\ref{lem-bounded-iterates}: bounded iterates\\$C_\nu\le L_0$, $\cL(\phi_T)\le3L_0$};
\node[box] (lp) at (6.2,-2.45) {Lemma~\ref{lem-step-lp}: one LP\\minimizes $\cL(\phi_{T'})$\\up to $\frac\eta2L_0$};
\node[box] (de) at (9.3,-2.45) {Lemma~\ref{lem-descent}: descent\\$\Psi(\phi_T)$ drops\\by $\ge\frac{\eta^2}{288}L_0$};
\node[box] (te) at (12.7,-2.45) {Lemma~\ref{lem-termination}: termination\\$\tau=\eta L_0$ after\\$\le\lceil576/\eta^2\rceil$ updates};
\begin{scope}[on background layer]
  \node[draw=figgray!70,densely dashed,rounded corners=3pt,inner sep=4pt,fit=(rt)(cp)(st)] (step1) {};
\end{scope}
\node[figlabs,anchor=south west,text=figgray] at (step1.north west) {Step 1 (Sections~\ref{sec-routing}--\ref{sec-stationarity}): every $\tau$-stationary relaxed solution leads to a good output};
\node[figlabs,anchor=north west,text=figgray] at ($(in.south west)+(0,-0.08)$) {Step 2 (Section~\ref{sec-computing}): the loop of Algorithm~\ref{alg-main} reaches $\tau=\eta L_0$};
\draw[arr] (step1.east) -- (mb.west);
\draw[arr] (mb.east) -- (th.west);
\draw[arr] (in) -- (bi); \draw[arr] (bi) -- (lp); \draw[arr] (lp) -- (de); \draw[arr] (de) -- (te);
\draw[arr] (te) -- node[figlabs,right] {$\tau=\eta L_0$} (th);
\end{tikzpicture}
\caption{How the lemmas of Section~\ref{sec-analysis} combine.
Here $\Cost^\circ=\Cost(P,S^\circ,\phi^\circ)$ is the cost of a comparison solution $(S^\circ,\phi^\circ)$, a fair fractional solution whose centers lie at the weighted centroids of their points, as defined in Section~\ref{sec-comparison}; an optimal fair solution is an example.
Step~1 holds for every $\tau$-stationary relaxed solution: the three lemmas in the dashed box bound the cost of the output, of the clustered centroids, and of the alignment $C_e$, and Lemma~\ref{lem-merge-bound} combines them so that $C_T$ and $C_e$ drop out.
Step~2 shows that the loop stops at such a solution with $\tau=\eta L_0$, and the bound $L_0\le(1+\eta)\Cost^\circ$ of Lemma~\ref{lem-initial} turns the additive error into the term $\epsilon$ of Theorem~\ref{thm-main}.}
\label{fig-analysis-map}
\end{figure}

\subsection{Routing a Relaxed Solution to Centers}\label{sec-routing}

Reassigning the entire column of each candidate to a nearest center in $S$ preserves fairness (Figure~\ref{fig-routing}(a)). By the centroid identity, the cost of this reassignment is $C_T+\Cost(\pi(T),S)$, where $C_T$ does not depend on $S$.

\begin{lemma}[Fair routing]\label{lem-fair-routing}
Consider a relaxed solution $(T,\phi_T)$ and a center set $S$.
The fair assignment LP of $S$ has a feasible solution, and its optimal solution $\phi^*_S$ satisfies
\[
\Cost(P,S,\phi^*_S)\le C_T+\Cost(\pi(T),S).
\]
In addition, $C_T\le\Cost(P,T,\phi_T)$.
\end{lemma}

\begin{proof}
Define $\widehat\phi_S$ by reassigning the column of each candidate $t$ to $\cN(\pi(t),S)$, a center of $S$ nearest to its centroid:
\[
\widehat\phi_S(p,s)=\sum_{t:\,\cN(\pi(t),S)=s}\phi_T(p,t).
\]
In words, every center $s$ of $S$ receives the whole columns of the candidates whose centroids have $s$ as their nearest center, and its column in $\widehat\phi_S$ is the sum of these columns, as in Figure~\ref{fig-routing}(a).
The row of $p$ in $\widehat\phi_S$ is obtained by grouping the entries of its row in $\phi_T$, so it sums to $1$.
Similarly, in terms of the weights $w(t)$ and $w^{(i)}(t)$ of the fragment clusters, the total weight and the weight of the group $P^{(i)}$ assigned to the center $s$ are
\[
\sum_{p\in P}\widehat\phi_S(p,s)=\sum_{t:\,\cN(\pi(t),S)=s}w(t),
\qquad
\sum_{p\in P^{(i)}}\widehat\phi_S(p,s)=\sum_{t:\,\cN(\pi(t),S)=s}w^{(i)}(t).
\]
Since $\phi_T$ is fair, every candidate satisfies $\beta_iw(t)\le w^{(i)}(t)\le\alpha_iw(t)$, and summing these inequalities over the candidates reassigned to $s$ shows that $\widehat\phi_S$ is fair.
In particular, the fair assignment LP of $S$ is feasible, since $\widehat\phi_S$ satisfies its constraints.

To bound the cost, we view the column of $t$ as a weighting of $P$ by the fractions $\phi_T(p,t)$; its total weight is $w(t)$ and its centroid is $\pi(t)$, and the centroid identity~\eqref{eq-centroid-identity} with $z=\cN(\pi(t),S)$ gives
\[
\sum_{p\in P}\phi_T(p,t)\lVert p-\cN(\pi(t),S)\rVert^2
=\sum_{p\in P}\phi_T(p,t)\lVert p-\pi(t)\rVert^2+w(t)\dist^2(\pi(t),S).
\]
Since $\phi^*_S$ is an optimal solution of the fair assignment LP, summing over $t\in T$ gives
\[
\Cost(P,S,\phi^*_S)\le\Cost(P,S,\widehat\phi_S)
=\sum_{t\in T}\sum_{p\in P}\phi_T(p,t)\lVert p-\cN(\pi(t),S)\rVert^2
=C_T+\Cost(\pi(T),S).
\]
The centroid identity for the same column, now with $z=t$, gives
\[
\sum_{p\in P}\phi_T(p,t)\lVert p-t\rVert^2
=\sum_{p\in P}\phi_T(p,t)\lVert p-\pi(t)\rVert^2+w(t)\lVert\pi(t)-t\rVert^2
\ge\sum_{p\in P}\phi_T(p,t)\lVert p-\pi(t)\rVert^2 ,
\]
and summing over $t$ gives $C_T\le\Cost(P,T,\phi_T)$.
\end{proof}

We denote the optimal $k$-means cost of the weighted centroid set by
\[
\OPT_\pi=\min_{|S|\le k}\Cost(\pi(T),S).
\]
The following identity relates the weighted means to $C_T$ and to the drift.

\begin{lemma}[Weighted-mean identity]\label{lem-weighted-mean}
For every relaxed solution,
\begin{equation}\label{eq-weighted-mean-identity}
\sum_{p\in P}\sum_{t\in T}\phi_T(p,t)\lVert\pi(t)-\nu_p\rVert^2=C_T-C_\nu ,
\qquad
C_\nu\le C_T .
\end{equation}
\end{lemma}

\begin{proof}
Fix $p\in P$.
Since the fractions $\phi_T(p,t)$ sum to one, $\nu_p$ is the centroid of the points $\pi(t)$, $t\in T$, weighted by these fractions, and the centroid identity with $z=p$ gives
\[
\sum_{t\in T}\phi_T(p,t)\lVert p-\pi(t)\rVert^2
=\sum_{t\in T}\phi_T(p,t)\lVert\pi(t)-\nu_p\rVert^2+\lVert p-\nu_p\rVert^2 .
\]
Summing over $p$ proves the identity, since the left-hand side sums to $C_T$ and the last term sums to $C_\nu$.
The left-hand side of the identity is nonnegative, so $C_\nu\le C_T$.
\end{proof}

\begin{figure}[!htb]
\centering
\begin{tikzpicture}[fig]
\begin{scope}
  \begin{scope}[on background layer]
    \draw[fighull] (0.45,1.65) ellipse (0.78 and 0.68);
    \draw[fighull] (2.45,-0.57) ellipse (0.78 and 0.55);
    \draw[fighull] (4.58,2.35) ellipse (0.9 and 0.55);
  \end{scope}
  \foreach \c/\q in {{(0.425,1.65)}/{(0.0,1.7)},{(0.425,1.65)}/{(0.5,2.2)},{(0.425,1.65)}/{(0.9,1.5)},{(0.425,1.65)}/{(0.3,1.2)},
                     {(2.45,-0.575)}/{(2.0,-0.5)},{(2.45,-0.575)}/{(2.6,-0.2)},{(2.45,-0.575)}/{(2.3,-0.9)},{(2.45,-0.575)}/{(2.9,-0.7)},
                     {(4.575,2.35)}/{(4.0,2.4)},{(4.575,2.35)}/{(4.6,2.7)},{(4.575,2.35)}/{(5.2,2.3)},{(4.575,2.35)}/{(4.5,2.0)}}
    \draw[figgray!70,line width=0.35pt] \c -- \q;
  \foreach \q in {(0.0,1.7),(0.9,1.5),(2.0,-0.5),(2.9,-0.7),(4.0,2.4),(5.2,2.3)} \node[figpt] at \q {};
  \foreach \q in {(0.5,2.2),(0.3,1.2),(2.6,-0.2),(2.3,-0.9),(4.6,2.7),(4.5,2.0)} \node[figpt,fill=white,draw=black,line width=0.5pt] at \q {};
  \foreach \c in {(0.425,1.65),(2.45,-0.575),(4.575,2.35)} \node[figbluept] at \c {};
  \node[figctr] (s1) at (1.6,0.55) {};
  \node[figctr] (s2) at (4.6,0.75) {};
  \draw[->,figblue,line width=0.8pt,densely dashed,shorten >=5.5pt] (0.425,1.65) -- (s1);
  \draw[->,figblue,line width=0.8pt,densely dashed,shorten >=5.5pt] (2.45,-0.575) -- (s1);
  \draw[->,figblue,line width=0.8pt,densely dashed,shorten >=5.5pt] (4.575,2.35) -- (s2);
  \node[figlab,right=5pt] at (s1) {$s_1$};
  \node[figlab,right=5pt] at (s2) {$s_2$};
  \node[figlabs,anchor=east] at (-0.4,1.65) {$t_1$};
  \node[figlabs,anchor=north] at (2.45,-1.2) {$t_2$};
  \node[figlabs,anchor=south] at (4.58,2.95) {$t_3$};
  \node at (2.4,-2.05) {(a) rerouting whole columns};
\end{scope}
\begin{scope}[xshift=8.4cm]
  \coordinate (p1) at (0.3,0.0);
  \coordinate (p2) at (2.9,1.9);
  \coordinate (pi) at (1.6,0.95);
  \coordinate (s1) at (-0.7,1.9);
  \coordinate (s2) at (4.1,0.3);
  \draw[figgray!70,line width=0.5pt] (p1) -- (pi) -- (p2);
  \draw[figred,line width=0.7pt,densely dashed] (p1) -- (s1);
  \draw[figred,line width=0.7pt,densely dashed] (p2) -- (s2);
  \draw[->,figblue,line width=0.9pt,shorten >=5.5pt] (pi) -- (s1);
  \draw[->,figblue,line width=0.9pt,shorten >=5.5pt] (pi) -- (s2);
  \draw[figgray,line width=0.5pt] ($(p1)!0.3cm!(pi)$) arc[start angle=36.2,end angle=117.8,radius=0.3cm];
  \node[figpt] at (p1) {}; \node[figpt] at (p2) {};
  \node[figbluept] at (pi) {};
  \node[figopt] at (s1) {}; \node[figopt] at (s2) {};
  \node[figlab,below=3pt] at (p1) {$p$};
  \node[figlab,above=3pt] at (p2) {$p'$};
  \node[figlab,below right=1pt and 2pt] at (pi) {$\pi(t)$};
  \node[figlab,left=5pt] at (s1) {$s$};
  \node[figlab,right=5pt] at (s2) {$s'$};
  \node[figlabs,text=figgray,anchor=north west] at (0.95,0.45) {$\lVert\pi(t)-p\rVert^2$};
  \node[figlabs,text=figred,anchor=east] at (-0.3,0.8) {$\lVert p-s\rVert^2$};
  \node[figlabs,text=figgray,anchor=north] at (0.3,-0.35) {cross term};
  \node at (1.7,-2.05) {(b) serving a centroid via a comparison solution};
\end{scope}
\end{tikzpicture}
\caption{The two ways of serving the centroids in Sections~\ref{sec-routing} and~\ref{sec-comparison}.
(a) Lemma~\ref{lem-fair-routing}: the fragment clusters $t_1,t_2,t_3$ (shaded; filled and open disks are two groups) are moved as whole columns to the center of $S$ nearest to their centroids (blue); this is the assignment $\widehat\phi_S$ of the proof.
Each cluster of $S$ is then a sum of fair columns and hence fair, and the cost splits into the spread of each column around its centroid (gray segments, which sum to $C_T$) and the cost $w(t)\dist^2(\pi(t),S)$ of moving the centroid (dashed arrows).
(b) Lemma~\ref{lem-comparison}: a centroid $\pi(t)$ is served by the centers $s,s'$ of a comparison solution (red circles) to which the points $p,p'$ of its column are assigned, in proportion to $\phi_T(p,t)\phi^\circ(p,s)$.
Writing $\pi(t)-s=(\pi(t)-p)+(p-s)$ splits the squared length of each blue arrow into a term of $C_T$, a term of $\Cost(P,S^\circ,\phi^\circ)$ and a cross term at the angle at $p$; averaged over the candidates $t$, the cross terms form the alignment $C_e$.}
\label{fig-routing}
\end{figure}

\subsection{Comparing the Centroids with a Fair Solution}\label{sec-comparison}

To bound $\OPT_\pi$, we compare $\pi(T)$ with the centers of a fixed fair solution (Figure~\ref{fig-routing}(b)).
This solution need not be optimal.
Since we will use the centroid identity, it is convenient to assume that each center lies at the weighted centroid of its points; we formalize this as follows.
We call $(S^\circ,\phi^\circ)$ a \keyterm{comparison solution} if $S^\circ$ is a finite indexed set of centers and $\phi^\circ$ is an $(\alpha,\beta)$-fair, possibly fractional, assignment matrix of $P$ to $S^\circ$ in which every center is located at the weighted centroid of its points:
\begin{equation}\label{eq-comparison-centroid}
s=\frac{\sum_{p\in P}\phi^\circ(p,s)\,p}{\sum_{p\in P}\phi^\circ(p,s)}
\qquad\text{for every }s\in S^\circ.
\end{equation}
Every fair solution can be turned into a comparison solution with no more centers and no larger cost: a center whose column is zero is deleted, and every other center is moved to the weighted centroid of its points.
The columns of the remaining centers do not change, so fairness is preserved, and the centroid identity shows that the cost does not increase.
For instance, the optimal solution~\eqref{eq-opt}, with centers $\widetilde s_j$ and clusters $\widetilde C_j$, is itself a comparison solution, since its centers are the centroids of its clusters; it uses at most $k$ centers and costs $\OPT$.
Throughout the analysis, the superscript $\circ$ marks the quantities of a comparison solution: its centers $S^\circ$, its assignment $\phi^\circ$, the points $\nu^\circ_p$ defined next, and its cost $\Cost^\circ=\Cost(P,S^\circ,\phi^\circ)$ in Figure~\ref{fig-analysis-map}.
For a comparison solution we set
\begin{equation}\label{eq-Ce}
\nu^\circ_p=\sum_{s\in S^\circ}\phi^\circ(p,s)\,s,
\qquad
C_e=\sum_{p\in P}\sum_{s\in S^\circ}\phi^\circ(p,s)\,(\nu_p-p)^\top(p-s)
=\sum_{p\in P}(\nu_p-p)^\top(p-\nu^\circ_p),
\end{equation}
where the two expressions for $C_e$ are equal because $\sum_s\phi^\circ(p,s)=1$; for the optimal solution, $\nu^\circ_p=\widetilde s_j$ whenever $p\in\widetilde C_j$.
A comparison solution is also a relaxed solution, with candidate set $S^\circ$.
By~\eqref{eq-comparison-centroid}, its centroids are $\pi(s)=s$, so its cost~\eqref{eq-CT} is $\Cost(P,S^\circ,\phi^\circ)$ and its weighted means are the points $\nu^\circ_p$.

The cross term $C_e$ is the alignment discussed in Section~\ref{sec-overview-merge}.
It can be positive or negative: for a center $s$ to which $p$ is assigned in the comparison solution, the corresponding term of $p$ is negative if the displacement $\nu_p-p$ has a positive component in the direction of $s-p$, and positive if this component is negative.
The next lemma does not bound $C_e$; this is done in Section~\ref{sec-stationarity}.

\begin{lemma}[Comparison bound]\label{lem-comparison}
Let $(T,\phi_T)$ be a relaxed solution and let $(S^\circ,\phi^\circ)$ be a comparison solution with $|S^\circ|\le k$.
Then
\[
\OPT_\pi\le\Cost(\pi(T),S^\circ)\le C_T+\Cost(P,S^\circ,\phi^\circ)+2C_e .
\]
\end{lemma}

\begin{proof}
The first inequality follows from the definition of $\OPT_\pi$, since $|S^\circ|\le k$.
For the second inequality, fix a candidate $t\in T$.
For all $p\in P$ and $s\in S^\circ$, we have $\dist^2(\pi(t),S^\circ)\le\lVert\pi(t)-s\rVert^2$; multiplying by $\phi_T(p,t)\phi^\circ(p,s)$ and summing over all pairs $(p,s)$, we obtain
\[
w(t)\dist^2(\pi(t),S^\circ)\le\sum_{p\in P}\sum_{s\in S^\circ}\phi_T(p,t)\,\phi^\circ(p,s)\lVert\pi(t)-s\rVert^2 ,
\]
because every row of $\phi^\circ$ sums to one and the column of $t$ in $\phi_T$ sums to $w(t)$.
We sum these inequalities over all $t\in T$.
Writing $\pi(t)-s=(\pi(t)-p)+(p-s)$ and expanding the square, we get
\begin{align*}
\Cost(\pi(T),S^\circ)
&\le\sum_{p\in P}\sum_{s\in S^\circ}\sum_{t\in T}\phi_T(p,t)\,\phi^\circ(p,s)\lVert\pi(t)-s\rVert^2\\
&=\sum_{p\in P}\sum_{t\in T}\phi_T(p,t)\lVert\pi(t)-p\rVert^2
+\sum_{p\in P}\sum_{s\in S^\circ}\phi^\circ(p,s)\lVert p-s\rVert^2\\
&\qquad+2\sum_{p\in P}\sum_{s\in S^\circ}\phi^\circ(p,s)\Bigl(\sum_{t\in T}\phi_T(p,t)(\pi(t)-p)\Bigr)^{\!\top}(p-s).
\end{align*}
For the first sum on the right-hand side we used that $\sum_s\phi^\circ(p,s)=1$, and for the second that $\sum_t\phi_T(p,t)=1$; the two sums are $C_T$ and $\Cost(P,S^\circ,\phi^\circ)$.
In the last term, the inner sum is $\sum_t\phi_T(p,t)(\pi(t)-p)=\nu_p-p$, so the term equals $2C_e$ by~\eqref{eq-Ce}.
\end{proof}

Without further information on $C_e$, the Cauchy--Schwarz inequality and Lemma~\ref{lem-weighted-mean} give only $|C_e|\le\sqrt{C_T\Cost(P,S^\circ,\phi^\circ)}$.
For the optimal solution, Lemma~\ref{lem-comparison} then yields $\OPT_\pi\le(\sqrt{C_T}+\sqrt{\OPT})^2$, which is about $4\,\OPT$ when $C_T\approx\OPT$. Combined with Lemma~\ref{lem-fair-routing}, this gives the ratio $1+4\rho$ of~\cite{song2025relax}.
The drift-regularized relaxation provides the missing bound on $C_e$.

\subsection{Stationarity Bounds the Alignment}\label{sec-stationarity}

Algorithm~\ref{alg-main} does not compute a minimizer of the drift-regularized cost, and the analysis does not need one: it suffices that the current solution is $\tau$-stationary in the sense of Definition~\ref{def-stationary}.
The following lemma shows how stationarity bounds the alignment.

\begin{lemma}[Stationarity bound]\label{lem-stationarity}
Every $\tau$-stationary relaxed solution $(T,\phi_T)$ and every comparison solution $(S^\circ,\phi^\circ)$ satisfy
\[
C_T+2\lambda C_\nu\le\Cost(P,S^\circ,\phi^\circ)-2\lambda C_e+\tau .
\]
\end{lemma}

\begin{proof}
As noted in Section~\ref{sec-comparison}, the comparison solution, viewed as a relaxed solution, has cost $\Cost(P,S^\circ,\phi^\circ)$ and weighted means $\nu^\circ_p$; hence by~\eqref{eq-Ce} its linearized cost at $(T,\phi_T)$ is
\[
\cL(\phi^\circ)=\Cost(P,S^\circ,\phi^\circ)+2\lambda\sum_{p\in P}(\nu_p-p)^\top(\nu^\circ_p-p)=\Cost(P,S^\circ,\phi^\circ)-2\lambda C_e .
\]
Applying stationarity to this relaxed solution gives $C_T+2\lambda C_\nu=\cL(\phi_T)\le\cL(\phi^\circ)+\tau$.
\end{proof}

Combining this lemma with the comparison bound gives the main inequality of the analysis.

\begin{lemma}[Cost of the final solution]\label{lem-merge-bound}
Assume that $(T,\phi_T)$ is $\tau$-stationary and that $(S^\circ,\phi^\circ)$ is a comparison solution with $|S^\circ|\le k$.
Then
\begin{equation}\label{eq-merge-bound}
C_T+\rho\,\OPT_\pi\le(2\rho+1)\Cost(P,S^\circ,\phi^\circ)-2\rho\,C_\nu+(1+\rho)\tau .
\end{equation}
Moreover, let $S_\pi$ be the center set that the $\rho$-approximate weighted $k$-means algorithm returns on $\pi(T)$; then the optimal fair assignment $\phi^*_{S_\pi}$ to $S_\pi$ satisfies
\begin{equation}\label{eq-merge-output}
\Cost(P,S_\pi,\phi^*_{S_\pi})\le(2\rho+1)\Cost(P,S^\circ,\phi^\circ)-2\rho\,C_\nu+(1+\rho)\tau .
\end{equation}
\end{lemma}

\begin{proof}
Lemma~\ref{lem-stationarity} gives $2\lambda C_e\le\Cost(P,S^\circ,\phi^\circ)-C_T-2\lambda C_\nu+\tau$.
We multiply the bound of Lemma~\ref{lem-comparison} by $\rho$ and substitute this inequality, using $\rho/\lambda=1+\rho$ and $(1+\rho)\cdot2\lambda=2\rho$, which hold by~\eqref{eq-lambda}:
\begin{align*}
\rho\,\OPT_\pi&\le\rho\,C_T+\rho\Cost(P,S^\circ,\phi^\circ)+\frac\rho\lambda\cdot2\lambda C_e\\
&\le\rho\,C_T+\rho\Cost(P,S^\circ,\phi^\circ)+(1+\rho)\bigl(\Cost(P,S^\circ,\phi^\circ)-C_T-2\lambda C_\nu+\tau\bigr)\\
&=(2\rho+1)\Cost(P,S^\circ,\phi^\circ)-C_T-2\rho\,C_\nu+(1+\rho)\tau ,
\end{align*}
which is~\eqref{eq-merge-bound}.
For the second statement, the $k$-means algorithm guarantees $\Cost(\pi(T),S_\pi)\le\rho\,\OPT_\pi$, so Lemma~\ref{lem-fair-routing} applied to $S_\pi$ gives $\Cost(P,S_\pi,\phi^*_{S_\pi})\le C_T+\rho\,\OPT_\pi$, and~\eqref{eq-merge-bound} bounds the right-hand side.
\end{proof}

Since $C_\nu\ge0$, inequality~\eqref{eq-merge-output} implies in particular
\[
\Cost(P,S_\pi,\phi^*_{S_\pi})\le(2\rho+1)\Cost(P,S^\circ,\phi^\circ)+(1+\rho)\tau .
\]

\subsection{The First Stage Computes a Nearly Stationary Solution}\label{sec-computing}

It remains to prove that the loop of Algorithm~\ref{alg-main} terminates with an $(\eta L_0)$-stationary relaxed solution.
Here $\eta=\min\{1/10,\epsilon/(3\rho)\}$, the loop starts at $(T_0,\phi^*_{T_0})$, and $L_0=\Cost(P,T_0,\phi^*_{T_0})$ is the optimal value of the LP that produces this initial solution.
We also use $u_p=\lambda(\nu_p-p)$, $p_\lambda=p-u_p$ and the sets $P_\lambda=\{p_\lambda\}$ and $\cU=\{u_p\}$ of~\eqref{eq-shift}, and the candidate sets $\WCS_\delta(Q)$ of Section~\ref{sec-prelim}, which satisfy~\eqref{eq-weighted-candidates} with accuracy $\delta$.
The proof proceeds in four steps.
First, the initial solution has drift-regularized cost at most $(1+\lambda)L_0$, and $L_0$ is at most $1+\eta$ times the cost of every comparison solution (Lemma~\ref{lem-initial}).
Second, relaxed solutions $(T,\phi_T)$ with $\Psi(\phi_T)\le(1+\lambda)L_0$ have bounded drift, so the candidate sets used in each iteration are sufficiently accurate (Lemma~\ref{lem-bounded-iterates}).
Third, the linear program of each iteration minimizes the linearized cost up to a small error (Lemmas~\ref{lem-column-identity} and~\ref{lem-step-lp}).
Fourth, each update decreases the drift-regularized cost by at least $\eta^2L_0/288$ (Lemma~\ref{lem-descent}), and Lemma~\ref{lem-termination} combines the four steps.

\begin{lemma}[Initial relaxed solution]\label{lem-initial}
The initial pair $(T_0,\phi^*_{T_0})$ of Algorithm~\ref{alg-main} is a relaxed solution with $C_{T_0}\le L_0$ and $\Psi(\phi^*_{T_0})\le(1+\lambda)L_0$.
Furthermore, every comparison solution $(S^\circ,\phi^\circ)$ satisfies $L_0\le(1+\eta)\Cost(P,S^\circ,\phi^\circ)$.
\end{lemma}

\begin{proof}
By Lemma~\ref{lem-fair-routing}, the fair assignment LP of $T_0$ is feasible; the optimal solution $\phi^*_{T_0}$ chosen by the algorithm is a fair assignment of $P$ to $T_0$, so $(T_0,\phi^*_{T_0})$ is a relaxed solution.
By the same lemma, $C_{T_0}\le\Cost(P,T_0,\phi^*_{T_0})=L_0$, and by Lemma~\ref{lem-weighted-mean}, $C_\nu(\phi^*_{T_0})\le C_{T_0}$; hence
\[
\Psi(\phi^*_{T_0})=C_{T_0}+\lambda\,C_\nu(\phi^*_{T_0})\le(1+\lambda)\,C_{T_0}\le(1+\lambda)L_0 .
\]
Now let $(S^\circ,\phi^\circ)$ be a comparison solution and $s\in S^\circ$.
We apply~\eqref{eq-weighted-candidates} to $T_0=\WCS_\eta(P)$ with the weights $\mu(p)=\phi^\circ(p,s)$; the minimum is attained at $c=s$, which is the weighted centroid by~\eqref{eq-comparison-centroid}, and we obtain a candidate $t_s\in T_0$ with
\begin{equation}\label{eq-initial-candidate}
\sum_{p\in P}\phi^\circ(p,s)\lVert p-t_s\rVert^2\le(1+\eta)\sum_{p\in P}\phi^\circ(p,s)\lVert p-s\rVert^2 .
\end{equation}
We reassign the column of every $s$ to its candidate $t_s$:
\[
\widehat\phi(p,t)=\sum_{s\in S^\circ:\,t_s=t}\phi^\circ(p,s)\qquad(p\in P,\ t\in T_0).
\]
Every row of $\widehat\phi$ sums to one, and, by the argument in the proof of Lemma~\ref{lem-fair-routing}, every column of $\widehat\phi$ is fair because it is a sum of columns of $\phi^\circ$; thus $\widehat\phi$ satisfies the constraints of the fair assignment LP of $T_0$.
Summing~\eqref{eq-initial-candidate} over all $s\in S^\circ$, we obtain
\[
L_0\le\Cost(P,T_0,\widehat\phi)=\sum_{s\in S^\circ}\sum_{p\in P}\phi^\circ(p,s)\lVert p-t_s\rVert^2\le(1+\eta)\Cost(P,S^\circ,\phi^\circ).
\]
\end{proof}

Having bounded the initial drift-regularized cost, we now bound the drift and the linearized cost of every relaxed solution whose drift-regularized cost is at most $(1+\lambda)L_0$, as the initial solution is.

\begin{lemma}[Bounded iterates]\label{lem-bounded-iterates}
Suppose that the relaxed solution $(T,\phi_T)$ satisfies $\Psi(\phi_T)\le(1+\lambda)L_0$, and let $\cL$ be the linearized cost at $(T,\phi_T)$.
\begin{enumerate}
\item $C_\nu\le L_0$, and $\cL(\phi_T)=C_T+2\lambda C_\nu\le3L_0$.
\item $\cL(\phi_{T'})\ge C_{T'}-2\sqrt{L_0C_{T'}}$ for every relaxed solution $(T',\phi_{T'})$.
\item If a relaxed solution $(T',\phi_{T'})$ satisfies $\cL(\phi_{T'})\le\frac{16}5L_0$, then $C_{T'}\le10L_0$.
\end{enumerate}
\end{lemma}

\begin{proof}
For item~1, Lemma~\ref{lem-weighted-mean} gives $C_\nu\le C_T$, hence $(1+\lambda)C_\nu\le C_T+\lambda C_\nu=\Psi(\phi_T)\le(1+\lambda)L_0$, that is, $C_\nu\le L_0$, and since $\lambda<1$,
\[
\cL(\phi_T)=C_T+2\lambda C_\nu=\Psi(\phi_T)+\lambda C_\nu\le(1+\lambda)L_0+\lambda L_0\le3L_0 .
\]
Item~2 follows from~\eqref{eq-linearized} and the following bound on the cross term, which uses $\lambda<1$, the Cauchy--Schwarz inequality, item~1, and Lemma~\ref{lem-weighted-mean} applied to $(T',\phi_{T'})$:
\begin{align*}
\Bigl|2\lambda\sum_{p\in P}(\nu_p-p)^\top\bigl(\nu_p(\phi_{T'})-p\bigr)\Bigr|
&\le2\Bigl(\sum_{p\in P}\lVert\nu_p-p\rVert^2\Bigr)^{1/2}\Bigl(\sum_{p\in P}\lVert\nu_p(\phi_{T'})-p\rVert^2\Bigr)^{1/2}\\
&=2\sqrt{C_\nu\,C_\nu(\phi_{T'})}\le2\sqrt{L_0C_{T'}} .
\end{align*}
For item~3, item~2 gives $C_{T'}-2\sqrt{L_0C_{T'}}\le\cL(\phi_{T'})\le\frac{16}5L_0$, a quadratic inequality in $\sqrt{C_{T'}}$; hence $\sqrt{C_{T'}}\le\bigl(1+\sqrt{1+\frac{16}5}\bigr)\sqrt{L_0}\le3.05\sqrt{L_0}$, and $C_{T'}\le9.31L_0\le10L_0$.
\end{proof}

These bounds make the candidate sets of an iteration accurate enough for the following guarantee: the linear program minimizes the linearized cost up to a small error.

\begin{lemma}[The linear program of an iteration]\label{lem-step-lp}
Suppose that $\Psi(\phi_T)\le(1+\lambda)L_0$, and let $(\phi_{T'},z)$ be an optimal solution of LP~\eqref{eq-step-lp} for the candidate sets $T'=\WCS_{\eta/88}(P_\lambda)$ and $\cV=\WCS_{\eta/4}(\cU)$.
Then $(T',\phi_{T'})$ is a relaxed solution whose linearized cost exceeds that of any relaxed solution $(T'',\phi_{T''})$ by at most $\frac\eta2L_0$, that is, $\cL(\phi_{T'})\le\cL(\phi_{T''})+\frac\eta2L_0$.
\end{lemma}

\begin{proof}
Every feasible $\phi_{T'}$ is a fair assignment of $P$ to $T'$; hence $(T',\phi_{T'})$ is a relaxed solution.
Let $\mathrm{LP}^*=\sum_{c\in T'}z_c$ denote the optimal value.
We prove
\begin{align}
\cL(\phi_{T'})&\le\mathrm{LP}^*+\frac\eta4L_0 ,\label{eq-step-upper}\\
\mathrm{LP}^*&\le\cL(\phi_{T''})+\frac\eta4L_0\qquad\text{for every relaxed solution with }\cL(\phi_{T''})\le\cL(\phi_T).\label{eq-step-lower}
\end{align}
The lemma follows from these two bounds.
For a relaxed solution with $\cL(\phi_{T''})\le\cL(\phi_T)$, we simply add the two bounds.
Every other relaxed solution satisfies $\cL(\phi_T)<\cL(\phi_{T''})$, and then the case $(T'',\phi_{T''})=(T,\phi_T)$, which is covered by the previous sentence, gives $\cL(\phi_{T'})\le\cL(\phi_T)+\frac\eta2L_0<\cL(\phi_{T''})+\frac\eta2L_0$.

\emph{Proof of~\eqref{eq-step-upper}.}
Fix $c\in T'$.
The constraints of LP~\eqref{eq-step-lp} give
\begin{align*}
z_c&\ge\max_{v\in\cV}\sum_{p\in P}\phi_{T'}(p,c)\Bigl(\lVert p_\lambda-c\rVert^2-\lVert u_p-v\rVert^2\Bigr)\\
&=\sum_{p\in P}\phi_{T'}(p,c)\lVert p_\lambda-c\rVert^2-\min_{v\in\cV}\sum_{p\in P}\phi_{T'}(p,c)\lVert u_p-v\rVert^2 .
\end{align*}
Since $\cV=\WCS_{\eta/4}(\cU)$, \eqref{eq-weighted-candidates} with the weights $\mu(u_p)=\phi_{T'}(p,c)$ yields
\begin{align*}
\min_{v\in\cV}\sum_{p\in P}\phi_{T'}(p,c)\lVert u_p-v\rVert^2
&\le\Bigl(1+\frac\eta4\Bigr)\min_{v\in\R^d}\sum_{p\in P}\phi_{T'}(p,c)\lVert u_p-v\rVert^2\\
&\le\min_{v\in\R^d}\sum_{p\in P}\phi_{T'}(p,c)\lVert u_p-v\rVert^2+\frac\eta4\sum_{p\in P}\phi_{T'}(p,c)\lVert u_p\rVert^2 ,
\end{align*}
where in the second step we bound one copy of the minimum by its value at $v=0$.
We combine the two displays and bound $\sum_p\phi_{T'}(p,c)\lVert p_\lambda-c\rVert^2$ from below by its minimum over all centers in $\R^d$:
\[
\min_{c'\in\R^d}\sum_{p\in P}\phi_{T'}(p,c)\lVert p_\lambda-c'\rVert^2-\min_{v\in\R^d}\sum_{p\in P}\phi_{T'}(p,c)\lVert u_p-v\rVert^2\le z_c+\frac\eta4\sum_{p\in P}\phi_{T'}(p,c)\lVert u_p\rVert^2 .
\]
We sum over $c\in T'$; the left-hand side becomes $\cL(\phi_{T'})$ by~\eqref{eq-linearized-columns}, and since the rows of $\phi_{T'}$ sum to one, the error terms on the right-hand side add up to $\frac\eta4\sum_p\lVert u_p\rVert^2$:
\[
\cL(\phi_{T'})\le\mathrm{LP}^*+\frac\eta4\sum_{p\in P}\lVert u_p\rVert^2=\mathrm{LP}^*+\frac\eta4\lambda^2C_\nu\le\mathrm{LP}^*+\frac\eta4L_0 ,
\]
because $\lVert u_p\rVert^2=\lambda^2\lVert\nu_p-p\rVert^2$, $\lambda<1$, and item~1 of Lemma~\ref{lem-bounded-iterates} gives $C_\nu\le L_0$.

\emph{Proof of~\eqref{eq-step-lower}.}
Let $(T'',\phi_{T''})$ be a relaxed solution with $\cL(\phi_{T''})\le\cL(\phi_T)$; since $\cL(\phi_T)\le3L_0$ by item~1 of Lemma~\ref{lem-bounded-iterates}, item~3 gives $C_{T''}\le10L_0$.
For each $t\in T''$, \eqref{eq-weighted-candidates} for $T'=\WCS_{\eta/88}(P_\lambda)$ with the weights $\mu(p_\lambda)=\phi_{T''}(p,t)$ gives a candidate $c_t\in T'$ such that
\begin{equation}\label{eq-step-center}
\sum_{p\in P}\phi_{T''}(p,t)\lVert p_\lambda-c_t\rVert^2\le\Bigl(1+\frac\eta{88}\Bigr)\min_{c\in\R^d}\sum_{p\in P}\phi_{T''}(p,t)\lVert p_\lambda-c\rVert^2 .
\end{equation}
We combine all columns that share the same candidate $c_t$:
\[
\widehat\phi(p,c)=\sum_{t\in T'':\,c_t=c}\phi_{T''}(p,t)\qquad(p\in P,\ c\in T').
\]
By the argument in the proof of Lemma~\ref{lem-fair-routing}, $\widehat\phi$ is a fair assignment of $P$ to $T'$; together with
\[
\widehat z_c=\max_{v\in\cV}\sum_{p\in P}\widehat\phi(p,c)\Bigl(\lVert p_\lambda-c\rVert^2-\lVert u_p-v\rVert^2\Bigr)
\]
it forms a feasible solution of LP~\eqref{eq-step-lp}, and hence $\mathrm{LP}^*\le\sum_{c\in T'}\widehat z_c$.
We split the sum that defines $\widehat z_c$ over the columns $t$ with $c_t=c$; since the maximum of a sum is bounded by the sum of the individual maxima, and since replacing $\cV$ by $\R^d$ does not decrease a maximum,
\[
\widehat z_c\le\sum_{t:\,c_t=c}\Bigl(\sum_{p\in P}\phi_{T''}(p,t)\lVert p_\lambda-c_t\rVert^2-\min_{v\in\R^d}\sum_{p\in P}\phi_{T''}(p,t)\lVert u_p-v\rVert^2\Bigr).
\]
Summing over $c\in T'$ and using~\eqref{eq-step-center} and~\eqref{eq-linearized-columns}, we obtain
\begin{equation}\label{eq-step-lower-sum}
\mathrm{LP}^*\le\cL(\phi_{T''})+\frac\eta{88}\sum_{t\in T''}\min_{c\in\R^d}\sum_{p\in P}\phi_{T''}(p,t)\lVert p_\lambda-c\rVert^2 .
\end{equation}
It remains to bound the last sum.
Fix $t\in T''$, whose centroid in $\phi_{T''}$ is $\pi(t)$.
Since $p_\lambda-\pi(t)=(p-\pi(t))-u_p$, the inequality $\lVert x-x'\rVert^2\le2\lVert x\rVert^2+2\lVert x'\rVert^2$ gives
\begin{align*}
\min_{c\in\R^d}\sum_{p\in P}\phi_{T''}(p,t)\lVert p_\lambda-c\rVert^2
&\le\sum_{p\in P}\phi_{T''}(p,t)\lVert p_\lambda-\pi(t)\rVert^2\\
&\le2\sum_{p\in P}\phi_{T''}(p,t)\lVert p-\pi(t)\rVert^2+2\sum_{p\in P}\phi_{T''}(p,t)\lVert u_p\rVert^2 .
\end{align*}
We sum over $t$; since every row of $\phi_{T''}$ sums to one, the sum in~\eqref{eq-step-lower-sum} is at most
\[
2C_{T''}+2\sum_{p\in P}\lVert u_p\rVert^2=2C_{T''}+2\lambda^2C_\nu\le20L_0+2L_0=22L_0 ,
\]
where we used $C_{T''}\le10L_0$, $\lambda<1$ and item~1 of Lemma~\ref{lem-bounded-iterates}.
Consequently, $\mathrm{LP}^*\le\cL(\phi_{T''})+\frac\eta{88}\cdot22L_0=\cL(\phi_{T''})+\frac\eta4L_0$.
\end{proof}

Each update that uses this solution decreases the drift-regularized cost by a fixed amount, which bounds the number of updates (Figure~\ref{fig-descent}).

\begin{figure}[!htb]
\centering
\begin{tikzpicture}[fig,x=8cm,y=3cm]
  \draw[figaxis,->] (0,0.9) -- (1.08,0.9) node[right,text=black] {$\theta$};
  \draw[figaxis,->] (0,0.9) -- (0,2.5) node[above,text=black] {$\Psi(\phi_{T_\theta})$};
  \draw[figgray,densely dashed,line width=0.6pt,domain=0:0.6,variable=\t] plot (\t,{2.2-2*\t});
  \draw[figred,densely dashed,line width=0.6pt,domain=0:0.45,smooth,samples=40,variable=\t] plot (\t,{2.2-2*\t+3*\t*\t});
  \draw[figblue,line width=0.9pt,domain=0:1,smooth,samples=60,variable=\t] plot (\t,{2.2-2*\t+1.2*\t*\t});
  \draw[figgray!70,densely dotted,line width=0.5pt] (0,2.2) -- (0.3333,2.2);
  \draw[figgray!70,densely dotted,line width=0.5pt] (0.3333,0.9) -- (0.3333,1.8667);
  \draw[->,figred,line width=0.7pt] (0.3333,2.2) -- (0.3333,1.8667);
  \node[figpt,fill=figblue] at (0,2.2) {};
  \node[figpt,fill=figblue] at (1,1.4) {};
  \node[figlab,left=3pt] at (0,2.2) {$\Psi(\phi_T)$};
  \node[figlab,right=3pt] at (1,1.4) {$\Psi(\phi_{T'})$};
  \draw[figgray,line width=0.5pt] (0.3333,0.88) -- (0.3333,0.92);
  \node[figlabs,below=2pt] at (0,0.9) {$0$};
  \node[figlabs,below=2pt] at (0.3333,0.9) {$\theta=\frac\eta{72}$};
  \node[figlabs,below=2pt] at (1,0.9) {$1$};
  \node[figlabs,text=figred,anchor=west,align=left] at (0.345,2.09) {decrease\\$\ge\frac{\eta^2}{288}L_0$};
  \node[figlabs,text=figred,anchor=north west] at (0.46,1.88) {upper bound (Lemma~\ref{lem-descent})};
  \node[figlabs,text=figblue,anchor=north] at (0.8,1.35) {$\Psi(\phi_{T_\theta})$ (Lemma~\ref{lem-mixing})};
  \node[figlabs,text=figgray,anchor=west] at (0.61,1.0) {slope $-g$};
\end{tikzpicture}
\caption{One update of the first stage of Algorithm~\ref{alg-main}.
Along the mixture of the current relaxed solution ($\theta=0$) with the solution $(T',\phi_{T'})$ of the linear program ($\theta=1$), the drift-regularized cost is the convex quadratic function $\Psi(\phi_{T_\theta})=\Psi(\phi_T)-\theta g+\lambda\theta^2\sum_p\lVert\nu_p-\nu_p(\phi_{T'})\rVert^2$ of $\theta$ (Lemma~\ref{lem-mixing}, blue), whose slope at $\theta=0$ is $-g$, where $g=\cL(\phi_T)-\cL(\phi_{T'})$ (gray).
The proof of Lemma~\ref{lem-descent} bounds the quadratic term by $18\theta^2L_0$ (red), so if $g>\frac\eta2L_0$, the step $\theta=\frac\eta{72}$ decreases $\Psi(\phi_T)$ by at least $\frac{\eta^2}{288}L_0$.
If $g\le\frac\eta2L_0$, the loop stops, and Lemma~\ref{lem-termination} shows that the current solution is $(\eta L_0)$-stationary.}
\label{fig-descent}
\end{figure}

\begin{lemma}[Descent]\label{lem-descent}
Suppose that $\Psi(\phi_T)\le(1+\lambda)L_0$, let $(T',\phi_{T'})$ be the relaxed solution of Lemma~\ref{lem-step-lp}, and assume that the gap $g=\cL(\phi_T)-\cL(\phi_{T'})$ is larger than $\frac\eta2L_0$.
Then, for $\theta=\frac\eta{72}$, the $\theta$-mixture $(T_\theta,\phi_{T_\theta})$ of the two relaxed solutions, whose candidate set is $T_\theta=T\uplus T'$, satisfies
\[
\Psi(\phi_{T_\theta})\le\Psi(\phi_T)-\frac{\eta^2}{288}L_0 .
\]
\end{lemma}

\begin{proof}
We first bound the quadratic term of~\eqref{eq-mixing-identity}.
Lemma~\ref{lem-step-lp} with $(T'',\phi_{T''})=(T,\phi_T)$ and item~1 of Lemma~\ref{lem-bounded-iterates} give
\[
\cL(\phi_{T'})\le\cL(\phi_T)+\frac\eta2L_0\le\Bigl(3+\frac1{20}\Bigr)L_0\le\frac{16}5L_0 ,
\]
since $\eta\le\frac1{10}$; hence $C_{T'}\le10L_0$ by item~3 of Lemma~\ref{lem-bounded-iterates}, and $C_\nu(\phi_{T'})\le C_{T'}\le10L_0$ by Lemma~\ref{lem-weighted-mean}.
The triangle inequality for the norm $(\sum_p\lVert x_p\rVert^2)^{1/2}$ and item~1 of Lemma~\ref{lem-bounded-iterates} then give
\begin{align*}
\Bigl(\sum_{p\in P}\bigl\lVert\nu_p-\nu_p(\phi_{T'})\bigr\rVert^2\Bigr)^{1/2}
&\le\Bigl(\sum_{p\in P}\lVert\nu_p-p\rVert^2\Bigr)^{1/2}+\Bigl(\sum_{p\in P}\lVert p-\nu_p(\phi_{T'})\rVert^2\Bigr)^{1/2}\\
&=\sqrt{C_\nu}+\sqrt{C_\nu(\phi_{T'})}\le(1+\sqrt{10})\sqrt{L_0} ,
\end{align*}
and therefore $\sum_p\lVert\nu_p-\nu_p(\phi_{T'})\rVert^2\le(11+2\sqrt{10})L_0<18L_0$.
Finally, Lemma~\ref{lem-mixing} with $\theta=\frac\eta{72}$, together with $g>\frac\eta2L_0$ and $\lambda<1$, gives
\begin{align*}
\Psi(\phi_{T_\theta})&=\Psi(\phi_T)-\theta g+\lambda\theta^2\sum_{p\in P}\bigl\lVert\nu_p-\nu_p(\phi_{T'})\bigr\rVert^2\\
&\le\Psi(\phi_T)-\frac\eta{72}\cdot\frac\eta2L_0+\frac{\eta^2}{5184}\cdot18L_0
=\Psi(\phi_T)-\frac{\eta^2}{288}L_0 .
\end{align*}
\end{proof}

\begin{lemma}[Termination]\label{lem-termination}
The loop of Algorithm~\ref{alg-main} updates the relaxed solution at most $\lceil576/\eta^2\rceil$ times, and the relaxed solution $(T,\phi_T)$ at termination is $(\eta L_0)$-stationary.
\end{lemma}

\begin{proof}
All relaxed solutions maintained by the loop satisfy $\Psi(\phi_T)\le(1+\lambda)L_0$: the initial one by Lemma~\ref{lem-initial}, and every update, which is performed only if $g>\frac\eta2L_0$, decreases $\Psi(\phi_T)$ by at least $\eta^2L_0/288$ by Lemma~\ref{lem-descent}.
Since $\Psi(\phi_T)\ge0$ for every relaxed solution and $\lambda<1$, there are at most $(1+\lambda)L_0\big/\frac{\eta^2L_0}{288}\le\frac{576}{\eta^2}$ updates.

When the loop terminates, its last linear program has produced $(T',\phi_{T'})$ with $g\le\frac\eta2L_0$, and for every relaxed solution $(T'',\phi_{T''})$, Lemma~\ref{lem-step-lp} yields
\[
\cL(\phi_T)=g+\cL(\phi_{T'})\le\frac\eta2L_0+\cL(\phi_{T''})+\frac\eta2L_0=\cL(\phi_{T''})+\eta L_0 ,
\]
that is, $(T,\phi_T)$ is $(\eta L_0)$-stationary.
\end{proof}

\subsection{Proof of Theorem~\ref{thm-main}}\label{sec-proof-main}

\begin{proof}[Proof of Theorem~\ref{thm-main}]
When the loop of Algorithm~\ref{alg-main} terminates, its current relaxed solution $(T,\phi_T)$ is $\tau$-stationary with $\tau=\eta L_0$ (Lemma~\ref{lem-termination}).
Consider an arbitrary $(\alpha,\beta)$-fair fractional solution $(S,\phi_S)$ with $|S|\le k$.
Moving each of its centers to the weighted centroid of its column turns it, as explained in Section~\ref{sec-comparison}, into a comparison solution $(S^\circ,\phi^\circ)$ with $|S^\circ|\le k$ and $\Cost(P,S^\circ,\phi^\circ)\le\Cost(P,S,\phi_S)$.
We apply Lemma~\ref{lem-merge-bound}, omitting the nonpositive term $-2\rho\,C_\nu$, and Lemma~\ref{lem-initial} to this comparison solution:
\[
\Cost(P,S_\pi,\phi^*_{S_\pi})\le(2\rho+1)\Cost(P,S^\circ,\phi^\circ)+(1+\rho)\eta L_0\le\bigl(2\rho+1+(1+\rho)\eta(1+\eta)\bigr)\Cost(P,S^\circ,\phi^\circ).
\]
Since $\eta\le\epsilon/(3\rho)$, $\eta\le\frac1{10}$ and $1+\rho\le2\rho$,
\[
2\rho+1+(1+\rho)\eta(1+\eta)\le2\rho+1+2\rho\cdot\frac{\epsilon}{3\rho}\cdot\frac{11}{10}=2\rho+1+\frac{11}{15}\epsilon\le2\rho+1+\epsilon ,
\]
and hence the output costs at most $(2\rho+1+\epsilon)\Cost(P,S,\phi_S)$.
Since $(S,\phi_S)$ was arbitrary, the output costs at most $(2\rho+1+\epsilon)\OPTfrac$, and $\OPTfrac\le\OPT$ gives the second inequality of~\eqref{eq-main-guarantee}.
The $k$-means algorithm returns at most $k$ centers, so $|S_\pi|\le k$; by Lemma~\ref{lem-fair-routing}, the fair assignment LP of $S_\pi$ is feasible, and the algorithm solves it exactly, so $\phi^*_{S_\pi}$ satisfies all fairness constraints.
No step of Algorithm~\ref{alg-main} uses randomness, and Proposition~\ref{prop-running-time} bounds the running time.
Finally, in fixed dimension, the PTAS of Lemma~\ref{lem-ptas} with $\epsilon'=\epsilon$ gives $\rho=1+\epsilon$, and the ratio becomes $2(1+\epsilon)+1+\epsilon=3+3\epsilon=3+O(\epsilon)$.
\end{proof}

\subsection{Running Time}\label{sec-running-time}

This subsection proves the following bound on the running time of Algorithm~\ref{alg-main}.\footnote{The running time could be reduced further by first replacing $P$ by a coreset for fair clustering~\cite{huang2019coresets,schmidt2020fair,braverman2022power,DBLP:journals/jcss/BandyapadhyayFS24}. The guarantees of these coresets are stated for settings that differ from ours, for instance for integral assignments~\cite{huang2019coresets} or for randomized constructions~\cite{braverman2022power,DBLP:journals/jcss/BandyapadhyayFS24}, and would have to be extended to our setting.}

\begin{proposition}[Running time]\label{prop-running-time}
Algorithm~\ref{alg-main} calls the $k$-means algorithm once and solves $N_{\mathrm{LP}}\le\lceil576/\eta^2\rceil+3=O((\rho/\epsilon)^2)$ linear programs: one for $T_0$, one in each iteration of the loop, and one for $S_\pi$.
Let $\cT_{\mathrm{LP}}$ be the maximum time needed to solve one of these linear programs, and let $\cT_{\mathrm{km}}$ be the time of the $k$-means call.
\begin{enumerate}
\item In fixed dimension $d$, every linear program has $O(nN)$ variables and $O(N^2+mN)$ constraints, where $N=O\bigl(n(\rho/\epsilon)^{d/2+1}\log(2\rho/\epsilon)\bigr)$, and the running time is
\[
N_{\mathrm{LP}}\cdot\cT_{\mathrm{LP}}+\cT_{\mathrm{km}}
+O\!\left(n^3\Bigl(\frac{\rho}{\epsilon}\Bigr)^{d+4}\log^2\frac{2\rho}{\epsilon}
+n^2d\Bigl(\frac{\rho}{\epsilon}\Bigr)^{d/2+5}\log\frac{2\rho}{\epsilon}\right).
\]
\item In unrestricted dimension, every linear program has $n^{O(\rho/\epsilon)}$ variables and $(n+m)\,n^{O(\rho/\epsilon)}$ constraints, the input of the $k$-means call has $n^{O(\rho/\epsilon)}$ points, and the running time is $N_{\mathrm{LP}}\cdot\cT_{\mathrm{LP}}+\cT_{\mathrm{km}}+d\,n^{O(\rho/\epsilon)}$.
\end{enumerate}
\end{proposition}

We first consider fixed dimension $d$.
We explicitly count the arithmetic operations outside the LP solver and the $k$-means algorithm, and we express the total time in terms of the running times $\cT_{\mathrm{LP}}$ and $\cT_{\mathrm{km}}$ of these two subroutines, which may depend on the bit length.
Since $1/\eta=\max\{10,3\rho/\epsilon\}\le10\rho/\epsilon$, every candidate set used by the algorithm has accuracy $\delta\in\{\eta,\eta/4,\eta/88\}$ with $1/\delta=O(\rho/\epsilon)$, and part~1 of Lemma~\ref{lem-candidates} bounds the sizes $|T_0|$, $|T'|$ and $|\cV|$ of all of them by
\[
N=O\!\left(n\Bigl(\frac{\rho}{\epsilon}\Bigr)^{d/2+1}\log\frac{2\rho}{\epsilon}\right);
\]
each of them can be computed in $O\bigl(n(\rho/\epsilon)\log(n\rho/\epsilon)+N\bigr)$ time.

\paragraph{The loop.}
By Lemma~\ref{lem-termination}, the number of iterations is at most $\lceil576/\eta^2\rceil+1=O((\rho/\epsilon)^2)$, and each update adds at most $N$ candidates; hence the candidate set always satisfies
\[
|T|\le\bigl(\lceil576/\eta^2\rceil+1\bigr)N=O\!\left(n\Bigl(\frac{\rho}{\epsilon}\Bigr)^{d/2+3}\log\frac{2\rho}{\epsilon}\right).
\]
In each iteration, one scan over the entries of $\phi_T$ computes the weights $w(t)$ and the centroids $\pi(t)$, and another scan computes the weighted means $\nu_p$, in $O(n|T|d)$ time even if the matrix is dense; the points $u_p$ and $p_\lambda$ take another $O(nd)$ time.
LP~\eqref{eq-step-lp} has $(n+1)|T'|$ variables and, besides the nonnegativity constraints, at most $|T'||\cV|+n+2m|T'|$ constraints.
Computing its coefficients $\lVert p_\lambda-c\rVert^2-\lVert u_p-v\rVert^2$ requires the $n(|T'|+|\cV|)$ squared distances $\lVert p_\lambda-c\rVert^2$ and $\lVert u_p-v\rVert^2$, each in $O(d)$ time, and one subtraction per coefficient, which takes $O(nNd+nN^2)$ time in total.
To compute the gap $g$, we need $\cL(\phi_T)=C_T+2\lambda C_\nu$, which is available after the two scans, and $\cL(\phi_{T'})$, which takes $O(nNd)$ time once the centroids and weighted means of $\phi_{T'}$ are known; an update takes $O(n(|T|+N))$ time.
Summed over all iterations, the time spent outside the LP solver is
\[
O\!\left(\Bigl(\frac{\rho}{\epsilon}\Bigr)^2\Bigl(n|T|d+nN^2+n\frac{\rho}{\epsilon}\log\frac{n\rho}{\epsilon}\Bigr)\right)
=O\!\left(n^3\Bigl(\frac{\rho}{\epsilon}\Bigr)^{d+4}\log^2\frac{2\rho}{\epsilon}+n^2d\Bigl(\frac{\rho}{\epsilon}\Bigr)^{d/2+5}\log\frac{2\rho}{\epsilon}\right);
\]
here the term on the left-hand side that contains $\log(n\rho/\epsilon)$ is bounded by a constant times the second term on the right-hand side, since $\log n\le n$.

\paragraph{The total time.}
The loop solves one linear program per iteration, and by Lemma~\ref{lem-termination} there are at most $\lceil576/\eta^2\rceil+1$ iterations; together with the linear programs for $T_0$ and $S_\pi$, the algorithm solves at most $\lceil576/\eta^2\rceil+3$ linear programs.
Before the loop, the construction of $T_0$ satisfies the bounds above, and its fair assignment LP has $n|T_0|$ variables and $O((n+m)|T_0|)$ constraints; LP~\eqref{eq-step-lp} has $(n+1)|T'|$ variables and at most $|T'||\cV|+n+2m|T'|$ constraints, so every linear program has $O(nN)$ variables and $O(N^2+mN)$ constraints, as $k\le n\le N$.
The $k$-means subroutine is called once, on $\pi(T)$, a set of at most $|T|$ points, and computing the distances for the final fair assignment LP, which has at most $nk$ variables and $O((n+m)k)$ constraints, takes $O(nkd)$ time.
Adding $N_{\mathrm{LP}}\cdot\cT_{\mathrm{LP}}$ and $\cT_{\mathrm{km}}$ yields the bound of Proposition~\ref{prop-running-time} in fixed dimension.

In unrestricted dimension, part~2 of Lemma~\ref{lem-candidates} gives candidate sets with at most $n^{\lceil1/\delta\rceil}$ points, where $1/\delta\le88/\eta\le880\rho/\epsilon$, each of which can be computed in $O(n^{\lceil1/\delta\rceil}d/\delta)$ time.
With $N=n^{O(\rho/\epsilon)}$, the same analysis shows that the arithmetic outside the subroutines takes $d\,n^{O(\rho/\epsilon)}$ time; every linear program then has $n^{O(\rho/\epsilon)}$ variables and $(n+m)\,n^{O(\rho/\epsilon)}$ constraints, and $\pi(T)$ has $n^{O(\rho/\epsilon)}$ points.
Hence all subproblems have polynomial size for fixed $\epsilon$ and $\rho$, and the whole algorithm runs in polynomial time if the $k$-means subroutine does.

\subsection{Beyond Group Fairness}\label{sec-beyond}

The proof of Theorem~\ref{thm-main} uses only a few properties of the fairness constraints.
This subsection identifies them and extends Theorem~\ref{thm-main} to all constraints that share them.

\begin{definition}[$k$-means with homogeneous cluster constraints]\label{def-homogeneous}
An instance consists of a set $P$ of $n$ weighted points in $\R^d$, an integer $k\ge1$, a number $r\ge1$ of constraints, and coefficients $a_j(p)$ for $p\in P$ and $j\in[r]$.
An assignment matrix $\phi_S$ of $P$ to a set $S\subset\R^d$ of centers is \keyterm{feasible} if
\begin{equation}\label{eq-homogeneous}
\sum_{p\in P}a_j(p)\,\phi_S(p,s)\ge0\qquad\text{for all }s\in S\text{ and }j\in[r].
\end{equation}
The goal is to find at most $k$ centers $S$ and a feasible assignment matrix $\phi_S$ that minimize $\Cost(P,S,\phi_S)$.
As for fair $k$-means, $\OPT$ and $\OPTfrac$ denote the optimal costs of the integral and the fractional version.
\end{definition}

The constraints are called homogeneous because the right-hand side of~\eqref{eq-homogeneous} is zero.
Hence a feasible column remains feasible when it is multiplied by a nonnegative number, and the sum of two feasible columns is feasible.
The algorithm relies on both properties: the $\theta$-mixtures of Section~\ref{sec-algorithm} scale columns, and the routing of Section~\ref{sec-routing} adds the columns that it sends to the same center.
For example, if every point is red or blue, the requirement that every cluster contain at least as many red points as blue points is the single constraint ($r=1$) with $a_1(p)=1$ for red points and $a_1(p)=-1$ for blue points.
Similarly, for two groups $P^{(i)}$ and $P^{(i')}$ and a number $\kappa>0$, the bound $w^{(i)}(s)\le\kappa\,w^{(i')}(s)$ on the ratio of their weights in every cluster is the constraint with $a_1(p)=w(p)\bigl(\kappa\,\ind{p\in P^{(i')}}-\ind{p\in P^{(i)}}\bigr)$.

\begin{remark}[Fair $k$-means as a special case]\label{rem-fair-special}
Fair $k$-means is the special case of Definition~\ref{def-homogeneous} with $r=2m$ constraints, a lower and an upper bound for every group.
To see this, recall from~\eqref{eq-cluster-weight} that both weights in the fairness constraint~\eqref{eq-fairness} are linear in the column of $s$:
\[
w(s)=\sum_{p\in P}w(p)\,\phi_S(p,s),
\qquad
w^{(i)}(s)=\sum_{p\in P}w(p)\,\ind{p\in P^{(i)}}\,\phi_S(p,s).
\]
Hence the lower bound $\beta_i\,w(s)\le w^{(i)}(s)$ is equivalent to
\[
\sum_{p\in P}w(p)\bigl(\ind{p\in P^{(i)}}-\beta_i\bigr)\,\phi_S(p,s)\ge0,
\]
and the upper bound $w^{(i)}(s)\le\alpha_i\,w(s)$ is equivalent to
\[
\sum_{p\in P}w(p)\bigl(\alpha_i-\ind{p\in P^{(i)}}\bigr)\,\phi_S(p,s)\ge0 .
\]
These are constraints of the form~\eqref{eq-homogeneous}, whose coefficients $a_j(p)$ are the factors in front of $\phi_S(p,s)$; for unweighted points, $w(p)=1$.
In the first sum, every point of $P^{(i)}$ contributes a nonnegative term and every other point a nonpositive one, and the sum is nonnegative exactly when $P^{(i)}$ has at least a $\beta_i$ fraction of the weight of the cluster; the second sum bounds this fraction by $\alpha_i$ in the same way.
Both formulations have the same feasible assignment matrices, and hence the same optimal costs.
\end{remark}

\begin{theorem}[Homogeneous cluster constraints]\label{thm-homogeneous}
An instance of Definition~\ref{def-homogeneous} has a feasible solution if and only if $\sum_{p\in P}a_j(p)\ge0$ for every $j\in[r]$, and in this case it also has a feasible solution with an integral assignment matrix.
Let $\rho\ge1$, and assume access to a deterministic polynomial-time $\rho$-approximation algorithm for weighted Euclidean $k$-means.
Then, for every instance with a feasible solution and every $\epsilon\in(0,1]$, a deterministic polynomial-time algorithm computes at most $k$ centers $S$ together with a feasible fractional assignment matrix $\phi_S$ of cost
\[
\Cost(P,S,\phi_S)\le(2\rho+1+\epsilon)\,\OPTfrac\le(2\rho+1+\epsilon)\,\OPT .
\]
Besides a single call to the $k$-means algorithm, it solves $O((\rho/\epsilon)^2)$ linear programs, and Proposition~\ref{prop-running-time} holds with $m$ replaced by~$r$.
\end{theorem}

\begin{proof}
If $\phi_S$ is feasible, summing~\eqref{eq-homogeneous} over $s\in S$ and using that every row of $\phi_S$ sums to one gives $\sum_{p\in P}a_j(p)\ge0$ for every $j$.
Conversely, if these inequalities hold, then assigning all of $P$ to a single center is feasible and integral; this is the argument of Section~\ref{sec-prelim} for fairness.

For the algorithm, we read ``fair'' as ``feasible'' throughout Sections~\ref{sec-algorithm} and~\ref{sec-analysis}: the fair assignment LP of a candidate set and LP~\eqref{eq-step-lp} contain the constraints~\eqref{eq-homogeneous} for every candidate in place of the fairness constraints.
Algorithm~\ref{alg-main} and the proofs of these two sections use the fairness constraints only through the following three properties, which~\eqref{eq-homogeneous} shares.
\begin{enumerate}
\item Every constraint involves the column of a single center.
Hence whether an assignment matrix satisfies the constraints depends on each of its columns separately, and moving a center without changing its column preserves them.
\item The constraints are linear inequalities in the entries of the assignment matrix.
Hence the fair assignment LP and LP~\eqref{eq-step-lp} remain linear programs, with $r$ instead of $2m$ such constraints per candidate.
\item The constraints are homogeneous.
Hence they are preserved when a column is multiplied by a nonnegative number, as in the $\theta$-mixture of Section~\ref{sec-algorithm}, and when columns are added, as in the proofs of Lemmas~\ref{lem-fair-routing}, \ref{lem-initial} and~\ref{lem-step-lp}.
\end{enumerate}
The cost, the candidate sets and the $k$-means call do not involve the constraints.
Therefore every lemma of Sections~\ref{sec-algorithm} and~\ref{sec-analysis} and the proof of Theorem~\ref{thm-main} remain valid, which gives the bound on the cost.
Replacing $2m$ by $r$ in the proof of Proposition~\ref{prop-running-time} gives the running time.
\end{proof}

\begin{corollary}[Approximation threshold of the class]\label{cor-homogeneous-threshold}
Let $d\ge1$ be fixed.
For every $\epsilon\in(0,1]$, Theorem~\ref{thm-homogeneous} with the PTAS of Lemma~\ref{lem-ptas} computes in polynomial time a feasible fractional solution of cost at most $(3+3\epsilon)\,\OPTfrac$ for every instance in $\R^d$ with a feasible solution.
Conversely, for every fixed $\delta\in(0,1]$, if $\mathrm P\ne\mathrm{NP}$, then no polynomial-time algorithm returns, for every such instance, at most $k$ centers and a feasible fractional assignment matrix of cost at most $(3-\delta)\,\OPT$.
\end{corollary}

\begin{proof}
With $\rho=1+\epsilon$, Theorem~\ref{thm-homogeneous} gives the ratio $3+3\epsilon$, as in the proof of Theorem~\ref{thm-main}.
For the lower bound, by Remark~\ref{rem-fair-special}, every instance of fair $k$-means is an instance of Definition~\ref{def-homogeneous} with the same optimal cost.
An algorithm as in the statement would therefore contradict Theorem~\ref{thm-hardness} for $d=1$ and Remark~\ref{rem-hardness-scope} for $d\ge2$.
\end{proof}

\section{Hardness of Approximation}\label{sec-hardness}

We prove Theorem~\ref{thm-hardness} by a reduction from label cover.
Section~\ref{sec-intermediate} constructs an instance on the line with weighted points, in which constraints on every cluster take the place of the groups.
If the label cover instance is satisfiable, the constructed instance has a solution of low cost (Section~\ref{sec-satisfiable}). If instead the value is small, every solution---fractional or integral---costs almost three times as much (Section~\ref{sec-small-value}).
Section~\ref{sec-groups} expresses the constraints by groups of points of unit weight, with the construction deferred to Appendix~\ref{app-groups}, and Section~\ref{sec-proof-hardness} completes the proof.

\subsection{Label Cover}\label{sec-label-cover}

We use label cover as defined in Definition~\ref{def-label-cover}; here and below, an edge $e=(u,v)$ always has $u\in U$ and $v\in V$.
If the graph is biregular, then every $u\in U$ lies on $|E|/|U|$ edges and every $v\in V$ on $|E|/|V|$ edges.

\begin{lemma}[Label cover, \cite{arora1998proof,feige1998threshold,raz1998parallel}]\label{lem-label-cover}
For every $\sigma\in(0,1]$ there is a polynomial-time reduction that maps every 3CNF formula to a label cover instance on a biregular graph, with label sets whose sizes depend only on $\sigma$, such that satisfiable formulas are mapped to instances of value $1$ and unsatisfiable formulas to instances of value at most $\sigma$.
\end{lemma}

This is the standard form of the label cover theorem; it follows from the PCP theorem for 3SAT formulas in which every variable occurs in exactly five clauses~\cite{feige1998threshold} and from Raz's parallel repetition theorem~\cite{raz1998parallel}.
For the rest of the section, we fix a label cover instance produced by Lemma~\ref{lem-label-cover} with $\sigma=\delta^2/128$; since $\delta$ is fixed, $|\mathcal A|$ and $|\mathcal B|$ are constants.

\subsection{The Constrained Instance}\label{sec-intermediate}

\paragraph{Points.}
Let $R=10(|U|+|V|)/\sigma$.
We call the elements $(u,a)$ of $U\times\mathcal A$ and $(v,b)$ of $V\times\mathcal B$ the \keyterm{labels}.
The labels are numbered, and the $j$-th label is assigned the interval $[(10j-1)R,(10j+1)R]$ of radius $R$ around $10jR$, which we denote by $I_{u,a}$ or $I_{v,b}$ if the label is $(u,a)$ or $(v,b)$; the intervals are pairwise disjoint.
The \keyterm{heavy points} have weight $M$, a large integer whose value is fixed in~\eqref{eq-heavy-weight} of Section~\ref{sec-satisfiable} so that the light points introduced next contribute little to the cost:
\begin{itemize}
\item every label $(u,a)$ of a vertex $u\in U$ has two points at distance $\frac12$ from the midpoint of $I_{u,a}$, the point $x^-_{u,a}$ to the left and the point $x^+_{u,a}$ to the right, which form the \keyterm{pair} of $(u,a)$;
\item every label $(v,b)$ of a vertex $v\in V$ has the single point $x_{v,b}$ at the midpoint of $I_{v,b}$.
\end{itemize}
The \keyterm{light points} have weight $1$ and all lie at the origin: there is a \keyterm{selector} $x_u$ for each $u\in U$ and a selector $x_v$ for each $v\in V$, and a \keyterm{label point} $x_{e,b}$ for each edge $e\in E$ and each $b\in\mathcal B$.
Together, the heavy and the light points form the weighted set $X$.
As everywhere in this paper, $k$ denotes the maximum number of centers that a solution may use.
In this construction we call it the \keyterm{budget}, because it is chosen to make centers scarce:
\begin{equation}\label{eq-budget}
k=(2|\mathcal A|-1)|U|+|\mathcal B||V| ,
\end{equation}
which is $|U|$ less than the number $2|\mathcal A||U|+|\mathcal B||V|$ of heavy points.
The budget allows the following placement.
Every $u\in U$ receives one center at the midpoint of the pair of one of its labels and two centers, one at each point, for the pair of each of its other $|\mathcal A|-1$ labels, which is $2|\mathcal A|-1$ centers in total, and every $v\in V$ receives one center at each of its $|\mathcal B|$ heavy points.
Selecting a label for $u$ means serving the pair of this label by the single center at its midpoint, at cost $M/2$, while the pairs of the other labels of $u$ are served at no cost (Figure~\ref{fig-budget}); the light points and the consistency constraints of the next paragraph connect these choices along the edges (Figure~\ref{fig-pair}).
If every vertex of $U$ selects a label in this way, the heavy points cost $M|U|/2$, and Sections~\ref{sec-satisfiable} and~\ref{sec-small-value} compare the cost of every solution with this value.

\begin{figure}[t]
\centering
\begin{tikzpicture}[fig]
\draw[figaxis,->] (-0.4,0) -- (13.3,0) node[right,text=black] {$\R$};
\draw[figgray] (0,0.07) -- (0,-0.07) node[figlab,below=1pt,text=black] {$0$};
\draw[figgray!60,densely dotted,line width=0.5pt] (0,0.08) -- (0,1.55);
\foreach \y/\lab in {1.4/{x_u},1.0/{x_v},0.6/{x_{e,b}}} {
  \node[figpt,minimum size=3.4pt] at (0,\y) {};
  \node[figlab,anchor=east] at (-0.12,\y) {$\lab$};
}
\node[figlabs,text=figgray,anchor=south west] at (0.1,1.55) {wt~$1$};
\draw[figband] (2.35,-0.19) rectangle (3.65,0.19);
\node[figpt] at (2.7,0) {}; \node[figpt] at (3.3,0) {};
\node[figctr] at (2.7,0) {}; \node[figctr] at (3.3,0) {};
\node[figlab] at (3.0,0.55) {$x^\pm_{u,1}$};
\node[figlab] at (3.0,-0.48) {$I_{u,1}$};
\draw[figband] (5.05,-0.19) rectangle (6.55,0.19);
\node[figpt] at (5.5,0) {}; \node[figpt] at (6.1,0) {};
\node[figopt] at (5.8,0) {};
\node[figlab] at (5.8,0.55) {$x^\pm_{u,2}$};
\node[figlab] at (5.8,-0.48) {$I_{u,2}$};
\node[figlabs,text=figred,anchor=south] at (5.8,0.85) {selected};
\draw[figband] (7.85,-0.19) rectangle (9.35,0.19);
\node[figpt] at (8.3,0) {}; \node[figpt] at (8.9,0) {};
\node[figctr] at (8.3,0) {}; \node[figctr] at (8.9,0) {};
\node[figlab] at (8.6,0.55) {$x^\pm_{u,3}$};
\node[figlab] at (8.6,-0.48) {$I_{u,3}$};
\draw[figgray,decorate,decoration={brace,mirror,amplitude=4pt}] (2.25,-0.85) -- (9.45,-0.85)
  node[figlab,midway,below=5pt,text=black,align=center,text width=7cm] {pairs of $u$:\\$2+1+2=2|\mathcal A|-1$ centers\\heavy cost $M/2$};
\draw[figband] (10.35,-0.19) rectangle (11.25,0.19);
\node[figpt] at (10.8,0) {}; \node[figctr] at (10.8,0) {};
\node[figlab] at (10.8,0.55) {$x_{v,1}$};
\node[figlab] at (10.8,-0.48) {$I_{v,1}$};
\draw[figband] (11.85,-0.19) rectangle (12.75,0.19);
\node[figpt] at (12.3,0) {}; \node[figctr] at (12.3,0) {};
\node[figlab] at (12.3,0.55) {$x_{v,2}$};
\node[figlab] at (12.3,-0.48) {$I_{v,2}$};
\draw[figgray,decorate,decoration={brace,mirror,amplitude=4pt}] (10.25,-0.85) -- (12.85,-0.85)
  node[figlab,midway,below=5pt,text=black,align=center,text width=2.6cm] {points of $v$:\\$|\mathcal B|$ centers\\cost $0$};
\end{tikzpicture}
\caption{Budget allocation in the completeness case ($|\mathcal A|=3$, $|\mathcal B|=2$), not to scale.
Heavy points (filled disks, weight~$M$) and their centers (circles): vertex~$u$ selects label~$2$ by placing one center (red circle) at the midpoint of its pair, at cost~$M/2$; the remaining pairs each receive two centers at their endpoints, at cost~$0$.
Each label of~$v$ receives one center at its heavy point.
The selectors and label points at the origin have unit weight; their total cost is at most $\delta M|U|/64$.
The budget is $k=(2|\mathcal A|-1)|U|+|\mathcal B||V|$.}
\label{fig-budget}
\end{figure}

\paragraph{Consistent clusters.}
These constraints link the selections of adjacent vertices to the edge maps.
In every cluster, the fraction of the selector of $u$ must be distributed among the labels of $u$ and the fraction of the selector of $v$ among the labels of $v$, and the label points of an edge $e=(u,v)$ force the two distributions to be related by $f_e$; the following definition makes this precise.
Let $\phi$ be an assignment matrix of $X$ to a center set $S$, so that $\phi(x,s)$ is the fraction of $x$ assigned to the cluster of $s$.
Given an edge $e=(u,v)$, we call the cluster of $s$ \keyterm{$e$-consistent} if there are nonnegative numbers $\chi_{u,a}$ ($a\in\mathcal A$) and $\chi_{v,b}$ ($b\in\mathcal B$), called \keyterm{witnesses}, that satisfy
\begin{align}
&\chi_{u,a}\le\phi(x^-_{u,a},s),\quad\chi_{u,a}\le\phi(x^+_{u,a},s)\quad(a\in\mathcal A),
&&\sum_{a\in\mathcal A}\chi_{u,a}=\phi(x_u,s),\label{eq-witness-left}\\
&\chi_{v,b}\le\phi(x_{v,b},s)\quad(b\in\mathcal B),
&&\sum_{b\in\mathcal B}\chi_{v,b}=\phi(x_v,s),\label{eq-witness-right}\\
&\phi(x_{e,b},s)=\sum_{a\in\mathcal A:\,f_e(a)=b}\chi_{u,a}+\sum_{b'\in\mathcal B\setminus\{b\}}\chi_{v,b'}
&&(b\in\mathcal B).\label{eq-witness-label}
\end{align}
We regard $\chi_{u,a}$ as the record of the label $(u,a)$ in the cluster: by~\eqref{eq-witness-left}, the witnesses $\chi_{u,a}$ distribute the fraction $\phi(x_u,s)$ of the selector of $u$ among the labels of $u$, and the amount given to $(u,a)$ is at most the fraction of each point of its pair in the cluster.
Similarly, by~\eqref{eq-witness-right}, the witnesses $\chi_{v,b}$ distribute the fraction $\phi(x_v,s)$ of the selector of $v$ among the labels of $v$, and the amount given to $(v,b)$ is at most the fraction of $x_{v,b}$ in the cluster.
By~\eqref{eq-witness-label}, the fraction of the label point $x_{e,b}$ in the cluster is the total record of the labels $(u,a)$ with $f_e(a)=b$ plus the total record of the labels $(v,b')$ with $b'\ne b$; Section~\ref{sec-small-value} uses this to show that the records at the two endpoints of $e$ agree when summed over all clusters.
In the integral case of Lemma~\ref{lem-satisfiable}, where the selector of $u$ is recorded only at $(u,a)$ and that of $v$ only at $(v,b')$, these totals are $\ind{f_e(a)=b}+\ind{b'\ne b}$, and they equal $1$ for every $b$ exactly when $f_e(a)=b'$.

A \keyterm{feasible solution} of the constrained instance is a pair $(S,\phi)$ in which $S\subset\R$ has at most $k$ points and $\phi$ is a fractional or integral assignment matrix of $X$ to $S$ such that, for every edge $e$, all clusters are $e$-consistent; the witnesses may differ between edges and between clusters.
Its cost is $\Cost(X,S,\phi)$ as defined in~\eqref{eq-assignment-cost}.

\begin{figure}[t]
\centering
\begin{tikzpicture}[fig]
\draw[figaxis,->] (-0.35,0) -- (12.75,0) node[right,text=black] {$\R$};
\draw[figgray] (0,0.07) -- (0,-0.07) node[figlab,below=1pt,text=black] {$0$};
\draw[figgray!60,densely dotted,line width=0.5pt] (0,0.08) -- (0,1.5);
\foreach \y/\lab in {1.5/{x_v},1.14/{x_u},0.78/{x_{e,1}},0.42/{x_{e,2}}} {
  \node[figpt,minimum size=3.8pt] at (0,\y) {};
  \node[figlab,anchor=east] at (-0.12,\y) {$\lab$};
}
\foreach \c/\a in {2.3/1,4.5/2,6.7/3} {
  \draw[figband] (\c-0.62,-0.17) rectangle (\c+0.62,0.17);
  \node[figpt] at (\c-0.28,0) {};
  \node[figpt] at (\c+0.28,0) {};
  \node[figlab] at (\c,0.5) {$x^{\pm}_{u,\a}$};
  \node[figlab] at (\c,-0.46) {$I_{u,\a}$};
}
\foreach \c/\b in {9.3/1,11.3/2} {
  \draw[figband] (\c-0.62,-0.17) rectangle (\c+0.62,0.17);
  \node[figpt] at (\c,0) {};
  \node[figlab] at (\c,0.5) {$x_{v,\b}$};
  \node[figlab] at (\c,-0.46) {$I_{v,\b}$};
}
\draw[->,figred,line width=0.7pt,densely dashed] (0.13,1.14) .. controls (1.4,2.05) and (3.6,2.1) .. (4.45,0.85)
  node[pos=0.5,fill=white,inner sep=1.5pt,text=figred] {$\chi_{u,a}$};
\draw[->,figblue,line width=0.7pt,densely dashed] (0.13,1.5) .. controls (2.6,3.2) and (8.1,3.15) .. (9.27,0.85)
  node[pos=0.5,fill=white,inner sep=1.5pt,text=figblue] {$\chi_{v,b}$};
\draw[figgray,decorate,decoration={brace,mirror,amplitude=4pt}] (1.62,-0.76) -- (7.38,-0.76)
  node[figlab,midway,below=5pt,text=black] {labels of $u$, with $f_e(1)=f_e(2)=1$ and $f_e(3)=2$};
\draw[figgray,decorate,decoration={brace,mirror,amplitude=4pt}] (8.62,-0.76) -- (11.98,-0.76)
  node[figlab,midway,below=5pt,text=black] {labels of $v$};
\end{tikzpicture}
\caption{The constrained instance for one edge $e=(u,v)$ with $|\mathcal A|=3$ and $|\mathcal B|=2$, not to scale: the intervals $I_{u,a}$ and $I_{v,b}$ (shaded) have radius $R$, and their midpoints are $10R$ apart.
Each label of $u$ has the pair $x^\pm_{u,a}$ and each label of $v$ has the single point $x_{v,b}$, all of weight $M$; the selectors and the label points lie at the origin, have weight $1$, and are drawn one above the other.
In a cluster, the witnesses $\chi_{u,a}$ distribute the fraction of the selector of $u$ among the pairs of $u$, where each pair receives at most the smaller of the fractions of its two points in the cluster, and the witnesses $\chi_{v,b}$ distribute the fraction of the selector of $v$ among the heavy points of $v$.
Since each label point is assigned exactly once in total, the two allocations agree when summed over all clusters (Equation~\eqref{eq-pushforward}).}
\label{fig-pair}
\end{figure}

The witnesses are auxiliary variables: they do not appear in the assignment $\phi$, but only in the conditions that define consistency. Since the label sets $\mathcal A$ and $\mathcal B$ have constant size, Fourier--Motzkin elimination removes the witnesses and produces a bounded number of homogeneous linear inequalities in the entries of $\phi$ alone. The next lemma makes this precise.

\begin{lemma}[Consistency as linear inequalities]\label{lem-consistency-linear}
There is an integer $K\ge1$, depending only on $|\mathcal A|$ and $|\mathcal B|$, with the following property.
For every edge $e=(u,v)$, there are at most $K$ homogeneous linear inequalities $\sum_{x\in X}c_x\phi(x,s)\ge0$, with integer coefficients $|c_x|\le K$ that are nonzero only at the heavy points of $u$ and $v$, the selectors $x_u$ and $x_v$, and the label points of $e$. A cluster is $e$-consistent if and only if all of these inequalities hold for its assignment.
The inequalities can be computed from $f_e$ in time that depends only on $|\mathcal A|$ and $|\mathcal B|$.
Moreover, a cluster with $\phi(x,s)=1$ for all $x\in X$ is $e$-consistent for every edge $e$.
\end{lemma}

\begin{proof}
Fix $e$ and $s$.
The cluster of $s$ is $e$-consistent if and only if the system that consists of~\eqref{eq-witness-left}--\eqref{eq-witness-label} and the conditions $\chi_{u,a}\ge0$ and $\chi_{v,b}\ge0$ has a solution in the witnesses.
After each equation is written as two inequalities, this is a finite system of homogeneous linear inequalities in the witnesses and in the numbers $\phi(x,s)$ of the points listed in the lemma, which we regard as variables; its coefficients lie in $\{-1,0,1\}$, its size depends only on $|\mathcal A|$ and $|\mathcal B|$, and it depends on $e$ only through $f_e$.
Fourier--Motzkin elimination of the $|\mathcal A|+|\mathcal B|$ witnesses \cite[Section~12.2]{schrijver1986theory} turns it into a finite system of homogeneous linear inequalities in the numbers $\phi(x,s)$ alone, which holds if and only if the original system has a solution in the witnesses.
Each elimination step only adds positive integer multiples of pairs of inequalities, so all coefficients remain integers, and the number of inequalities and the absolute values of their coefficients are bounded in terms of the size of the initial system.
Since there are only $|\mathcal B|^{|\mathcal A|}$ maps from $\mathcal A$ to $\mathcal B$, these quantities are bounded by a constant $K$ that depends only on $|\mathcal A|$ and $|\mathcal B|$; it can be computed by running the elimination for all these maps.
For the last claim, choose $a\in\mathcal A$, set $\chi_{u,a}=\chi_{v,f_e(a)}=1$ and all other witnesses to $0$; then~\eqref{eq-witness-left} and~\eqref{eq-witness-right} hold, and the right-hand side of~\eqref{eq-witness-label} equals $1$ for every $b$.
\end{proof}

\subsection{Completeness}\label{sec-satisfiable}

The light points lie at the origin and the heavy points to its right, so all points of $X$ lie in the interval $[0,\diam(X)]$, and so do all centers used below.
Serving the $|U|+|V|+|\mathcal B||E|$ light points by such centers therefore costs at most $(|U|+|V|+|\mathcal B||E|)\diam(X)^2$, which does not depend on $M$.
We choose the weight of the heavy points as
\begin{equation}\label{eq-heavy-weight}
M=\biggl\lceil\frac{64\,\bigl(|U|+|V|+|\mathcal B||E|\bigr)\diam(X)^2}{\delta|U|}\biggr\rceil,
\end{equation}
so that the light points cost at most $\frac{\delta M|U|}{64}$, a $\frac\delta{32}$ fraction of $\frac{M|U|}2$.

\begin{lemma}[Completeness]\label{lem-satisfiable}
If the label cover instance has value $1$, then some integral feasible solution of the constrained instance has cost at most $(1+\frac\delta{32})\frac{M|U|}2$.
\end{lemma}

\begin{proof}
Let $(a_u)_{u\in U}$, $(b_v)_{v\in V}$ be a labeling that satisfies every edge.
For every $u\in U$, we place one center at the midpoint of the pair of $(u,a_u)$, which we call the center of $u$, and one center at each point of the pairs of the other labels of $u$; for every $v\in V$, we place one center at each point $x_{v,b}$.
This gives $(2|\mathcal A|-1)|U|+|\mathcal B||V|=k$ centers.
Every heavy point is assigned to the center at its position, except that both points of the pair of $(u,a_u)$ are assigned to the center of $u$.
The selector of $u$ is assigned to the center of $u$ and the selector of $v$ to the center at $x_{v,b_v}$; for an edge $e=(u,v)$, the label point $x_{e,b_v}$ is assigned to the center of $u$, and every other label point $x_{e,b}$ of $e$ to the center at $x_{v,b_v}$.

We choose witnesses for an edge $e=(u,v)$ cluster by cluster.
The cluster of the center of $u$ contains the pair of $(u,a_u)$ and the selector of $u$ but not the selector of $v$, and among the label points of $e$ only $x_{e,b_v}$; here we set $\chi_{u,a_u}=1$, and~\eqref{eq-witness-label} holds because $f_e(a_u)=b_v$.
The cluster at $x_{v,b_v}$ contains $x_{v,b_v}$ and the selector of $v$ but not the selector of $u$, and among the label points of $e$ exactly those $x_{e,b}$ with $b\ne b_v$; here we set $\chi_{v,b_v}=1$.
All other witnesses, also those of the remaining clusters, are $0$; with this choice, \eqref{eq-witness-left}--\eqref{eq-witness-label} hold in every cluster.

The only heavy points that are not located at their centers are the points of the pairs of the labels $(u,a_u)$, and each such pair costs $2M(\frac12)^2=\frac M2$; hence the heavy points cost $\frac{M|U|}2$, and the light points cost at most $\frac{\delta M|U|}{64}$.
\end{proof}

\subsection{Soundness}\label{sec-small-value}

Throughout this subsection, the label cover instance has value at most $\sigma$, and we fix a feasible solution $(S,\phi)$ of the constrained instance with cost less than $\frac{3M|U|}2$.
For each edge $e=(u,v)$ and each center $s\in S$, we fix witnesses $\chi^e_{u,a}(s)$ and $\chi^e_{v,b}(s)$ for the $e$-consistency of the cluster of $s$.
The witnesses define distributions on the labels that nearly agree along the edges, so for most vertices of $U$, no single label receives a large fraction of the choice (Lemma~\ref{lem-concentration}).
A pair that carries a fraction $t$ of the choice of its vertex costs about $Mt$, that is, about $M$ per unit of fraction, which is twice the cost $M/2$ of a selected label (Lemma~\ref{lem-pair-cost}); Proposition~\ref{prop-small-value} adds up these bounds.

\paragraph{Witness distributions.}
For an edge $e=(u,v)$, let
\begin{equation}\label{eq-edge-distributions}
\mu^e_u(a)=\sum_{s\in S}\chi^e_{u,a}(s)\quad(a\in\mathcal A),
\qquad
\mu^e_v(b)=\sum_{s\in S}\chi^e_{v,b}(s)\quad(b\in\mathcal B).
\end{equation}
By~\eqref{eq-witness-left}, the entries of $\mu^e_u$ sum to the total fraction of the selector $x_u$ over all clusters, which is $1$; so $\mu^e_u$ is a probability distribution on $\mathcal A$, and by~\eqref{eq-witness-right}, $\mu^e_v$ is a probability distribution on $\mathcal B$.
For a distribution $\mu$ on $\mathcal A$, we write $(f_e)_*\mu$ for its image on $\mathcal B$, that is, $(f_e)_*\mu(b)=\sum_{a:\,f_e(a)=b}\mu(a)$.
Since the label point $x_{e,b}$ is assigned exactly once in total, summing~\eqref{eq-witness-label} over all $s\in S$ gives $1=(f_e)_*\mu^e_u(b)+\bigl(1-\mu^e_v(b)\bigr)$ for every $b$, and hence
\begin{equation}\label{eq-pushforward}
\mu^e_v=(f_e)_*\mu^e_u .
\end{equation}

\paragraph{Canonical distributions and outflow.}
The distribution $\mu^e_u$ depends on the edge $e$; to obtain one distribution per vertex, we use the positions of the centers.
The \keyterm{canonical distribution} $\mu_u$ of a vertex $u\in U$ assigns, for every center $s\in S$, the fraction $\phi(x_u,s)$ of the selector of $u$ to a label $a\in\mathcal A$ whose interval $I_{u,a}$ is nearest to $s$; this defines a probability distribution on $\mathcal A$, and the canonical distribution $\mu_v$ on $\mathcal B$ of a vertex $v\in V$ is defined in the same way.
The two kinds of distributions differ only through heavy points that are assigned to centers outside their intervals.
We call
\begin{equation}\label{eq-outflow}
\Out(u)=\sum_{a\in\mathcal A}\ \sum_{s\in S\setminus I_{u,a}}\phi(x^-_{u,a},s),
\qquad
\Out(v)=\sum_{b\in\mathcal B}\ \sum_{s\in S\setminus I_{v,b}}\phi(x_{v,b},s)
\end{equation}
the \keyterm{outflow} of $u\in U$ and of $v\in V$.
Each point in these sums has weight $M$ and lies within distance $\frac12$ of the midpoint of its interval, so its distance to every center outside this interval is at least $R-\frac12$; hence
\begin{equation}\label{eq-outflow-cost}
\Cost(X,S,\phi)\ge M\Bigl(R-\frac12\Bigr)^2\Bigl(\sum_{u\in U}\Out(u)+\sum_{v\in V}\Out(v)\Bigr).
\end{equation}

\begin{lemma}[Canonical distributions]\label{lem-canonical}
For every edge $e=(u,v)$,
\[
\lVert\mu_u-\mu^e_u\rVert_1\le2\Out(u),
\quad
\lVert\mu_v-\mu^e_v\rVert_1\le2\Out(v),
\quad
\lVert\mu_v-(f_e)_*\mu_u\rVert_1\le2\bigl(\Out(u)+\Out(v)\bigr).
\]
\end{lemma}

\begin{proof}
Consider a center $s\in S$.
Its contribution to $\mu_u$ is the vector with entry $\phi(x_u,s)$ at a single label $a(s)$, namely the label chosen for $s$ in the definition of $\mu_u$, and zeros elsewhere; its contribution to $\mu^e_u$ is the vector $(\chi^e_{u,a}(s))_{a\in\mathcal A}$.
Since $\phi(x_u,s)=\sum_a\chi^e_{u,a}(s)$ by~\eqref{eq-witness-left}, the $\ell_1$ distance between these two vectors is $2\sum_{a\ne a(s)}\chi^e_{u,a}(s)$.
For $a\ne a(s)$, the center $s$ lies outside $I_{u,a}$, since otherwise $I_{u,a}$ would be the nearest interval; together with $\chi^e_{u,a}(s)\le\phi(x^-_{u,a},s)$ from~\eqref{eq-witness-left}, this shows that each term of this sum is at most the fraction of $x^-_{u,a}$ that is assigned to $s$, a center outside $I_{u,a}$.
Summing over $s$ yields the first bound by the triangle inequality.
The second bound follows in the same way from $\chi^e_{v,b}(s)\le\phi(x_{v,b},s)$, and the third from~\eqref{eq-pushforward}, the first two bounds and the fact that the map $\mu\mapsto(f_e)_*\mu$ is nonexpansive in the $\ell_1$ norm.
\end{proof}

By Lemma~\ref{lem-canonical}, every distribution $\mu^e_u$ of a vertex $u$ is within $2\Out(u)$ of its canonical distribution, so the choice of the edge matters little when the outflow is small.
For the rest of this subsection, we fix an edge $e_u$ at every vertex $u\in U$.

\begin{lemma}[Concentration]\label{lem-concentration}
The distributions $\mu^{e_u}_u$ satisfy
\[
\sum_{u\in U}\sum_{a\in\mathcal A}\mu^{e_u}_u(a)^2\le2\sigma|U| .
\]
\end{lemma}

\begin{proof}
Consider the random labeling that gives every $u\in U$ a label $a_u$ that maximizes $\mu_u$, and every $v\in V$, independently, a label $b_v$ drawn from $\mu_v$.
An edge $e=(u,v)$ is satisfied with probability $\mu_v(f_e(a_u))$, and since $(f_e)_*\mu_u$ puts mass at least $\max_a\mu_u(a)$ on $f_e(a_u)$, Lemma~\ref{lem-canonical} gives
\[
\mu_v(f_e(a_u))\ge\max_a\mu_u(a)-\lVert\mu_v-(f_e)_*\mu_u\rVert_1\ge\max_a\mu_u(a)-2\bigl(\Out(u)+\Out(v)\bigr).
\]
The expected number of satisfied edges is at most $\sigma|E|$, because no labeling satisfies more edges.
We sum the bound over all edges, in which every $u$ appears $|E|/|U|$ times and every $v$ appears $|E|/|V|$ times, and divide by $|E|/|U|$; this gives
\begin{equation}\label{eq-concentration-avg}
\sum_{u\in U}\max_a\mu_u(a)\le\sigma|U|+2\sum_{u\in U}\Out(u)+\frac{2|U|}{|V|}\sum_{v\in V}\Out(v) .
\end{equation}
For every $u$, we have
\[
\sum_{a\in\mathcal A}\mu^{e_u}_u(a)^2\le\max_a\mu^{e_u}_u(a)\le\max_a\mu_u(a)+2\Out(u),
\]
where the first inequality holds because $\mu^{e_u}_u$ is a probability distribution and the second follows from Lemma~\ref{lem-canonical}.
Summing over $u$ and using~\eqref{eq-concentration-avg}, we obtain
\[
\sum_{u\in U}\sum_{a\in\mathcal A}\mu^{e_u}_u(a)^2\le\sigma|U|+\Bigl(4+\frac{2|U|}{|V|}\Bigr)\Bigl(\sum_{u\in U}\Out(u)+\sum_{v\in V}\Out(v)\Bigr).
\]
Since the cost is less than $\frac{3M|U|}2$, \eqref{eq-outflow-cost} bounds the total outflow by $\frac32|U|/(R-\frac12)^2$; moreover, $4+2|U|/|V|\le6(|U|+|V|)$ and $R-\frac12\ge\frac R2=5(|U|+|V|)/\sigma$.
Hence the last term is at most $6(|U|+|V|)\cdot\frac32|U|\cdot\sigma^2/\bigl(25(|U|+|V|)^2\bigr)\le\sigma|U|$.
\end{proof}

The next lemma bounds the cost of the pair of a label $(u,a)$ in terms of the number of centers in its interval and the fraction $t=\mu^{e_u}_u(a)$ of the choice of $u$ that the pair carries; here $t\in[0,1]$ is a number, not a candidate center (Figure~\ref{fig-three}).
A solution may distribute its centers among the intervals arbitrarily, so the bound is linear in the number of centers, and after summation over all labels only the total number of centers matters.
The last term $-4t^2/\delta$ of the bound is small when $t$ is small, that is, when the choice of $u$ is spread over many labels, and Lemma~\ref{lem-concentration} bounds the sum of these terms over all labels.
The proof is based on the centroid identity for the pair: the equal fractions of the two points that are recorded by a witness can be moved to the midpoint, where they form a single point of their total weight, at a fixed cost that does not depend on the centers.

\begin{figure}[t]
\centering
\begin{tikzpicture}[fig]
\def\figu{1.3}
\def\figt{0.3}
\begin{scope}
  \draw[figaxis] (-2.05,0) -- (2.05,0);
  \node[figpt] at (-\figu,0) {}; \node[figbigpt] at (0,0) {}; \node[figpt] at (\figu,0) {};
  \node[figlab,above=3pt] at (-\figu,0) {$M(1{-}t)$};
  \node[figlab,above=3pt] at (0,0) {$2Mt$};
  \node[figlab,above=3pt] at (\figu,0) {$M(1{-}t)$};
  \node[figlab,below=3pt] at (-\figu,0) {$-\tfrac12$};
  \node[figlab,below=3pt] at (0,0) {$0$};
  \node[figlab,below=3pt] at (\figu,0) {$\tfrac12$};
  \node at (0,1.05) {the set $Q$};
\end{scope}
\begin{scope}[xshift=5.2cm]
  \begin{scope}[on background layer]
    \draw[fighull] (-\figu-0.24,-0.24) rectangle (\figu+0.24,0.24);
  \end{scope}
  \draw[figaxis] (-2.05,0) -- (2.05,0);
  \node[figpt] at (-\figu,0) {}; \node[figbigpt] at (0,0) {}; \node[figpt] at (\figu,0) {};
  \node[figctr,minimum size=12pt] at (0,0) {};
  \node at (0,1.05) {one center};
  \node[figlab] at (0,-0.8) {cost $\frac{M(1-t)}2+\frac{Mt}2=\frac M2$};
\end{scope}
\begin{scope}[xshift=10.4cm]
  \pgfmathsetmacro{\figc}{\figu*(1-\figt)/(1+\figt)}
  \begin{scope}[on background layer]
    \draw[fighull] (-\figu-0.24,-0.24) rectangle (-\figu+0.24,0.24);
    \draw[fighull] (-0.24,-0.24) rectangle (\figu+0.24,0.24);
  \end{scope}
  \draw[figaxis] (-2.05,0) -- (2.05,0);
  \node[figpt] at (-\figu,0) {}; \node[figbigpt] at (0,0) {}; \node[figpt] at (\figu,0) {};
  \node[figctr] at (-\figu,0) {};
  \node[figctr] at (\figc,0) {};
  \node at (0,1.05) {two centers};
  \node[figlab] at (0,-0.8) {cost $\frac{Mt(1-t)}{2(1+t)}+\frac{Mt}2=\frac{Mt}{1+t}$};
\end{scope}
\end{tikzpicture}
\caption{The cost of a pair (Lemma~\ref{lem-pair-cost}), drawn for $t=0.3$, with the midpoint of its interval translated to $0$.
Left: by the centroid identity, a pair that carries the fraction $t$ costs $Mt/2$ plus the cost of a fractional assignment of the three-point set $Q$, and hence at least $Mt/2+\Cost(Q,S)$; the equal fractions of the two points recorded by the witnesses are moved to the midpoint, where they form a point of weight $2Mt$, and the remaining weight $M(1-t)$ of each point stays at its position.
Middle and right: the optimal clusters (shaded) and centers (circles) for one and for two centers; with two centers, each part of an optimal partition of $Q$ consists of consecutive points.
With one center, the pair costs $M/2$ for every $t$, and with two centers it costs $Mt/(1+t)\approx Mt$ for small $t$; thus carrying a fraction $t$ costs about twice as much per unit as selecting the label.}
\label{fig-three}
\end{figure}

\begin{lemma}[Pair cost]\label{lem-pair-cost}
Let $(u,a)$ be a label of $u\in U$, let $t=\mu^{e_u}_u(a)$, and let $k_{u,a}=|S\cap I_{u,a}|$ be the number of centers in the interval of $(u,a)$.
Then the two heavy points of $(u,a)$ contribute at least
\[
M\Bigl(t+\Bigl(\frac12-\frac\delta{16}\Bigr)\bigl(2-k_{u,a}\bigr)-\frac{4t^2}\delta\Bigr)
\]
to the cost.
\end{lemma}

\begin{proof}
We first bound the cost of the pair for each value of $k_{u,a}$, and then compare these bounds with the claim.
Translate the line so that the midpoint of $I_{u,a}$ is $0$; then the pair is at $\mp\frac12$.
If $k_{u,a}=0$, every center has distance at least $R-\frac12$ from both points of the pair, which have total weight $2M$, so the pair costs at least $2M(R-\frac12)^2$.

Let $k_{u,a}\ge1$.
For every $s\in S$, we have $0\le\chi^{e_u}_{u,a}(s)\le\min\{\phi(x^-_{u,a},s),\phi(x^+_{u,a},s)\}$ by~\eqref{eq-witness-left}, and $\sum_s\chi^{e_u}_{u,a}(s)=t$ by~\eqref{eq-edge-distributions}.
Since $(s+\frac12)^2+(s-\frac12)^2=2s^2+\frac12$, the cost of the pair in the cluster of $s$ is
\begin{align*}
&M\phi(x^-_{u,a},s)\Bigl(s+\frac12\Bigr)^2+M\phi(x^+_{u,a},s)\Bigl(s-\frac12\Bigr)^2\\
&\qquad=M\bigl(\phi(x^-_{u,a},s)-\chi^{e_u}_{u,a}(s)\bigr)\Bigl(s+\frac12\Bigr)^2+M\bigl(\phi(x^+_{u,a},s)-\chi^{e_u}_{u,a}(s)\bigr)\Bigl(s-\frac12\Bigr)^2\\
&\qquad\qquad{}+2M\chi^{e_u}_{u,a}(s)\,s^2+\frac M2\chi^{e_u}_{u,a}(s).
\end{align*}
The numbers $\phi(x^\pm_{u,a},s)-\chi^{e_u}_{u,a}(s)$ and $\chi^{e_u}_{u,a}(s)$ are nonnegative, and summed over all $s$, they give $1-t$ and $t$, respectively.
Hence the cost of the pair is $\frac{Mt}2$ plus the cost of some fractional assignment to $S$ of the weighted set $Q$ that consists of the three points $-\frac12$, $0$ and $\frac12$ with weights $M(1-t)$, $2Mt$ and $M(1-t)$, and the cost of this assignment is at least $\Cost(Q,S)$, the cost of assigning every point of $Q$ to a nearest center of $S$.
Replacing each center of $S\cap I_{u,a}$ by the nearest point of $[-\frac12,\frac12]$ does not increase its distance to any point of $Q$; the moved centers have distance at most $1$ to every point of $Q$, whereas the centers outside $I_{u,a}$ have distance at least $R-\frac12>1$.
Consequently, $\Cost(Q,S)$ is at least the minimum cost of serving $Q$ with at most $k_{u,a}$ centers.
\begin{itemize}
\item With one center, the optimal center is the weighted mean $0$, with cost $2\cdot M(1-t)\cdot\frac14=\frac{M(1-t)}2$; adding $\frac{Mt}2$ shows that the pair costs at least $\frac M2$.
\item With two centers, an optimal solution partitions $Q$ into at most two parts and serves each part by its weighted mean; recall that two points with weights $w_1$ and $w_2$ that are served by one center at their weighted mean cost $\frac{w_1w_2}{w_1+w_2}$ times their squared distance.
The partition into $\{-\frac12,\frac12\}$ and $\{0\}$ costs $\frac{M(1-t)}2$, as does the partition consisting of $Q$ alone, whereas splitting $Q$ into $\{-\frac12\}$ and $\{0,\frac12\}$, or into $\{-\frac12,0\}$ and $\{\frac12\}$, costs $\frac{2Mt\cdot M(1-t)}{M(1+t)}\cdot\frac14=\frac{Mt(1-t)}{2(1+t)}$, which is not larger since $\frac t{1+t}\le1$.
Adding $\frac{Mt}2$ gives $\frac{Mt}2\cdot\frac{(1-t)+(1+t)}{1+t}=\frac{Mt}{1+t}$.
\item With three or more centers, we only use $\Cost(Q,S)\ge0$, so the bound is $\frac{Mt}2$.
\end{itemize}

It remains to show that each of these bounds is at least the claimed bound; recall that $t\in[0,1]$ and $\delta\le1$.
For $k_{u,a}=0$, the claimed bound is at most $M(1+1-\frac\delta8)<2M$, whereas $2M(R-\frac12)^2\ge2M$ since $R\ge\frac32$.
For $k_{u,a}=1$, the difference between $\frac M2$ and the claimed bound is $M(\frac\delta{16}-t+\frac{4t^2}\delta)=\frac{4M}\delta(t-\frac\delta8)^2\ge0$.
For $k_{u,a}=2$, we have $\frac{Mt}{1+t}\ge M(t-t^2)\ge M(t-\frac{4t^2}\delta)$, because $(t-t^2)(1+t)=t-t^3\le t$ and $\delta\le4$.
For $k_{u,a}\ge3$, the claimed bound is at most $M(t-(\frac12-\frac\delta{16})-\frac{4t^2}\delta)$, and it is at most $\frac{Mt}2$ because $\frac t2-\frac{4t^2}\delta$ is maximized at $t=\frac\delta{16}$, where it equals $\frac\delta{64}\le\frac12-\frac\delta{16}$.
\end{proof}

\begin{proposition}[Soundness]\label{prop-small-value}
If the value of the label cover instance is at most $\sigma$, then every feasible solution of the constrained instance has cost at least $(3-\frac\delta4)\frac{M|U|}2$.
\end{proposition}

\begin{proof}
Let $(S,\phi)$ be a feasible solution.
We may assume that its cost is less than $\frac{3M|U|}2$; then the assumptions of this subsection hold.
For a label $(v,b)$ of a vertex $v\in V$, let $k_{v,b}=|S\cap I_{v,b}|$.
The cost of $x_{v,b}$ is at least $M(\frac12-\frac\delta{16})(1-k_{v,b})$: if $k_{v,b}=0$, every center has distance more than $R$ from $x_{v,b}$, so this cost is at least $MR^2\ge\frac M2$.
No center lies in two intervals, so the sum of all numbers $k_{u,a}$ and $k_{v,b}$ is at most $|S|\le k$.
We ignore the nonnegative cost of the light points and apply Lemma~\ref{lem-pair-cost} to every pair:
\begin{align*}
\frac{\Cost(X,S,\phi)}M
&\ge\sum_{(u,a)}\Bigl(\mu^{e_u}_u(a)+\Bigl(\frac12-\frac\delta{16}\Bigr)\bigl(2-k_{u,a}\bigr)-\frac{4\mu^{e_u}_u(a)^2}\delta\Bigr)+\sum_{(v,b)}\Bigl(\frac12-\frac\delta{16}\Bigr)\bigl(1-k_{v,b}\bigr)\\
&=|U|+\Bigl(\frac12-\frac\delta{16}\Bigr)\Bigl(2|\mathcal A||U|+|\mathcal B||V|-\sum_{(u,a)}k_{u,a}-\sum_{(v,b)}k_{v,b}\Bigr)-\frac4\delta\sum_{(u,a)}\mu^{e_u}_u(a)^2\\
&\ge|U|+\Bigl(\frac12-\frac\delta{16}\Bigr)|U|-\frac4\delta\cdot2\sigma|U|
=\Bigl(\frac32-\frac\delta8\Bigr)|U| .
\end{align*}
Here the equality holds because each $\mu^{e_u}_u$ has total mass $1$.
The inequality uses $\frac12-\frac\delta{16}\ge0$, the bound on the number of centers, the identity $2|\mathcal A||U|+|\mathcal B||V|-k=|U|$ implied by~\eqref{eq-budget}, and Lemma~\ref{lem-concentration}, and the last step uses $\sigma=\delta^2/128$.
Therefore $\Cost(X,S,\phi)\ge(\frac32-\frac\delta8)M|U|=(3-\frac\delta4)\frac{M|U|}2$.
\end{proof}

In this proof, the coefficient $\frac12-\frac\delta{16}$ of the number of centers loses $\frac{\delta|U|}{16}$ compared with the coefficient $\frac12$, and the last term of Lemma~\ref{lem-pair-cost} loses at most $\frac4\delta\cdot2\sigma|U|$; the choice $\sigma=\delta^2/128$ makes the two losses equal.

\subsection{Encoding by Groups}\label{sec-groups}

It remains to express the constraints of the constrained instance by groups of points of unit weight.
Let $K$ be the constant of Lemma~\ref{lem-consistency-linear}, and let $n=2K\sum_{x\in X}w(x)$, where $w(x)\in\{1,M\}$ is the weight of $x\in X$.
We replace every $x\in X$ by $2Kw(x)$ \keyterm{copies} of unit weight located at $x$, which gives $n$ points; thus $n$ is the number of points of the final instance, as in Section~\ref{sec-problems}.
An assignment matrix of the copies to a center set $S$ is \keyterm{synchronized} if, for each $x\in X$ and each $s\in S$, all copies of $x$ have the same fraction in the cluster of $s$; we denote this fraction by $\phi(x,s)$ and call $\phi$ the assignment matrix of $X$ \keyterm{induced} by the copies.
The copies of $x$ lie at $x$ and have total weight $2Kw(x)$, so a synchronized assignment costs $2K\Cost(X,S,\phi)$.

\begin{lemma}[Encoding by groups]\label{lem-encoding}
One can compute in polynomial time at most $1+n^2/4+K|E|$ groups of copies, each with $\beta_i=\frac12$ and $\alpha_i\in\{\frac12,1\}$. An assignment of the copies to $S$ satisfies all these fairness constraints if and only if it is synchronized and every cluster is $e$-consistent for every edge $e$.
\end{lemma}

Appendix~\ref{app-groups} gives the construction: groups with both bounds equal to $\frac12$ force all copies of a point to have the same fractions, and every inequality of Lemma~\ref{lem-consistency-linear} becomes one group with lower bound $\frac12$ and upper bound $1$.

\subsection{Proof of Theorem~\ref{thm-hardness}}\label{sec-proof-hardness}

\begin{proof}[Proof of Theorem~\ref{thm-hardness}]
\emph{The instance.}
Given a 3CNF formula, we compute the label cover instance of Lemma~\ref{lem-label-cover} for $\sigma=\delta^2/128$ and from it the \keyterm{final instance}: the $n$ copies of unit weight, the groups of Lemma~\ref{lem-encoding}, and the budget $k$ of~\eqref{eq-budget}.
Since $|\mathcal A|$, $|\mathcal B|$ and $K$ are constants, the numbers $R$, $\diam(X)$, $M$ and $n$ are polynomial in $|U|+|V|+|E|$, and hence in the length of the formula; therefore the final instance can be computed in polynomial time.
Its points lie on the line and have unit weight, and its groups have $\beta_i=\frac12$ and $\alpha_i\in\{\frac12,1\}$.
Assigning all copies to a single center gives a synchronized assignment whose induced matrix has $\phi(x,s)=1$ for all $x\in X$; by Lemma~\ref{lem-encoding} and the last part of Lemma~\ref{lem-consistency-linear}, it is fair, so fair solutions exist.

\emph{The two cases.}
If the formula is satisfiable, the label cover instance has value $1$, and Lemma~\ref{lem-satisfiable} gives an integral feasible solution of the constrained instance with cost at most $(1+\frac\delta{32})\frac{M|U|}2$.
Assigning all copies of each point $x$ to the center to which this solution assigns $x$ gives an integral synchronized assignment whose induced matrix is that solution; by Lemma~\ref{lem-encoding} it is fair, and its cost is at most $2K(1+\frac\delta{32})\frac{M|U|}2=(1+\frac\delta{32})KM|U|$, so $\OPT\le(1+\frac\delta{32})KM|U|$.
If the formula is unsatisfiable, the value of the label cover instance is at most $\sigma$.
Consider a fair assignment of the copies to a center set $S$ with $|S|\le k$, integral or fractional.
By Lemma~\ref{lem-encoding}, it is synchronized, and its induced matrix $\phi$ gives a feasible solution $(S,\phi)$ of the constrained instance; hence, by Proposition~\ref{prop-small-value}, its cost $2K\Cost(X,S,\phi)$ is at least $2K(3-\frac\delta4)\frac{M|U|}2=(3-\frac\delta4)KM|U|$.

\emph{The reduction.}
Assume that some deterministic algorithm with polynomial running time outputs, for each instance, at most $k$ centers and an $(\alpha,\beta)$-fair fractional assignment to them with cost at most $(3-\delta)\OPT$.
We apply the algorithm to the final instance, accept if the cost of its output is less than the threshold $(3-\frac\delta2)KM|U|$, and reject otherwise.
For a satisfiable formula, the cost is at most $(3-\delta)(1+\frac\delta{32})KM|U|=(3-\frac{29\delta}{32}-\frac{\delta^2}{32})KM|U|$, which is below the threshold, and for an unsatisfiable formula, it is at least $(3-\frac\delta4)KM|U|$, which is above the threshold.
This decides satisfiability of 3CNF formulas in polynomial time, which implies $\mathrm P=\mathrm{NP}$.
\end{proof}

\begin{remark}[Fixed dimension]\label{rem-hardness-scope}
By embedding the line into $\R^d$, we obtain the same statement in every fixed dimension $d\ge1$: centers outside the line can be projected onto it without increasing the cost, and centers that coincide after the projection can be merged, which preserves fairness because the fairness constraints are homogeneous.
Hence, unless $\mathrm P=\mathrm{NP}$, fair $k$-means admits no PTAS in any fixed dimension, in contrast to standard $k$-means, which admits a PTAS in every fixed dimension by the local search algorithms of~\cite{cohen2019local,friggstad2019local}; Lemma~\ref{lem-ptas} in Section~\ref{sec-prelim} states the weighted version that our algorithm uses.
\end{remark}

\section{\texorpdfstring{$k$}{k}-Sparse Wasserstein Barycenters}\label{sec-wb}

We now prove Theorem~\ref{thm-wb} by making precise, for weighted points, the correspondence between $k$-sparse barycenters and fair fractional assignments described in Section~\ref{sec-problems}.

\paragraph{Notation.}
The input consists of discrete probability distributions $P^{(1)},\dots,P^{(m)}$ on $\R^d$; an atom $q$ of $P^{(i)}$ has mass $w(q)>0$, and the masses of each distribution sum to $1$.
The disjoint union $P=P^{(1)}\uplus\dots\uplus P^{(m)}$ is a weighted indexed set of $n$ points.
For a probability distribution $B$ with finite support, let $b_s>0$ denote its mass at $s\in\supp(B)$.
In this notation, a transport plan from $P^{(i)}$ to $B$, as defined in Section~\ref{sec-problems}, is a nonnegative function $F_i$ on $P^{(i)}\times\supp(B)$ with marginals $w$ and $b$,
\begin{equation}\label{eq-transport-plan}
\sum_{s\in\supp(B)}F_i(q,s)=w(q)\quad(q\in P^{(i)}),
\qquad
\sum_{q\in P^{(i)}}F_i(q,s)=b_s\quad(s\in\supp(B)),
\end{equation}
and $\mathcal W^2(P^{(i)},B)$, $\Cost_{\mathrm{WB}}(B)$ and $\OPT_{\mathrm{WB}}$ are defined as in Section~\ref{sec-problems}.

\paragraph{The associated fair instance.}
The associated fair instance is the instance of fair $k$-means whose points are the points of $P$ with weights $w(q)$, whose groups are the pairwise disjoint sets $P^{(1)},\dots,P^{(m)}$, whose bounds are $\alpha_i=\beta_i=\frac1m$ for all $i$, and whose number of centers is the same $k$.
Every group has total weight $1$ and $P$ has total weight $m$, so by~\eqref{eq-fairness} a cluster is fair if and only if each distribution contributes a $1/m$ fraction of the weight of the cluster; in particular, the cluster that contains all of $P$ is fair.

\begin{lemma}[Transport plans and fair assignments]\label{lem-wb-correspondence}
For the associated fair instance, the following statements hold.
\begin{enumerate}
\item Let $S$ be a center set and $\phi_S$ an $(\alpha,\beta)$-fair fractional assignment of $P$ to $S$.
The masses $w(s)/m$ at the centers $s\in S$ define a probability distribution $B$ with $\supp(B)\subseteq S$ and $m\Cost_{\mathrm{WB}}(B)\le\Cost(P,S,\phi_S)$.
\item Let $B$ be a probability distribution with finite support $S$, and let $F_i$ be an optimal transport plan from $P^{(i)}$ to $B$ for each $i$.
Then $\phi_S(q,s)=F_i(q,s)/w(q)$, for $q\in P^{(i)}$ and $s\in S$, is an $(\alpha,\beta)$-fair fractional assignment of $P$ to $S$ with $\Cost(P,S,\phi_S)=m\Cost_{\mathrm{WB}}(B)$.
\end{enumerate}
Consequently, $\OPTfrac=m\,\OPT_{\mathrm{WB}}$ for the associated instance.
\end{lemma}

\begin{proof}
For part~1, the masses $w(s)/m$ sum to $w(P)/m=1$.
For $q\in P^{(i)}$ and $s\in\supp(B)$, let $F_i(q,s)=w(q)\phi_S(q,s)$.
The rows of $\phi_S$ sum to one, so $\sum_sF_i(q,s)=w(q)$, and fairness gives $\sum_{q\in P^{(i)}}F_i(q,s)=w^{(i)}(s)=w(s)/m$.
Hence $F_i$ is a transport plan from $P^{(i)}$ to $B$, and
\[
\mathcal W^2\bigl(P^{(i)},B\bigr)\le\sum_{q\in P^{(i)}}\sum_{s\in S}w(q)\,\phi_S(q,s)\lVert q-s\rVert^2 .
\]
Summing these bounds over $i$ yields $m\Cost_{\mathrm{WB}}(B)\le\Cost(P,S,\phi_S)$.

For part~2, the rows of $\phi_S$ sum to one, since the first marginal of $F_i$ is $w$.
The weight of the group $P^{(i)}$ at $s$ is $w^{(i)}(s)=\sum_{q\in P^{(i)}}F_i(q,s)=b_s$, which does not depend on $i$; hence $w(s)=mb_s$, and $w^{(i)}(s)=w(s)/m$ shows that $\phi_S$ is fair.
Its cost is
\[
\Cost(P,S,\phi_S)=\sum_{i=1}^m\sum_{q\in P^{(i)}}\sum_{s\in S}F_i(q,s)\lVert q-s\rVert^2=\sum_{i=1}^m\mathcal W^2\bigl(P^{(i)},B\bigr)=m\Cost_{\mathrm{WB}}(B).
\]

Part~1 shows that every fair fractional solution with at most $k$ centers has cost at least $m\OPT_{\mathrm{WB}}$, and part~2 shows that every distribution $B$ with at most $k$ atoms yields a fair fractional solution of cost $m\Cost_{\mathrm{WB}}(B)$; hence the two infima are equal up to the factor $m$.
\end{proof}

\begin{remark}[Weighted points]\label{rem-weighted}
The associated fair instance has weighted points, since the atoms carry their masses as weights. Sections~\ref{sec-algorithm} and~\ref{sec-analysis} extend to positive weights $w(p)$ by multiplying every summand over input points by $w(p)$: this applies to $w(t)$, $\pi(t)$, $C_T$, $C_\nu$, $C_e$, $\cL$, the fairness constraints~\eqref{eq-fairness}, and the candidate-set guarantee~\eqref{eq-weighted-candidates}. The definition of $\nu_p$ does not change, since every row of an assignment matrix still sums to one. All proofs remain valid because they use only the centroid identity~\eqref{eq-centroid-identity}, the Cauchy--Schwarz inequality, and the row-sum condition, each of which holds for weighted sets. We refer to the resulting algorithm and theorem as the weighted versions of Algorithm~\ref{alg-main} and Theorem~\ref{thm-main}.
\end{remark}

\begin{proof}[Proof of Theorem~\ref{thm-wb}]
We run the weighted version of Algorithm~\ref{alg-main} on the associated fair instance, which has $n$ points and $m$ groups, and obtain a set $S$ of at most $k$ centers together with the fair assignment $\phi^*_S$.
By the weighted version of Theorem~\ref{thm-main} and Lemma~\ref{lem-wb-correspondence},
\[
\Cost(P,S,\phi^*_S)\le(2\rho+1+\epsilon)\,\OPTfrac=(2\rho+1+\epsilon)\,m\,\OPT_{\mathrm{WB}} .
\]
By part~1 of Lemma~\ref{lem-wb-correspondence}, we can compute from $\phi^*_S$, in $O(n|S|)$ time, a distribution $B$ with $\supp(B)\subseteq S$, hence $|\supp(B)|\le k$, and
\[
\Cost_{\mathrm{WB}}(B)\le\frac1m\Cost(P,S,\phi^*_S)\le(2\rho+1+\epsilon)\,\OPT_{\mathrm{WB}} .
\]
For fixed $d$, Lemma~\ref{lem-ptas} with $\epsilon'=\epsilon$ provides the subroutine with $\rho=1+\epsilon$, which gives the ratio $2(1+\epsilon)+1+\epsilon=3+O(\epsilon)$.
\end{proof}

\section{Conclusion}\label{sec-conclusion}

We showed that ``$3$'' is the approximation threshold of fractional fair $k$-means in fixed dimension.
Given a $\rho$-approximation algorithm for weighted $k$-means, our algorithm achieves the ratio $2\rho+1+\epsilon$ by solving the drift-regularized relaxation and calling this algorithm once, so every improvement of $\rho$ carries over directly.
The same ratio holds for $k$-sparse Wasserstein barycenters, a special case of fractional fair $k$-means, and in fixed dimension it improves the best known ratio for this problem from $5+O(\epsilon)$ to $3+O(\epsilon)$.
Theorem~\ref{thm-homogeneous} and Corollary~\ref{cor-homogeneous-threshold} extend both bounds to $k$-means with homogeneous cluster constraints.

Several questions remain open.
First, capacities and lower bounds on the sizes of the clusters, which the constrained $k$-means frameworks of~\cite{ding2020unified,bhattacharya2018faster} handle in time exponential in~$k$, are linear but not homogeneous: scaling a column in the $\theta$-mixture can violate a lower bound, and merging fragment clusters at a common center can violate a capacity.
It is open whether the drift-regularized relaxation can be adapted to such constraints.
Second, the analysis relies on the centroid identity of squared Euclidean distances.
Fair $k$-median has no such identity; Bregman divergences satisfy it~\cite{banerjee2005clustering}, but the comparison in Section~\ref{sec-comparison} also uses the inner product of $\R^d$.
It would be interesting to extend the drift-regularized relaxation to these objectives and to determine whether the construction of Section~\ref{sec-hardness} gives tight lower bounds for them.
Third, the lower bound $3$ holds in every dimension, but since $k$-means is APX-hard in high dimension~\cite{awasthi2015hardness,lee2017improved}, the ratio $2\rho+1+\epsilon$ exceeds $3$ by a constant in arbitrary dimension unless $\mathrm P=\mathrm{NP}$.
It is open whether this ratio can be lowered in arbitrary dimension.

\Needspace*{6\baselineskip}
\section*{AI Methodology}
We used ChatGPT (OpenAI) to assist with proving parts of the results of this paper.
The tool materially affected Sections~\ref{sec-stationarity}, \ref{sec-computing} and~\ref{sec-hardness}.
The authors verified the correctness and originality of all content, including references, and take full responsibility for all results.

\appendix
\section{Proof of Lemma~\ref{lem-candidates}}\label{app-candidate}

Scaling all weights $\mu(q)$ by a common positive factor scales both sides of~\eqref{eq-weighted-candidates} by the same factor, so we may assume that the weights sum to $1$.
Let
\[
c=\sum_{q\in Q}\mu(q)\,q,
\qquad
\sigma^2=\sum_{q\in Q}\mu(q)\lVert q-c\rVert^2
\]
be the weighted centroid and the weighted variance.
By the centroid identity, $\sigma^2$ is the minimum on the right-hand side of~\eqref{eq-weighted-candidates}, and every $t\in\R^d$ satisfies $\sum_q\mu(q)\lVert q-t\rVert^2=\sigma^2+\lVert t-c\rVert^2$.
Hence it suffices to show that each of the two sets contains a point $t$ with $\lVert t-c\rVert^2\le\delta\sigma^2$.

Both parts use an averaging argument of Inaba, Katoh, and Imai~\cite{inaba1994applications}, who state it for samples drawn uniformly from a point set; the calculation below repeats it for samples drawn with probabilities $\mu(q)$.
A random sample has a small expected error, so some sample has at most this error, and since the candidate set contains a point for every possible sample, the construction uses no randomness.
We draw $r$ points $y_1,\dots,y_r$ of $Q$ independently, where each draw equals $q$ with probability $\mu(q)$, and we let $Y$ be the indexed set of the samples, all of unit weight; a point that is drawn several times appears in $Y$ with the corresponding multiplicity.
Each sample has mean $c$ and satisfies $\mathbb E\lVert y_\ell-c\rVert^2=\sigma^2$, and by independence the cross terms vanish when $\lVert\Cen(Y)-c\rVert^2=\lVert\frac1r\sum_\ell(y_\ell-c)\rVert^2$ is expanded; hence
\begin{equation}\label{eq-sample-mean}
\mathbb E\lVert\Cen(Y)-c\rVert^2=\frac{\sigma^2}r .
\end{equation}
The centroid identity for $Y$ with $z=c$ also gives the expected average squared distance of the samples to their mean:
\begin{equation}\label{eq-sample-spread}
\mathbb E\Bigl[\frac1r\sum_{\ell=1}^r\lVert y_\ell-\Cen(Y)\rVert^2\Bigr]
=\frac1r\sum_{\ell=1}^r\mathbb E\lVert y_\ell-c\rVert^2-\mathbb E\lVert\Cen(Y)-c\rVert^2
=\Bigl(1-\frac1r\Bigr)\sigma^2 .
\end{equation}

\paragraph{Part~2.}
Let $r=\lceil1/\delta\rceil$.
For every outcome, the samples form a multiset of $r$ points of $Q$, so their mean $\Cen(Y)$ belongs to the set of part~2.
By~\eqref{eq-sample-mean}, the expected value of $\lVert\Cen(Y)-c\rVert^2$ is $\sigma^2/r\le\delta\sigma^2$, so some outcome satisfies $\lVert\Cen(Y)-c\rVert^2\le\delta\sigma^2$.
Since the set of part~2 contains the mean of every multiset of $r$ points of $Q$, it includes this outcome without knowing $\mu$; no randomness is used by the algorithm.
A multiset of $r$ points of $Q$ corresponds to a vector of nonnegative integer multiplicities, one per point, that sum to $r$; hence the number of such multisets is
\[
\binom{n+r-1}{r}=\prod_{j=1}^{r}\frac{n+j-1}{j}\le n^r ,
\]
where every factor is at most $n$ because $n+j-1\le nj$ is equivalent to $(n-1)(j-1)\ge0$.
Computing each mean takes $O(rd)$ time.

\paragraph{Part~1.}
Now let $r=\lceil4/\delta\rceil$, and let $\CS_{\sqrt\delta}(Q^{[r]})$ be the set that the algorithm of Matou\v{s}ek~\cite{matouvsek2000approximate} computes for $Q^{[r]}$ with accuracy $\sqrt\delta$.
In every outcome, each point of $Q$ appears at most $r$ times in $Y$, so $Y$ can be represented by $r$ distinct elements of $Q^{[r]}$, using different copies of a point that is drawn more than once.
Applying~\eqref{eq-approximate-centroid-set}, with accuracy $\sqrt\delta$, to this subset of $Q^{[r]}$, whose centroid is $\Cen(Y)$, gives a point $t_Y\in\CS_{\sqrt\delta}(Q^{[r]})$ with
\[
\lVert t_Y-\Cen(Y)\rVert^2\le\frac\delta9\cdot\frac1r\sum_{\ell=1}^r\lVert y_\ell-\Cen(Y)\rVert^2 .
\]
Using $\lVert x+x'\rVert^2\le2\lVert x\rVert^2+2\lVert x'\rVert^2$, then~\eqref{eq-sample-spread} and~\eqref{eq-sample-mean}, and finally $r\ge4/\delta$, we obtain
\[
\mathbb E\lVert t_Y-c\rVert^2
\le2\,\mathbb E\lVert t_Y-\Cen(Y)\rVert^2+2\,\mathbb E\lVert\Cen(Y)-c\rVert^2
\le\frac{2\delta}9\sigma^2+\frac{2\sigma^2}r
\le\Bigl(\frac{2\delta}9+\frac\delta2\Bigr)\sigma^2\le\delta\sigma^2 .
\]
Hence some outcome yields a point $t_Y\in\CS_{\sqrt\delta}(Q^{[r]})$ with $\lVert t_Y-c\rVert^2\le\delta\sigma^2$.
The set $Q^{[r]}$ has $nr=O(n/\delta)$ points, and $\log(2/\sqrt\delta)\le\log(2/\delta)$.
The construction of Matou\v{s}ek~\cite{matouvsek2000approximate} described in Section~\ref{sec-prelim}, applied with accuracy $\sqrt\delta$, therefore produces $O(nr\,\delta^{-d/2}\log(2/\delta))=O(n\delta^{-1-d/2}\log(2/\delta))$ points in time
\[
O\bigl(nr\log(nr)+nr\,\delta^{-d/2}\log(2/\delta)\bigr)
=O\bigl(n\delta^{-1}\log(n/\delta)+n\delta^{-1-d/2}\log(2/\delta)\bigr).
\]

\section{Proof of Lemma~\ref{lem-encoding}}\label{app-groups}

We use the notation of Section~\ref{sec-groups}.
A group with both bounds equal to $\frac12$ forces every cluster to contain the same weight from the group and from its complement.
We use such groups to synchronize the copies, and groups with lower bound $\frac12$ and upper bound $1$ to express the inequalities of Lemma~\ref{lem-consistency-linear}.

\begin{lemma}[Synchronization]\label{lem-synchronization}
There is a family of at most $1+n^2/4$ groups of copies, each with the bounds $\beta_i=\alpha_i=\frac12$, such that an assignment matrix of the copies satisfies the fairness constraints of all these groups if and only if it is synchronized.
\end{lemma}

\begin{proof}
Choose a set $G$ of copies that contains exactly half of the copies of every $x\in X$, that is, $Kw(x)$ of them.
The family consists of $G$ and, for every $x\in X$ and all copies $y\in G$ and $y'\notin G$ of $x$, the group $(G\setminus\{y\})\cup\{y'\}$; there are $\sum_x(Kw(x))^2\le n^2/4$ groups of the second kind.
Every group of the family contains exactly half of the copies of each $x$, so under a synchronized assignment every group contributes exactly half of the weight of every cluster, and all constraints hold.
Conversely, if all constraints hold, then $G$ and $(G\setminus\{y\})\cup\{y'\}$ both contribute exactly half of the weight of every cluster, so $y$ and $y'$ have the same fraction in every cluster.
Since every copy of $x$ in $G$ is compared in this way with every copy of $x$ outside $G$, and both kinds of copies exist, all copies of $x$ have the same fractions.
\end{proof}

\begin{lemma}[Realizing an inequality]\label{lem-realize}
Let $c_x$, $x\in X$, be integers with $|c_x|\le K$, and consider the group that consists of $Kw(x)+c_x$ copies of each $x\in X$, with the bounds $\beta_i=\frac12$ and $\alpha_i=1$.
Under a synchronized assignment with induced matrix $\phi$, this group satisfies its fairness constraint in the cluster of a center $s$ if and only if $\sum_xc_x\phi(x,s)\ge0$.
\end{lemma}

\begin{proof}
Since $w(x)\ge1$, the number $Kw(x)+c_x$ lies between $0$ and $2Kw(x)$, so the group is well defined.
Under a synchronized assignment, the weight of the group in the cluster of $s$ is $\sum_x(Kw(x)+c_x)\phi(x,s)$, and the weight of the cluster is $\sum_x2Kw(x)\phi(x,s)$.
Hence the lower bound $\frac12$ is equivalent to $\sum_xc_x\phi(x,s)\ge0$, and the upper bound $1$ always holds.
\end{proof}

\begin{proof}[Proof of Lemma~\ref{lem-encoding}]
We take the groups of Lemma~\ref{lem-synchronization} and, for every edge $e$, the at most $K$ groups that Lemma~\ref{lem-realize} gives for the inequalities of Lemma~\ref{lem-consistency-linear}; this gives at most $1+n^2/4+K|E|$ groups.
They can be listed in polynomial time, because each group contains at most $n$ copies and the inequalities of an edge are computed in constant time.
If an assignment matrix of the copies satisfies the fairness constraints of all these groups, then it is synchronized by Lemma~\ref{lem-synchronization}, and by Lemma~\ref{lem-realize}, every cluster of the induced matrix satisfies all inequalities of Lemma~\ref{lem-consistency-linear}, that is, it is $e$-consistent for every edge $e$.
Conversely, a synchronized assignment whose induced matrix has only $e$-consistent clusters satisfies the constraints of the groups of Lemma~\ref{lem-synchronization}, and, by Lemmas~\ref{lem-consistency-linear} and~\ref{lem-realize}, those of the other groups.
\end{proof}

\section{Proofs of Corollaries~\ref{cor-overlapping} and~\ref{cor-partition}}\label{app-rounding}

Let $(S,\phi_S)$ be the solution returned by Algorithm~\ref{alg-main}, so that $S=S_\pi$ and $\phi_S=\phi^*_{S_\pi}$ is an optimal solution of the fair assignment LP of $S$ (Definition~\ref{def-fair-lp}).
By Theorem~\ref{thm-main}, $\Cost(P,S,\phi_S)\le(2\rho+1+\epsilon)\,\OPTfrac$.
Both corollaries therefore follow once we compute an integral assignment matrix $\phi'_S$ of $P$ to the same centers that has the stated additive violation and satisfies
\begin{equation}\label{eq-rounding-goal}
\Cost(P,S,\phi'_S)\le\Cost(P,S,\phi_S).
\end{equation}
We write $w(s)$ and $w^{(i)}(s)$ for the weights~\eqref{eq-cluster-weight} of the clusters of $\phi_S$, and $w'(s)$ and $w'^{(i)}(s)$ for those of $\phi'_S$.
All input points have weight $1$, so the fair assignment LP of $S$ reads
\begin{align*}
\text{minimize}\quad&\sum_{p\in P}\sum_{s\in S}\phi(p,s)\,\lVert p-s\rVert^2\\
\text{subject to}\quad&\sum_{s\in S}\phi(p,s)=1&&\text{for all }p\in P,\\
&\beta_i\sum_{p\in P}\phi(p,s)\le\sum_{p\in P^{(i)}}\phi(p,s)\le\alpha_i\sum_{p\in P}\phi(p,s)&&\text{for all }s\in S\text{ and }i\in[m],\\
&\phi(p,s)\ge0&&\text{for all }p\in P\text{ and }s\in S.
\end{align*}

\paragraph{Overlapping groups.}
This linear program is the linear programming relaxation of the fair $p$-assignment problem of Bera, Chakrabarty, Flores, and Negahbani~\cite[Section~3.2]{bera2019fair} with $p=2$, in which the points of $P$ are the clients, the centers of $S$ are the facilities, and the distance is Euclidean.
Their Theorem~6 gives an algorithm that runs in time polynomial in $n$, $k$ and $m$ and returns an integral assignment of $P$ to $S$ whose cost is at most the optimal value of this linear program and whose additive violation is at most $4\Delta+3$, where $\Delta$ is the maximum number of groups that contain a common point.
The algorithm keeps the pairs $(p,s)$ to which an optimal solution of the linear program assigns a positive fraction and selects one pair for every point, that is, a basis of a partition matroid on these pairs.
For every center $s$ and every group $P^{(i)}$, it requires the number of selected pairs at $s$ and the number of selected pairs at $s$ whose point lies in $P^{(i)}$ to lie between the floor and the ceiling of the corresponding fractional values.
The optimal solution of the linear program is feasible for the natural relaxation of this degree-bounded matroid basis problem, so iterative rounding finds a basis of no larger cost that violates each degree bound by an additive amount that depends only on $\Delta$, and this translates into the additive violation $4\Delta+3$ of the fairness bounds.
Since $\phi_S$ is an optimal solution of the linear program, the returned assignment $\phi'_S$ satisfies~\eqref{eq-rounding-goal}, which proves Corollary~\ref{cor-overlapping}.

\paragraph{Non-overlapping groups.}
Now assume that the groups do not overlap.
If some points belong to no group, we add them as a group $P^{(m+1)}$ with $\beta_{m+1}=0$ and $\alpha_{m+1}=1$, whose bounds hold for every assignment, and increase $m$ by one.
Hence every point belongs to exactly one group, and we let $i(p)$ be the index of this group.
Following Song, Mo, and Ding~\cite[Lemma~4]{song2025relax}, we consider the network that consists of a source, a sink, a node for every point $p\in P$, a node for every pair $(i,s)\in[m]\times S$ and a node for every center $s\in S$, with the following arcs:
\begin{itemize}
\item an arc from the source to every point $p$, with lower and upper capacity $1$;
\item for every point $p$ and every center $s$, an arc from $p$ to $(i(p),s)$ with lower capacity $0$, upper capacity $1$ and cost $\lVert p-s\rVert^2$ per unit of flow;
\item an arc from every node $(i,s)$ to $s$, with lower capacity $\lfloor w^{(i)}(s)\rfloor$ and upper capacity $\lceil w^{(i)}(s)\rceil$;
\item an arc from every center $s$ to the sink, with lower capacity $\lfloor w(s)\rfloor$ and upper capacity $\lceil w(s)\rceil$.
\end{itemize}
All arcs except those from points to pairs have cost $0$, and we seek a minimum-cost flow of value $n$ from the source to the sink.

The solution $\phi_S$ defines a feasible flow: it sends one unit to every point $p$, the amount $\phi_S(p,s)$ from $p$ to $(i(p),s)$, the amount $w^{(i)}(s)$ from $(i,s)$ to $s$, and the amount $w(s)$ from $s$ to the sink.
Flow is conserved at every point $p$ because $\sum_{s\in S}\phi_S(p,s)=1$, at every node $(i,s)$ because its incoming flow is $\sum_{p\in P^{(i)}}\phi_S(p,s)=w^{(i)}(s)$, and at every center $s$ because every point belongs to exactly one group, so that $\sum_{i\in[m]}w^{(i)}(s)=w(s)$.
The cost of this flow is $\Cost(P,S,\phi_S)$.

All capacities are integers, so the network has an integral minimum-cost flow $f$; Song, Mo, and Ding compute it in $O(n^3k^2)$ time.
The cost of $f$ is at most that of the flow defined by $\phi_S$.
Let $\phi'_S(p,s)$ be the flow of $f$ on the arc from $p$ to $(i(p),s)$.
Every point receives one unit from the source and, since $f$ is integral, sends it along exactly one of its arcs, so $\phi'_S$ is an integral assignment matrix.
Its cost equals the cost of $f$, which proves~\eqref{eq-rounding-goal}.

It remains to bound the violation.
By flow conservation, $w'^{(i)}(s)$ is the flow of $f$ on the arc from $(i,s)$ to $s$ and $w'(s)$ is its flow on the arc from $s$ to the sink, so $\lvert w'^{(i)}(s)-w^{(i)}(s)\rvert<1$ and $\lvert w'(s)-w(s)\rvert<1$.
Since $\phi_S$ is fair and $\beta_i\le1$,
\[
w'^{(i)}(s)>w^{(i)}(s)-1\ge\beta_i\,w(s)-1>\beta_i\bigl(w'(s)-1\bigr)-1\ge\beta_i\,w'(s)-2 ,
\]
and in the same way $w'^{(i)}(s)<\alpha_i\,w'(s)+2$.
Hence $\phi'_S$ has additive violation at most $2$, which proves Corollary~\ref{cor-partition}.

\clearpage
\phantomsection
\addcontentsline{toc}{section}{References}
\bibliographystyle{alpha}
\bibliography{refs}

\newcommand{\etalchar}[1]{$^{#1}$}
\begin{thebibliography}{CAEMN22}

\bibitem[ABA21]{altschuler2021wasserstein}
Jason~M. Altschuler and Enric Boix-Adser{\`a}.
\newblock {W}asserstein barycenters can be computed in polynomial time in fixed
  dimension.
\newblock {\em Journal of Machine Learning Research}, 22(44):1--19, 2021.

\bibitem[ABA22]{altschuler2022wasserstein}
Jason~M. Altschuler and Enric Boix-Adser{\`a}.
\newblock {W}asserstein barycenters are {NP}-hard to compute.
\newblock {\em SIAM Journal on Mathematics of Data Science}, 4(1):179--203,
  2022.

\bibitem[ABM16]{anderes2016discrete}
Ethan Anderes, Steffen Borgwardt, and Jacob Miller.
\newblock Discrete {W}asserstein barycenters: Optimal transport for discrete
  data.
\newblock {\em Mathematical Methods of Operations Research}, 84:389--409, 2016.

\bibitem[AC11]{agueh2011barycenters}
Martial Agueh and Guillaume Carlier.
\newblock Barycenters in the {W}asserstein space.
\newblock {\em SIAM Journal on Mathematical Analysis}, 43(2):904--924, 2011.

\bibitem[ACCA{\etalchar{+}}27]{anand2026spectral}
Aditya Anand, Moses Charikar, Vincent Cohen-Addad, Ruiquan Gao, Fabrizio
  Grandoni, Euiwoong Lee, Amatya Sharma, and Ernest van Wijland.
\newblock Spectral dual fitting for {$k$}-means.
\newblock In {\em Proceedings of the 2027 Annual ACM-SIAM Symposium on Discrete
  Algorithms (SODA 2027)}, 2027.
\newblock To appear.

\bibitem[ACKS15]{awasthi2015hardness}
Pranjal Awasthi, Moses Charikar, Ravishankar Krishnaswamy, and Ali~Kemal Sinop.
\newblock The hardness of approximation of {E}uclidean {$k$}-means.
\newblock In {\em 31st International Symposium on Computational Geometry (SoCG
  2015)}, volume~34 of {\em Leibniz International Proceedings in Informatics},
  pages 754--767, 2015.

\bibitem[ALM{\etalchar{+}}98]{arora1998proof}
Sanjeev Arora, Carsten Lund, Rajeev Motwani, Madhu Sudan, and Mario Szegedy.
\newblock Proof verification and the hardness of approximation problems.
\newblock {\em Journal of the ACM}, 45(3):501--555, 1998.

\bibitem[ANFSW20]{ahmadian2020better}
Sara Ahmadian, Ashkan Norouzi-Fard, Ola Svensson, and Justin Ward.
\newblock Better guarantees for {$k$}-means and {E}uclidean {$k$}-median by
  primal-dual algorithms.
\newblock {\em SIAM Journal on Computing}, 49(4):FOCS17--97--FOCS17--156, 2020.

\bibitem[ATL18]{alelyani2018feature}
Salem Alelyani, Jiliang Tang, and Huan Liu.
\newblock Feature selection for clustering: A review.
\newblock In {\em Data Clustering}, pages 29--60. Chapman and Hall/CRC, 2018.

\bibitem[BCAJ{\etalchar{+}}22]{braverman2022power}
Vladimir Braverman, Vincent Cohen-Addad, Shaofeng H.-C. Jiang, Robert
  Krauthgamer, Chris Schwiegelshohn, Mads~Bech Toftrup, and Xuan Wu.
\newblock The power of uniform sampling for coresets.
\newblock In {\em 63rd Annual IEEE Symposium on Foundations of Computer Science
  (FOCS 2022)}, pages 462--473. IEEE, 2022.

\bibitem[BCF{\etalchar{+}}26]{bandyapadhyay2026polynomial}
Sayan Bandyapadhyay, Eden Chlamt{\'a}{\v{c}}, Zachary Friggstad, Mahya
  Jamshidian, Yury Makarychev, and Ali Vakilian.
\newblock A polynomial-time approximation for pairwise fair {$k$}-median
  clustering.
\newblock In {\em Proceedings of the 29th International Conference on
  Artificial Intelligence and Statistics}, volume 300 of {\em Proceedings of
  Machine Learning Research}, pages 4933--4941. PMLR, 2026.

\bibitem[BCFN19]{bera2019fair}
Suman~K. Bera, Deeparnab Chakrabarty, Nicolas Flores, and Maryam Negahbani.
\newblock Fair algorithms for clustering.
\newblock In {\em Advances in Neural Information Processing Systems},
  volume~32, 2019.

\bibitem[BFL{\etalchar{+}}21]{bohm2021algorithms}
Matteo B{\"o}hm, Adriano Fazzone, Stefano Leonardi, Cristina Menghini, and
  Chris Schwiegelshohn.
\newblock Algorithms for fair {$k$}-clustering with multiple protected
  attributes.
\newblock {\em Operations Research Letters}, 49(5):787--789, 2021.

\bibitem[BFS24]{DBLP:journals/jcss/BandyapadhyayFS24}
Sayan Bandyapadhyay, Fedor~V. Fomin, and Kirill Simonov.
\newblock On coresets for fair clustering in metric and {E}uclidean spaces and
  their applications.
\newblock {\em Journal of Computer and System Sciences}, 142:103506, 2024.

\bibitem[BGH{\etalchar{+}}26]{byrka2026kclustering}
Jaros{\l}aw Byrka, Yuhao Guo, Yang Hu, Shi Li, Chengzhang Wan, and Zaixuan
  Wang.
\newblock {$k$}-clustering via iterative randomized rounding.
\newblock In {\em Proceedings of the 67th IEEE Annual Symposium on Foundations
  of Computer Science (FOCS 2026)}, 2026.
\newblock To appear.

\bibitem[BJK18]{bhattacharya2018faster}
Anup Bhattacharya, Ragesh Jaiswal, and Amit Kumar.
\newblock Faster algorithms for the constrained {$k$}-means problem.
\newblock {\em Theory of Computing Systems}, 62:93--115, 2018.

\bibitem[BMDG05]{banerjee2005clustering}
Arindam Banerjee, Srujana Merugu, Inderjit~S. Dhillon, and Joydeep Ghosh.
\newblock Clustering with {B}regman divergences.
\newblock {\em Journal of Machine Learning Research}, 6:1705--1749, 2005.

\bibitem[Bor22]{borgwardt2022lp}
Steffen Borgwardt.
\newblock An {LP}-based, strongly-polynomial 2-approximation algorithm for
  sparse {W}asserstein barycenters.
\newblock {\em Operational Research}, 22(2):1511--1551, 2022.

\bibitem[BP21]{borgwardt2021computational}
Steffen Borgwardt and Stephan Patterson.
\newblock On the computational complexity of finding a sparse {W}asserstein
  barycenter.
\newblock {\em Journal of Combinatorial Optimization}, 41(3):736--761, 2021.

\bibitem[CA79]{coleman1979image}
Guy~Barrett Coleman and Harry~C. Andrews.
\newblock Image segmentation by clustering.
\newblock {\em Proceedings of the IEEE}, 67(5):773--785, 1979.

\bibitem[CAEMN22]{cohenaddad2022improved}
Vincent Cohen-Addad, Hossein Esfandiari, Vahab Mirrokni, and Shyam Narayanan.
\newblock Improved approximations for {E}uclidean {$k$}-means and {$k$}-median,
  via nested quasi-independent sets.
\newblock In {\em Proceedings of the 54th Annual ACM SIGACT Symposium on Theory
  of Computing (STOC 2022)}, pages 1621--1628. ACM, 2022.

\bibitem[CAJY{\etalchar{+}}25]{cohenaddad2025sliding}
Vincent Cohen-Addad, Shaofeng H.-C. Jiang, Qiaoyuan Yang, Yubo Zhang, and
  Samson Zhou.
\newblock Fair clustering in the sliding window model.
\newblock In {\em The Thirteenth International Conference on Learning
  Representations}, 2025.

\bibitem[CAKL22]{cohenaddad2022johnson}
Vincent Cohen-Addad, {Karthik C. S.}, and Euiwoong Lee.
\newblock {J}ohnson {C}overage {H}ypothesis: Inapproximability of {$k$}-means
  and {$k$}-median in {$\ell_p$}-metrics.
\newblock In {\em Proceedings of the 2022 Annual ACM-SIAM Symposium on Discrete
  Algorithms (SODA 2022)}, pages 1493--1530. SIAM, 2022.

\bibitem[CAKM19]{cohen2019local}
Vincent Cohen-Addad, Philip~N. Klein, and Claire Mathieu.
\newblock Local search yields approximation schemes for {$k$}-means and
  {$k$}-median in {E}uclidean and minor-free metrics.
\newblock {\em SIAM Journal on Computing}, 48(2):644--667, 2019.

\bibitem[CAKSS26]{cohenaddad2026almost}
Vincent Cohen-Addad, {Karthik C. S.}, David Saulpic, and Chris Schwiegelshohn.
\newblock Almost-optimal upper and lower bounds for clustering in low
  dimensional {E}uclidean spaces.
\newblock In {\em 42nd International Symposium on Computational Geometry (SoCG
  2026)}, volume 367 of {\em Leibniz International Proceedings in Informatics},
  pages 34:1--34:16. Schloss Dagstuhl -- Leibniz-Zentrum f{\"u}r Informatik,
  2026.

\bibitem[CCAG{\etalchar{+}}25]{charikar2025improved}
Moses Charikar, Vincent Cohen-Addad, Ruiquan Gao, Fabrizio Grandoni, Euiwoong
  Lee, and Ernest van Wijland.
\newblock An improved greedy approximation for (metric) {$k$}-means.
\newblock In {\em 2025 IEEE 66th Annual Symposium on Foundations of Computer
  Science (FOCS)}, pages 233--240. IEEE, 2025.

\bibitem[CCAG{\etalchar{+}}26]{charikar2026nonmonotone}
Moses Charikar, Vincent Cohen-Addad, Ruiquan Gao, Fabrizio Grandoni, Euiwoong
  Lee, and Ernest van Wijland.
\newblock A {$(4+\varepsilon)$}-approximation for {E}uclidean {$k$}-means via
  non-monotone dual-fitting.
\newblock In {\em Proceedings of the 58th Annual ACM Symposium on Theory of
  Computing (STOC 2026)}, pages 1881--1891. ACM, 2026.

\bibitem[CCS18]{claici2018stochastic}
Sebastian Claici, Edward Chien, and Justin Solomon.
\newblock Stochastic {W}asserstein barycenters.
\newblock In {\em Proceedings of the 35th International Conference on Machine
  Learning}, volume~80 of {\em Proceedings of Machine Learning Research}, pages
  999--1008. PMLR, 2018.

\bibitem[CD14]{cuturi2014fast}
Marco Cuturi and Arnaud Doucet.
\newblock Fast computation of {W}asserstein barycenters.
\newblock In {\em Proceedings of the 31st International Conference on Machine
  Learning}, volume~32 of {\em Proceedings of Machine Learning Research}, pages
  685--693. PMLR, 2014.

\bibitem[CDH{\etalchar{+}}20]{chzhen2020fair}
Evgenii Chzhen, Christophe Denis, Mohamed Hebiri, Luca Oneto, and Massimiliano
  Pontil.
\newblock Fair regression with {W}asserstein barycenters.
\newblock In {\em Advances in Neural Information Processing Systems},
  volume~33, pages 7321--7331, 2020.

\bibitem[CFLM19]{chen2019proportionally}
Xingyu Chen, Brandon Fain, Liang Lyu, and Kamesh Munagala.
\newblock Proportionally fair clustering.
\newblock In {\em Proceedings of the 36th International Conference on Machine
  Learning}, volume~97 of {\em Proceedings of Machine Learning Research}, pages
  1032--1041. PMLR, 2019.

\bibitem[CKLV17]{chierichetti2017fair}
Flavio Chierichetti, Ravi Kumar, Silvio Lattanzi, and Sergei Vassilvitskii.
\newblock Fair clustering through fairlets.
\newblock In {\em Advances in Neural Information Processing Systems},
  volume~30, 2017.

\bibitem[CMV22]{chlamtac2022approximating}
Eden Chlamt{\'a}{\v{c}}, Yury Makarychev, and Ali Vakilian.
\newblock Approximating fair clustering with cascaded norm objectives.
\newblock In {\em Proceedings of the 2022 Annual ACM-SIAM Symposium on Discrete
  Algorithms (SODA 2022)}, pages 2664--2683. SIAM, 2022.

\bibitem[CP16]{cuturi2016smoothed}
Marco Cuturi and Gabriel Peyr{\'e}.
\newblock A smoothed dual approach for variational {W}asserstein problems.
\newblock {\em SIAM Journal on Imaging Sciences}, 9(1):320--343, 2016.

\bibitem[CWM{\etalchar{+}}17]{chang2017deep}
Jianlong Chang, Lingfeng Wang, Gaofeng Meng, Shiming Xiang, and Chunhong Pan.
\newblock Deep adaptive image clustering.
\newblock In {\em Proceedings of the IEEE International Conference on Computer
  Vision}, pages 5879--5887, 2017.

\bibitem[DEMZ23]{dickerson2023doubly}
John Dickerson, Seyed Esmaeili, Jamie~H. Morgenstern, and Claire~Jie Zhang.
\newblock Doubly constrained fair clustering.
\newblock In {\em Advances in Neural Information Processing Systems},
  volume~36, pages 13267--13293, 2023.

\bibitem[DEMZ25]{dickerson2025sok}
John Dickerson, Seyed~A. Esmaeili, Jamie Morgenstern, and Claire~Jie Zhang.
\newblock {SoK}: Fair clustering: Critique, caveats, and future directions.
\newblock In {\em 2025 IEEE Conference on Secure and Trustworthy Machine
  Learning (SaTML)}, pages 698--713. IEEE, 2025.

\bibitem[DLDE25]{duppala2025robust}
Sharmila Duppala, Juan Luque, John~P. Dickerson, and Seyed~A. Esmaeili.
\newblock Robust fair clustering with group membership uncertainty sets.
\newblock In {\em Proceedings of the 28th International Conference on
  Artificial Intelligence and Statistics}, volume 258 of {\em Proceedings of
  Machine Learning Research}, pages 2350--2358. PMLR, 2025.

\bibitem[DMM{\etalchar{+}}19]{dognin2019wasserstein}
Pierre Dognin, Igor Melnyk, Youssef Mroueh, Jarret Ross, Cicero {Dos Santos},
  and Tom Sercu.
\newblock {W}asserstein barycenter model ensembling.
\newblock In {\em International Conference on Learning Representations}, 2019.

\bibitem[DMV22]{dai2022fair}
Zhen Dai, Yury Makarychev, and Ali Vakilian.
\newblock Fair representation clustering with several protected classes.
\newblock In {\em Proceedings of the 2022 ACM Conference on Fairness,
  Accountability, and Transparency (FAccT 2022)}, pages 814--823. ACM, 2022.

\bibitem[DX20]{ding2020unified}
Hu~Ding and Jinhui Xu.
\newblock A unified framework for clustering constrained data without locality
  property.
\newblock {\em Algorithmica}, 82(4):808--852, 2020.

\bibitem[Fei98]{feige1998threshold}
Uriel Feige.
\newblock A threshold of {$\ln n$} for approximating set cover.
\newblock {\em Journal of the ACM}, 45(4):634--652, 1998.

\bibitem[FHHS26]{funk2026constant}
Nicole Funk, Annika Hennes, Johanna Hillebrand, and Sarah Sturm.
\newblock Constant-factor approximations for doubly constrained fair
  {$k$}-center, {$k$}-median and {$k$}-means.
\newblock In {\em 20th Scandinavian Symposium on Algorithm Theory (SWAT 2026)},
  volume 370 of {\em Leibniz International Proceedings in Informatics}, pages
  19:1--19:19. Schloss Dagstuhl -- Leibniz-Zentrum f{\"u}r Informatik, 2026.

\bibitem[FRS19]{friggstad2019local}
Zachary Friggstad, Mohsen Rezapour, and Mohammad~R. Salavatipour.
\newblock Local search yields a {PTAS} for {$k$}-means in doubling metrics.
\newblock {\em SIAM Journal on Computing}, 48(2):452--480, 2019.

\bibitem[FW56]{frank1956algorithm}
Marguerite Frank and Philip Wolfe.
\newblock An algorithm for quadratic programming.
\newblock {\em Naval Research Logistics Quarterly}, 3(1--2):95--110, 1956.

\bibitem[GOR{\etalchar{+}}22]{grandoni2022refined}
Fabrizio Grandoni, Rafail Ostrovsky, Yuval Rabani, Leonard~J. Schulman, and
  Rakesh Venkat.
\newblock A refined approximation for {E}uclidean {$k$}-means.
\newblock {\em Information Processing Letters}, 176:106251, 2022.

\bibitem[GSM14]{glassman2014feature}
Elena~L. Glassman, Rishabh Singh, and Robert~C. Miller.
\newblock Feature engineering for clustering student solutions.
\newblock In {\em Proceedings of the First ACM Conference on Learning @ Scale},
  pages 171--172, 2014.

\bibitem[GSV21]{ghadiri2021socially}
Mehrdad Ghadiri, Samira Samadi, and Santosh~S. Vempala.
\newblock Socially fair {$k$}-means clustering.
\newblock In {\em Proceedings of the 2021 ACM Conference on Fairness,
  Accountability, and Transparency (FAccT 2021)}, pages 438--448. ACM, 2021.

\bibitem[HJV19]{huang2019coresets}
Lingxiao Huang, Shaofeng H.-C. Jiang, and Nisheeth~K. Vishnoi.
\newblock Coresets for clustering with fairness constraints.
\newblock In {\em Advances in Neural Information Processing Systems},
  volume~32, 2019.

\bibitem[HNY{\etalchar{+}}17]{ho2017multilevel}
Nhat Ho, XuanLong Nguyen, Mikhail Yurochkin, Hung~Hai Bui, Viet Huynh, and Dinh
  Phung.
\newblock Multilevel clustering via {W}asserstein means.
\newblock In {\em Proceedings of the 34th International Conference on Machine
  Learning}, volume~70 of {\em Proceedings of Machine Learning Research}, pages
  1501--1509. PMLR, 2017.

\bibitem[IEA{\etalchar{+}}23]{ikotun2023kmeans}
Abiodun~M. Ikotun, Absalom~E. Ezugwu, Laith Abualigah, Belal Abuhaija, and Jia
  Heming.
\newblock K-means clustering algorithms: A comprehensive review, variants
  analysis, and advances in the era of big data.
\newblock {\em Information Sciences}, 622:178--210, 2023.

\bibitem[IKI94]{inaba1994applications}
Mary Inaba, Naoki Katoh, and Hiroshi Imai.
\newblock Applications of weighted {V}oronoi diagrams and randomization to
  variance-based {$k$}-clustering.
\newblock In {\em Proceedings of the Tenth Annual Symposium on Computational
  Geometry (SoCG 1994)}, pages 332--339. ACM, 1994.

\bibitem[Jag13]{jaggi2013revisiting}
Martin Jaggi.
\newblock Revisiting {F}rank--{W}olfe: Projection-free sparse convex
  optimization.
\newblock In {\em Proceedings of the 30th International Conference on Machine
  Learning}, volume~28 of {\em Proceedings of Machine Learning Research}, pages
  427--435. PMLR, 2013.

\bibitem[Jai10]{DBLP:journals/prl/Jain10}
Anil~K. Jain.
\newblock Data clustering: 50 years beyond {$k$}-means.
\newblock {\em Pattern Recognition Letters}, 31(8):651--666, 2010.

\bibitem[Kha80]{khachiyan1980polynomial}
Leonid~G. Khachiyan.
\newblock Polynomial algorithms in linear programming.
\newblock {\em USSR Computational Mathematics and Mathematical Physics},
  20(1):53--72, 1980.

\bibitem[KLPK25]{kim2025fair}
Kunwoong Kim, Jihu Lee, Sangchul Park, and Yongdai Kim.
\newblock Fair clustering via alignment.
\newblock In {\em Proceedings of the 42nd International Conference on Machine
  Learning}, volume 267 of {\em Proceedings of Machine Learning Research},
  pages 30363--30399. PMLR, 2025.

\bibitem[KP24]{kellerhals2024proportional}
Leon Kellerhals and Jannik Peters.
\newblock Proportional fairness in clustering: A social choice perspective.
\newblock In {\em Advances in Neural Information Processing Systems},
  volume~37, pages 111299--111317, 2024.

\bibitem[LHC{\etalchar{+}}20]{lin2020fixed}
Tianyi Lin, Nhat Ho, Xi~Chen, Marco Cuturi, and Michael~I. Jordan.
\newblock Fixed-support {W}asserstein barycenters: Computational hardness and
  fast algorithm.
\newblock In {\em Advances in Neural Information Processing Systems},
  volume~33, pages 5368--5380, 2020.

\bibitem[LSW17]{lee2017improved}
Euiwoong Lee, Melanie Schmidt, and John Wright.
\newblock Improved and simplified inapproximability for {$k$}-means.
\newblock {\em Information Processing Letters}, 120:40--43, 2017.

\bibitem[Mat00]{matouvsek2000approximate}
Ji{\v{r}}{\'\i} Matou{\v{s}}ek.
\newblock On approximate geometric {$k$}-clustering.
\newblock {\em Discrete \& Computational Geometry}, 24(1):61--84, 2000.

\bibitem[MN21]{montesuma2021wasserstein}
Eduardo~Fernandes Montesuma and Fred~Maurice {Ngol\`e Mboula}.
\newblock {W}asserstein barycenter for multi-source domain adaptation.
\newblock In {\em Proceedings of the IEEE/CVF Conference on Computer Vision and
  Pattern Recognition}, pages 16785--16793, 2021.

\bibitem[MNV12]{DBLP:journals/tcs/MahajanNV12}
Meena Mahajan, Prajakta Nimbhorkar, and Kasturi~R. Varadarajan.
\newblock The planar {$k$}-means problem is {NP}-hard.
\newblock {\em Theoretical Computer Science}, 442:13--21, 2012.

\bibitem[MS20]{micha2020proportionally}
Evi Micha and Nisarg Shah.
\newblock Proportionally fair clustering revisited.
\newblock In {\em 47th International Colloquium on Automata, Languages, and
  Programming (ICALP 2020)}, volume 168 of {\em Leibniz International
  Proceedings in Informatics}, pages 85:1--85:16. Schloss Dagstuhl --
  Leibniz-Zentrum f{\"u}r Informatik, 2020.

\bibitem[NM10]{nugent2010overview}
Rebecca Nugent and Marina Meila.
\newblock An overview of clustering applied to molecular biology.
\newblock In {\em Statistical Methods in Molecular Biology}, pages 369--404.
  Humana Press, 2010.

\bibitem[PC19]{peyre2019computational}
Gabriel Peyr{\'e} and Marco Cuturi.
\newblock Computational optimal transport: With applications to data science.
\newblock {\em Foundations and Trends in Machine Learning}, 11(5--6):355--607,
  2019.

\bibitem[Raz98]{raz1998parallel}
Ran Raz.
\newblock A parallel repetition theorem.
\newblock {\em SIAM Journal on Computing}, 27(3):763--803, 1998.

\bibitem[RPDB12]{rabin2012wasserstein}
Julien Rabin, Gabriel Peyr{\'e}, Julie Delon, and Marc Bernot.
\newblock {W}asserstein barycenter and its application to texture mixing.
\newblock In {\em Scale Space and Variational Methods in Computer Vision (SSVM
  2011)}, pages 435--446. Springer, 2012.

\bibitem[RQN16]{ronan2016avoiding}
Tom Ronan, Zhijie Qi, and Kristen~M. Naegle.
\newblock Avoiding common pitfalls when clustering biological data.
\newblock {\em Science Signaling}, 9(432):re6, 2016.

\bibitem[Sch86]{schrijver1986theory}
Alexander Schrijver.
\newblock {\em Theory of Linear and Integer Programming}.
\newblock John Wiley \& Sons, 1986.

\bibitem[SJ20]{singh2020model}
Sidak~Pal Singh and Martin Jaggi.
\newblock Model fusion via optimal transport.
\newblock In {\em Advances in Neural Information Processing Systems},
  volume~33, pages 22045--22055, 2020.

\bibitem[SLD18]{srivastava2018scalable}
Sanvesh Srivastava, Cheng Li, and David~B. Dunson.
\newblock Scalable {B}ayes via barycenter in {W}asserstein space.
\newblock {\em Journal of Machine Learning Research}, 19(8):1--35, 2018.

\bibitem[SMD25]{song2025relax}
Shihong Song, Guanlin Mo, and Hu~Ding.
\newblock Relax and merge: A simple yet effective framework for solving fair
  {$k$}-means and {$k$}-sparse {W}asserstein barycenter problems.
\newblock In {\em The Thirteenth International Conference on Learning
  Representations}, 2025.

\bibitem[SSS20]{schmidt2020fair}
Melanie Schmidt, Chris Schwiegelshohn, and Christian Sohler.
\newblock Fair coresets and streaming algorithms for fair {$k$}-means.
\newblock In {\em Approximation and Online Algorithms (WAOA 2019)}, volume
  11926 of {\em Lecture Notes in Computer Science}, pages 232--251. Springer,
  2020.

\bibitem[The26]{thejaswi2026hardness}
Suhas Thejaswi.
\newblock On the hardness of approximation of the fair {$k$}-center problem.
\newblock {\em Information and Computation}, 312:105503, 2026.

\bibitem[VY22]{vakilian2022improved}
Ali Vakilian and Mustafa Yal{\c{c}}{\i}ner.
\newblock Improved approximation algorithms for individually fair clustering.
\newblock In {\em Proceedings of the 25th International Conference on
  Artificial Intelligence and Statistics}, volume 151 of {\em Proceedings of
  Machine Learning Research}, pages 8758--8779. PMLR, 2022.

\bibitem[WFW24]{wu2024approximation}
Di~Wu, Qilong Feng, and Jianxin Wang.
\newblock Approximation algorithms for fair {$k$}-median problem without
  fairness violation.
\newblock {\em Theoretical Computer Science}, 985:114332, 2024.

\bibitem[WS11]{wang2011ckmeans}
Haizhou Wang and Mingzhou Song.
\newblock {Ckmeans.1d.dp}: optimal {$k$}-means clustering in one dimension by
  dynamic programming.
\newblock {\em The R Journal}, 3(2):29--33, 2011.

\bibitem[WX26a]{wang2025closing}
Jiaqi Wang and Weijun Xie.
\newblock Closing the gap: Efficient algorithms for discrete {W}asserstein
  barycenters.
\newblock In {\em Integer Programming and Combinatorial Optimization (IPCO
  2026)}, Lecture Notes in Computer Science, pages 507--520. Springer, 2026.

\bibitem[WX26b]{wang2026closingfull}
Jiaqi Wang and Weijun Xie.
\newblock Closing the gap: Efficient algorithms for discrete {W}asserstein
  barycenters.
\newblock arXiv:2511.04607v2, 2026.
\newblock Full version.

\bibitem[YD24]{yang2024approximate}
Qingyuan Yang and Hu~Ding.
\newblock Approximate algorithms for {$k$}-sparse {W}asserstein barycenter with
  outliers.
\newblock In {\em Proceedings of the Thirty-Third International Joint
  Conference on Artificial Intelligence (IJCAI 2024)}, pages 5316--5325, 2024.

\bibitem[YWLW23]{yuan2023index}
Zhe Yuan, Zhewei Wei, Fangrui Lv, and Ji-Rong Wen.
\newblock Index-free triangle-based graph local clustering.
\newblock {\em Frontiers of Computer Science}, 18(3):183404, 2023.

\bibitem[YWWL17]{ye2017fast}
Jianbo Ye, Panruo Wu, James~Z. Wang, and Jia Li.
\newblock Fast discrete distribution clustering using {W}asserstein barycenter
  with sparse support.
\newblock {\em IEEE Transactions on Signal Processing}, 65(9):2317--2332, 2017.

\bibitem[ZFT{\etalchar{+}}23]{Zhang2023}
Jing Zhang, Ruidong Fan, Hong Tao, Jiacheng Jiang, and Chenping Hou.
\newblock Constrained clustering with weak label prior.
\newblock {\em Frontiers of Computer Science}, 18, 2023.

\end{thebibliography}

\end{document}